\documentclass[10pt]{article}
\usepackage{palatcm, color, verbatim}

\usepackage{amssymb}
\usepackage{amsmath}
\usepackage{epsfig}
\usepackage[usenames,dvipsnames]{xcolor}
\usepackage{mathdots}

\begin{document}
\newtheorem{theorem}{Theorem}
\newtheorem{corollary}{Corollary}
\newtheorem{conjecture}{Conjecture}
\newtheorem{definition}{Definition}
\newtheorem{lemma}{Lemma}
\newtheorem{remark}{Remark}
\newcommand{\define}{\stackrel{\triangle}{=}}
\newcommand{\xS}{\mathbf{S}}
\newtheorem{construction}{Construction}

\newcommand\blfootnote[1]{%
  \begingroup
  \renewcommand\thefootnote{}\footnote{#1}%
  \addtocounter{footnote}{-1}%
  \endgroup
}

\pagestyle{empty}

\def\QED{\mbox{\rule[0pt]{1.5ex}{1.5ex}}}
\def\proof{\noindent{\it Proof: }}

\date{}

\title{Index Coding -- An Interference Alignment Perspective}

\author{Hamed Maleki, Viveck R. Cadambe, Syed A. Jafar\\
\blfootnote{Presented in part at ISIT 2012. Hamed Maleki and Syed Jafar (email: hmaleki@uci.edu, syed@uci.edu) are with the Center for Pervasive Communications and Computing (CPCC) at the University of California Irvine, Irvine, CA, 92697. Viveck Cadambe (email: viveck@mit.edu) is with MIT. This work is supported in part by ONR N00014-12-1-0067 and by NSF CCF-1143982.}}


\maketitle
\thispagestyle{empty}
\begin{abstract}
The index coding problem is studied from an interference alignment perspective, providing new results as well as new insights into, and generalizations of, previously known results. An equivalence is established between multiple unicast index coding where each message is desired by exactly one receiver, and multiple groupcast index coding where a message can be desired by multiple receivers, which settles the heretofore open question of insufficiency of linear codes for the multiple unicast index coding problem by equivalence with multiple groupcast settings where this question has previously been answered. Necessary and sufficient conditions for the achievability of rate half per message are shown to be a natural consequence of interference alignment constraints, and generalizations to feasibility of rate $\frac{1}{L+1}$ per message when each destination desires at least $L$ messages, are similarly obtained. Finally, capacity optimal solutions are presented to a series of symmetric index coding problems inspired by the local connectivity and local interference characteristics of wireless networks. The solutions are based on vector linear coding.
\end{abstract}

\newpage


\section{Introduction}

Much progress in network information theory can be attributed to the pursuit of the capacity of simple-to-describe canonical network communication models. Simplicity in the network communication models often affords a clear formulation of techniques involved in the communication system. The focus of this paper is the \emph{index coding} problem which  is arguably the simplest multiuser capacity problem  because it is a communication network that has \emph{only one link with finite capacity}. Yet, this turns out to be the proverbial case where appearances can be quite deceiving. More than a decade after it was introduced by Birk and Kol in \cite{Birk_Kol,Birk_Kol_Trans}, the index coding problem not only remains open, but also has been shown to include as special cases a number of difficult problems in both wired and wireless settings --- such as the general  multiple  unicast problem with linear network coding \cite{Rouayheb_Sprintson_Georghiades}, multi-way relay networks \cite{Yazdi_Savari_Kramer}, and the blind cellular interference alignment problem in wireless networks \cite{CBIA}, to name a few. Remarkably, the index coding problem is also the origin of the fundamental idea of interference alignment \cite{Birk_Kol_Trans}, which was re-discovered, extensively studied and developed in a variety of forms in wireless networks \cite{Jafar_FnT,Cadambe_Jafar_int,Cadambe_Jafar_X} and has recently found applications in network coding problems such as the distributed data storage exact repair problem \cite{Cadambe_Jafar_Maleki_Ramchandran_Suh, Dimakis_survey} and the 3 unicast problem \cite{Das_Vishwanath_Jafar_Markopoulou_ISIT, Meng_Ramakrishnan_Markopoulou_Jafar}. In this paper, we attempt to bring this idea ``home", by applying the understanding of the principles of interference alignment, into the original setting --- the index coding problem.

The essence of the index coding problem lies in its focus on a single bottleneck network. Having only one link with  finite capacity concentrates the challenge  of network coding in one place, highlighting some of the most fundamental, challenging, and surprising aspects of the network coding problem. Understanding the role of a single bottleneck edge in a network when the rest of the network is composed only of trivial links (of  infinite capacity), is a natural stepping stone toward a broader understanding of communication networks\footnote{The edge-removal problem introduced in \cite{Ho_Effros_Jalali_Allerton10, Jalali_Effros_Ho_ITA11} is another intriguing open problem that seeks to understand the role of a single edge in a network, and has been shown to be closely related to the general question of whether zero-error capacity and $\epsilon$-error capacity are the same for general network coding instances\cite{Langberg_Effros_Allerton11}.}. We start with a discussion of similarly motivated single-bottleneck  settings for both wired and wireless networks.

\subsection{Single Bottleneck Wired Networks -- Index Coding}

\begin{figure}[h]\centering
\includegraphics[width=5in]{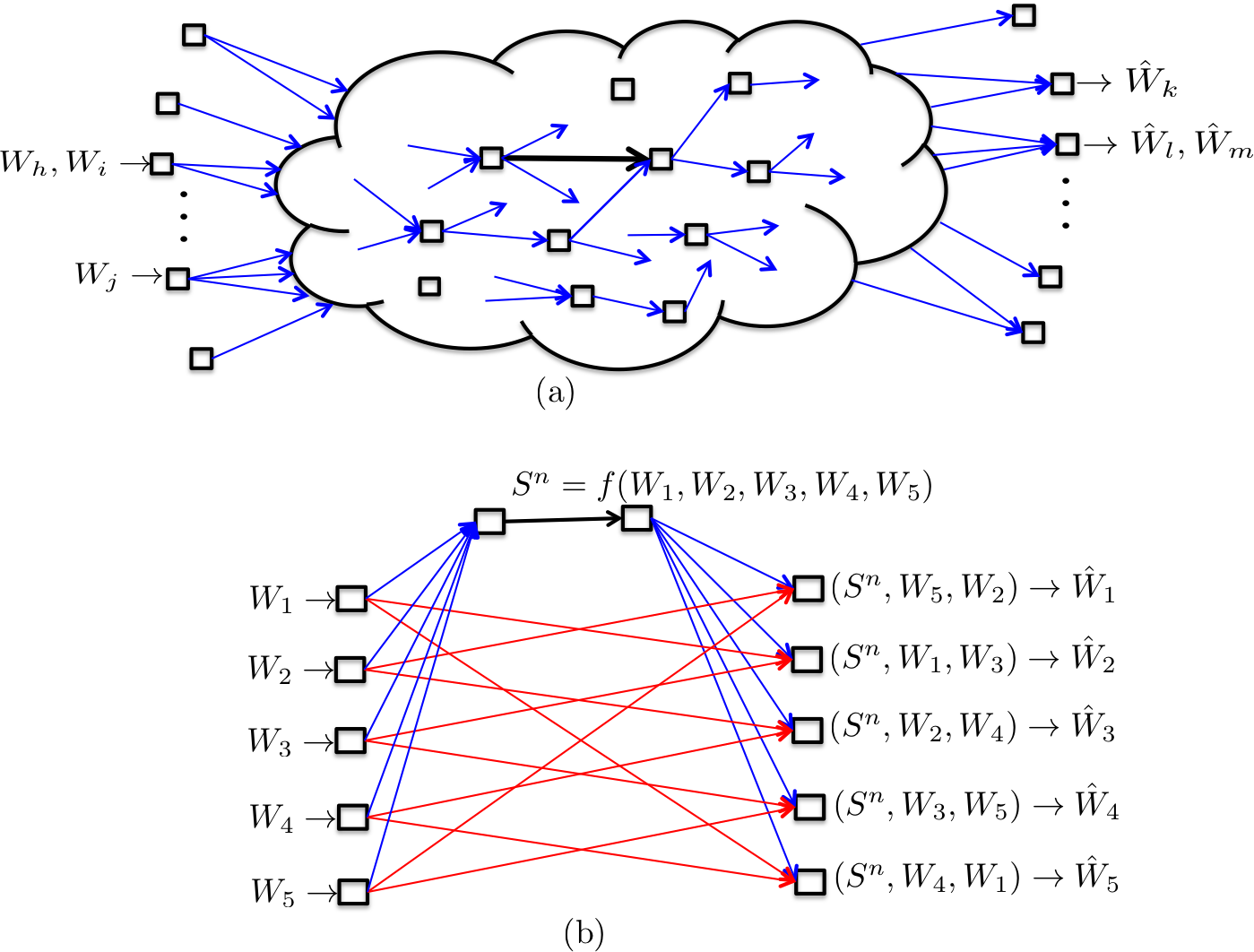}
\caption{\small (a) General network coding problem: If only one link (shown in black) in the intermediate network has finite (unit) capacity, and all the other links have infinite capacity, then the remaining problem is the index coding problem. (b) Example of an index coding setting.}\label{fig:generalnet}
\end{figure}

Consider a general network coding setting shown in Figure \ref{fig:generalnet}(a) where the source nodes on the left communicate with the destination nodes on the right through a network of intermediate nodes connected via orthogonal, noiseless, capacitated links. If only one link (shown in black) in the intermediate network has finite (unit) capacity, and all the other links have infinite capacity, then the remaining problem  is the index coding problem. Clearly, the only non-trivial message flows are those for which every path between the source and the desired destination(s) \emph{must} pass through the finite capacity link. All other messages have either rate zero or infinity, and can be eliminated. The remaining network contains three kinds of infinite capacity links in addition to the finite capacity link. First, the remaining source nodes must connect to the transmitter of the finite capacity link via infinite capacity links. Second, the destination nodes must connect to the receiver of the finite capacity link through infinite capacity links. Third, if there are infinite capacity paths between any sources and their non-desired destinations, those paths are replaced with infinite capacity links known as ``antidote" links. An example is shown in Figure \ref{fig:generalnet}(b) where 5 messages originate at the sources on the left and  are desired by the destination nodes on the right.  The bottleneck link (finite capacity link) is shown at the top of the figure in black and  carries a sequence of symbols from a finite alphabet $S^n\in\mathcal{S}^n$, which are chosen with full knowledge of all messages. The antidote links are shown in red. Clearly, the best use of the antidote links, which have infinite capacity, is to convey all the information, i.e., the messages, from the transmitters to the receivers of the antidote links. These antidotes comprise the side information that makes the problem interesting and quite challenging in general. Each destination must be able to decode its desired message based on the sequence of symbols sent over the bottleneck link and the set of undesired messages available to it as antidotes.

Index coding can also be seen as ``source coding with side information", or as a broadcast channel with cognitive receivers, i.e., where certain receivers have full knowledge of certain messages a-priori. 

\subsection{Single Bottleneck Wireless Networks -- Wireless Index Coding}
\begin{figure}[!h]\centering
\includegraphics[width=5in]{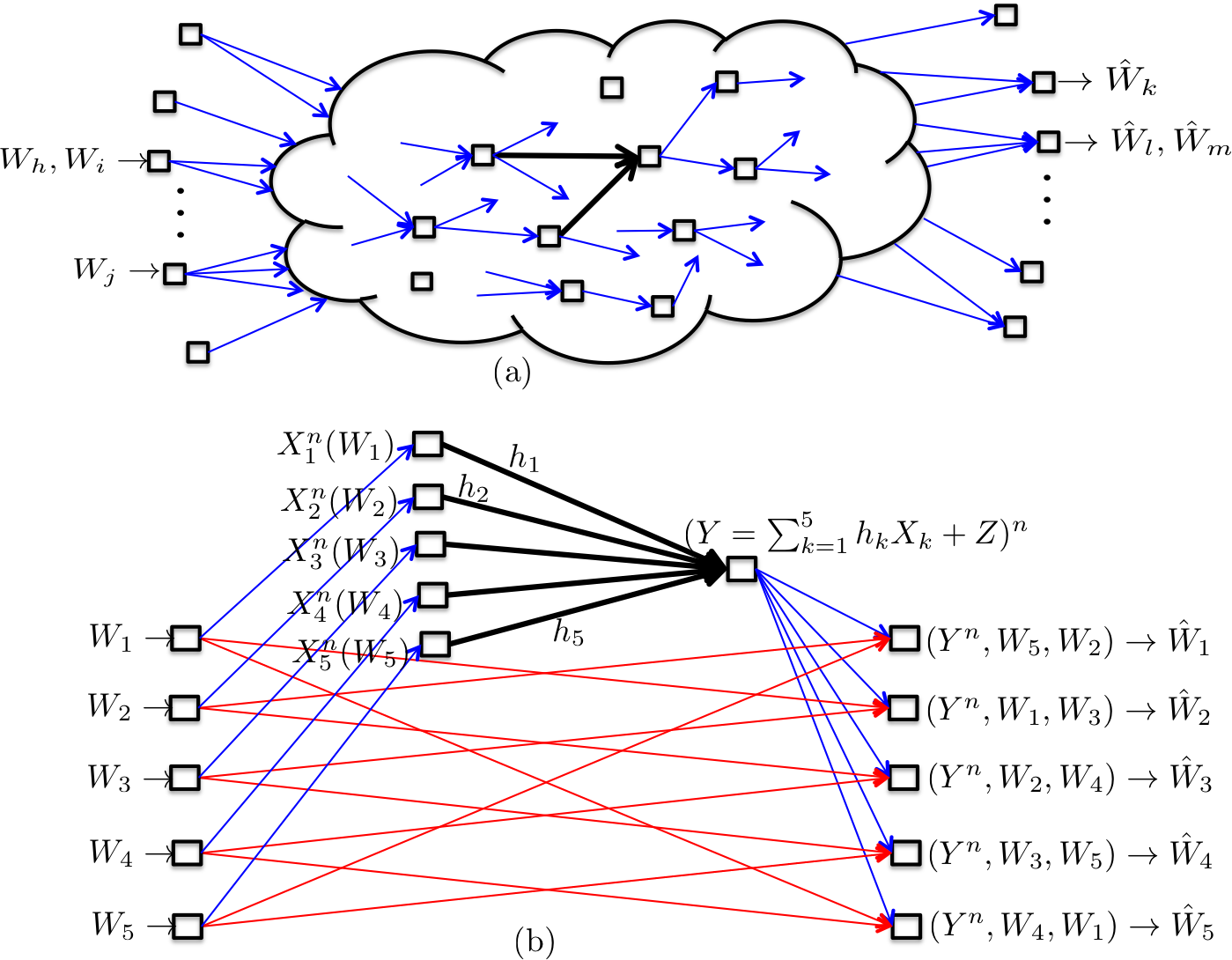}
\caption{\small (a) Wireless network: If only one receiver (shown with incoming signals in black) in the intermediate network has non-zero (unit) AWGN variance, and all the other receivers have zero noise (infinite capacity), then the remaining problem is the wireless index coding problem. (b) Example of a wireless index coding setting.}\label{fig:wirelessindexcoding}
\end{figure}
Consider a wireless network shown in Figure \ref{fig:wirelessindexcoding}(a)  comprised of the source nodes shown on the left, which communicate with destination nodes shown on the right, through an intermediate network of relay nodes. Depending on propagation path loss different pairs of nodes may be connected or disconnected. Because this is a wireless setting, signals emerging from the same transmitter are broadcast, and signals arriving at the same receiver interfere. All transmitters are subject to power constraint $P$, and generally the receivers experience additive white Gaussian noise (AWGN) in addition to the superposition of fading signals from connected transmitters. As an analogue to the index coding problem defined by a single bottleneck link, let us assume only one of the receivers in the intermediate network experiences AWGN, e.g., of unit variance, while all other receivers experience no noise, i.e., have infinite resolution of the complex valued signals, essentially providing them infinite capacity links to their respective connected transmitters. Eliminating messages that have infinite capacity paths between their sources and all their desired destinations, what remains is the wireless index coding problem, introduced in \cite{CBIA}.  While, depending on the wireless network topology, the resulting wireless index coding problem can in general be quite  involved, e.g., if the wireless network contains cycles that provide feedback from the output of the bottleneck receiver to the distributed or partially cooperating nodes transmitting to the bottleneck receiver, Fig. \ref{fig:wirelessindexcoding}(b) shows a relatively simple example of the wireless index coding problem that corresponds to the index coding problem of Fig. \ref{fig:generalnet}(b), in the sense that the capacity  of the index coding problem maps directly to the degrees of freedom (DoF)  of the wireless index coding problem. The capacity per message of the index coding problem in Figure \ref{fig:generalnet}(b) is 2/5, as is the DoF value per message for the wireless index coding problem in Fig. \ref{fig:wirelessindexcoding}(b), and in both cases the ``unit" for measurement is the capacity/DoF of the bottleneck link/receiver. The index coding problem normalizes the bottleneck link capacity to unity, so that all rates are measured as multiples of the bottleneck link capacity, and the wireless index coding problem normalizes the number of signal dimensions (DoF) available to the bottleneck receiver to unity, and all DoF are measured as multiples of the bottleneck DoF. As explained in \cite{CBIA}, the relationship between the index coding problem and the wireless index coding problem goes much further, and much more can be said about their similarities and differences. For instance, if full cooperation is allowed between all sources directly transmitting to the bottleneck receiver in a wireless index coding problem, the DoF of the resulting network is the same as the capacity of the corresponding index coding problem (in their respective units). The DoF of the  wireless index coding problem are, in general, bounded above by the capacity of the index coding problem. It also highlights the main difference between the index coding problem and the wireless index coding problem --- all sources are necessarily allowed to fully cooperate in the former because the bottleneck transmitter has full knowledge of all messages, but not necessarily in the latter (depending on the topology of the original network in Figure \ref{fig:wirelessindexcoding}(a)). However, if the index coding problem has a capacity optimal vector linear coding solution that can be translated to the complex field, then the same solution may be applied in the wireless index coding problem as well. This is because vector linear solutions are comprised of a superposition of separately encoded messages, and a superposition over complex field is naturally provided by the wireless medium \cite{CBIA}. Somewhat surprisingly, this is a very common situation, e.g., all the instances of the index coding problems studied in this paper have capacity optimal vector linear coding solutions that translate to the complex field, thereby simultaneously providing the DoF characterization for the corresponding wireless index coding problem.

In the wireless index coding problem discussed above,  the bottleneck is concentrated at one receiver, lending the bottleneck a multiple access character. Another formulation of the wireless index coding problem is also conceivable where the bottleneck may be concentrated at one transmitter, e.g., all receivers experience additive noise and there is only one transmitter with finite power (all other transmitters have infinite power), which would lend the bottleneck a broadcast character, and which could be a similarly interesting and promising research avenue. 

The motivation for studying single bottleneck networks in both wired and wireless settings is evident from an information-theoretical perspective as a stepping stone to a broader understanding of communication networks. What is surprising is that the index coding problem, in spite of its simple formulation, not only already captures much of the complexity of the full-fledged network capacity problem, but also contains a class of problems known as cellular blind interference alignment \cite{CBIA} problems (CBIA) which are of immediate practical interest for cellular wireless networks. Indeed, it is the CBIA setting that  motivates most of the instances of index coding that we solve in this work.

\subsection{Cellular Blind Interference Alignment Problem}
Consider, as an example, the cellular downlink setting shown in Figure \ref{fig:5cellpic}(a) comprised of 5 partially overlapping cells depicted as circles, with each circle containing a transmitter (base station) near its center, shown as a black square, and 3 receivers (users), shown as white squares. Propagation path loss is modeled by the assumption that each base station transmitter (black square) is only heard within the circular region defining its own cell. This gives rise to the connectivity pattern where each transmitter can be heard by three receivers and each receiver can hear three transmitters. The resulting wireless network connectivity is shown in Figure \ref{fig:5cellpic}(b) where the links show the non-zero channel coefficients. The  knowledge of non-zero channel coefficient values, which are assumed to be drawn from identical distributions,  is not available to the (blind) transmitters. Depending upon the message sets, e.g., whether each transmitter sends a message to only one user in its cell, or whether each transmitter sends 3 independent messages to the 3 users in its cell, we have the partially connected interference channel or X channel setting, respectively. Since both settings will benefit significantly from interference alignment and no knowledge of non-zero channel coefficient values is assumed, this is known as the cellular blind interference alignment  problem \cite{CBIA}.

\begin{figure}[!h]\centering
\includegraphics[width=6in]{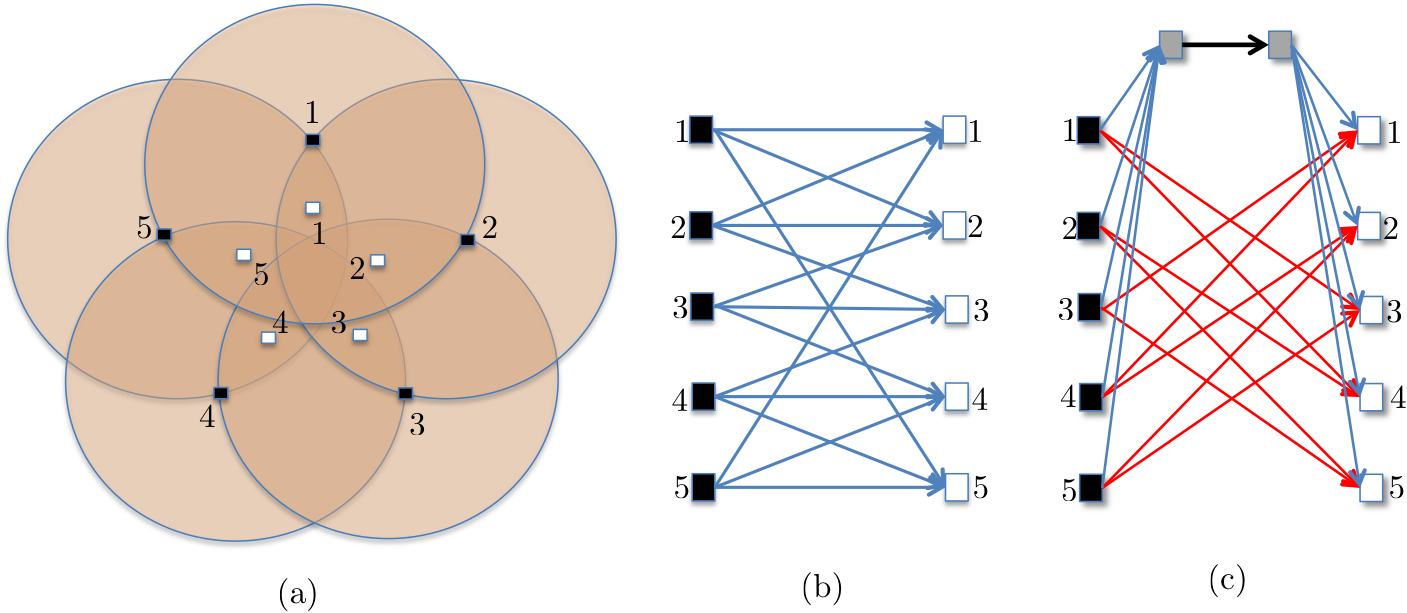}
\caption{\small Cellular blind interference alignment problem: (a) Cellular layout governing the connectivity pattern, (b) Locally connected network representation, (c) Corresponding index coding problem}\label{fig:5cellpic}
\end{figure}

The key to the CBIA problem lies in its close relationship to a corresponding index coding problem. For the CBIA problem of Figure \ref{fig:5cellpic}(a) and Figure \ref{fig:5cellpic}(b), the corresponding index coding problem is shown in Figure \ref{fig:5cellpic}(c). In the corresponding index coding problem, the graph connecting black (transmitter) and white (receiver) nodes in Figure \ref{fig:5cellpic}(c) is complementary to that in Figure \ref{fig:5cellpic}(b), i.e., a black and white node pair is connected in the index coding problem of Figure \ref{fig:5cellpic}(c) by an antidote link if and only if it is not connected in the  locally connected network representation of Figure \ref{fig:5cellpic}(b). Thus, antidote links which allow a receiver to subtract from its received signal the contribution from the corresponding undesired messages, especially  in a superposition based coding scheme as must be the case in the CBIA problem due to the distributed nature of the source nodes and the additive nature of the wireless medium, play the same role as a zero channel coefficient value in the local connectivity pattern. The main difference in the index coding problem is that full cooperation between sources is allowed, which makes the index coding capacity an outer bound on the DoF of the CBIA setting. However, since the optimal solution to the index coding problem is often based on  vector linear coding and robust to the choice of the underlying field as the real or complex field, as we will see, e.g., for the problem in Fig. \ref{fig:5cellpic} as well as for several other index coding problems, the optimal solution of the index coding problem will automatically provide the optimal  solution to the CBIA problem as well. 

Finally, we note that the CBIA problem has a natural counterpart in the network coding setting, which is a \emph{blind linear network coding} problem. Consider, for example, a wired network where the intermediate nodes perform random linear network coding, creating a linear channel matrix comprised of  polynomials in the network coding coefficients of the intermediate nodes. Recent work \cite{Das_Vishwanath_Jafar_Markopoulou_ISIT, Ramakrishnan_Das_Maleki_Markopoulou_Jafar_Vishwanath, Meng_Ramakrishnan_Markopoulou_Jafar, Han_Wang_Shroff_Allerton11} has investigated how to exploit the knowledge of these effective channel matrices at the source nodes to achieve interference alignment by linear precoding, in a manner that mimics the wireless interference channel. However, suppose that the source nodes only know the end-to-end connectivity but do not know the channel coefficient values, e.g., because they do not keep track of all the network coding coefficients. Aside from the significant distinction of working over finite fields, the resulting \emph{blind network coding} problem is virtually identical to the CBIA problem, and is similarly related to a corresponding index coding problem. For instance, if we take the graph of Figure \ref{fig:5cellpic}(b) to represent the resulting connectivity of a blind network coding problem, with no knowledge of the channel realizations at the transmitters, then the corresponding index coding problem is shown in Figure \ref{fig:5cellpic}(c). Solving this index coding problem will solve both the corresponding CBIA problem as well as the corresponding blind network coding problem.

The preceding discussion of closely related problems sheds light on the significance of the index coding problem. In spite of having only one finite capacity link, the richness  of the index coding problem is  evident, and its complexity is further underscored in the result by Rouayheb et al. in \cite{Rouayheb_Sprintson_Georghiades} where an equivalence is established between the general network coding problem restricted to linear codes, and the index coding problem. On the other hand, the relative simplicity of the index coding problem setting does make this setting more tractable. As an example, we note that the equivalence of $\epsilon$-error capacity and zero-error capacity, which remains open in the general network coding problem, has been established for the index coding problem \cite{Langberg_Effros_Allerton11}.

\section{Index Coding -- Problem Formulation} \label{sec:Problem Formulation}
The index coding problem consists of a set of $M$ independent messages $$\mathcal{W}=\{W_1, W_2, \ldots, W_M\},$$ and a set of $K$ destination nodes $$\mathcal{D}=\{D_1, D_2, \cdots, D_K\},$$ with the $k^{th}$ destination node $D_k$ identified as 
$$D_k=(\mathcal{W}_k, \mathcal{A}_k)$$ where $\mathcal{W}_k\subseteq \mathcal{W}$ is the set of messages desired by $D_k$,  the  set $\mathcal{A}_k\subset\mathcal{W}$ is comprised of the messages  available to destination $D_k$ as side information (antidotes), and $\mathcal{W}_k\cap\mathcal{A}_k=\phi$, i.e., a destination node does not desire a message that is already available to it.

An $(\mathcal{S}, n, \mathcal{R})$ index coding scheme corresponds to the choice of a finite alphabet $\mathcal{S}$ of cardinality $|\mathcal{S}|> 1$,  a coding  function, $f$, and a decoding function $g_{k,i}$, for each desired message $W_i$ at each destination $D_k$. The coding function $f$ maps all the messages to the sequence of transmitted symbols
$$ f(W_1, W_2, \cdots, W_M)={ S}^n$$
where  ${S}^n\in{\mathcal{S}^n}$ is the sequence of symbols transmitted over $n$ channel uses. Here, $\forall m\in\{1,2,\cdots, M\}$, message $W_m$ is a random variable uniformly distributed over the set
$$W_m\in\{1, 2, \cdots, |\mathcal{S}|^{nR_m}\},$$
 and  $\mathcal{R}$ is simply the rate vector
$$\mathcal{R}=(R_1, R_2, \cdots, R_M).$$
At each destination, $D_k$, there is a  decoding function for each desired message
$$g_{k,i}({S}^n,\mathcal{A}_k)=\hat{W}_{k,i}, ~~~\forall i \mbox{ such that } W_i\in\mathcal{W}_k.$$
The decoding is said to be in error if any desired message is decoded incorrectly. The probability of error  is 
$$P_e=1-\mbox{Prob}[\hat{W}_{k,i}=W_i, ~~~\forall i,k \mbox{ such that } W_i\in\mathcal{W}_k].$$

A rate tuple $\mathcal{R}=(R_1, R_2, \ldots, R_M)$ is said to be \emph{achievable} if for every $\epsilon,\delta>0$ there exists a $(\mathcal{S}, n, (\overline{R}_1, \overline{R}_2, \cdots, \overline{R}_M))$ coding scheme, for some $\mathcal{S}, n$, such that $\forall m\in\{1,2,\ldots,M\}$, $\overline{R}_m\geq R_m-\delta$, and the probability of error $P_e\leq\epsilon$. The capacity region of the index coding problem is defined as the set of all achievable rate tuples $(R_1, R_2,\ldots, R_M)$ and is denoted by $\mathcal{C}$. 

As an example of the notation, in Figure \ref{fig:generalnet}(b), we have $M=K=5$, $\mathcal{W}=\{W_1, W_2, \cdots, W_5\}$, $D_1=(\{W_1\}, \{W_5, W_2\})$, $D_2=(\{W_2\}, \{W_1, W_4\}), D_3=(\{W_3\}, \{W_2,W_4\})$, $D_4=(\{W_4\},\{W_3,W_5\})$, $D_5=(\{W_5\}, \{W_4,W_1\})$.

 The definition of capacity region presented above is in the classical sense of asymptotically vanishing probability of error, also known as $\epsilon$-error capacity, which can in general be larger than the zero-error capacity where only achievable schemes with $P_e=0$  are allowed. Even for networks comprised of noise-less links, e.g., in the  network coding problem, the equivalence of the two is not known. Remarkably, for index coding, it has been shown by Langberg and Effros in \cite{Langberg_Effros_Allerton11} that the $\epsilon$-error capacity is the same as the zero-error capacity.

It is noteworthy that the  choice  of the alphabet, $\mathcal{S}$, is inconsequential for the capacity region, e.g., one could restrict $\mathcal{S}=\{0,1\}$ without affecting the capacity as defined above, or interpret multiple channel uses as a single channel use over a larger alphabet, again without impacting capacity. Note, in particular, that the rates are measured in base-$\mathcal{S}$ units, and the capacity of the bottleneck link is automatically normalized to one unit.  The choice of alphabet is important to distinguish between linear and non-linear coding schemes, or scalar and vector coding, as will be explained later in this section.

We refer to the general index coding problem statement presented above  as the {\bf multiple groupcast} setting, where each message may be desired by multiple destination nodes. This general term includes within its scope both the multiple unicast setting where each message is desired by exactly one destination node, and the multiple multicast setting where each message is desired by all destination nodes. 

\bigskip

\noindent{\it Notation:} For any subset of messages $\overline{\mathcal{W}}\subset\mathcal{W}$ we define the compact notation,
$$\overline{\mathcal{W}}^c\define\mathcal{W}-\overline{\mathcal{W}}$$ as the set of messages in $\mathcal{W}$ that are not in $\overline{\mathcal{W}}$.
Further, we use the compact notation $W_{i,j}=\{W_i, W_j\}, W_{i,j,k}=\{W_i,W_j,W_k\}$, etc.
Also we define the compact notation, $R_{i_{1:L}}=R_{i_1}+R_{i_2}+\ldots+R_{i_L}$, $\mathcal{K}=\{1,2,\ldots,K\}$ and $\mathcal{M}=\{1,2,\ldots,M\}$

A few important classes of the index coding problem are formalized next.

\begin{definition}{---\bf Multiple Unicast Index Coding}\\
The index coding problem is called a multiple unicast index coding problem if and only if 
\begin{eqnarray}
\forall k_1,k_2\in\mathcal{K}, k_1\neq k_2, &&\mathcal{W}_{k_1}\cap\mathcal{W}_{k_2}=\phi
\end{eqnarray}
\end{definition}
In other words, a multiple unicast index coding problem is one where no message is desired by more than one destination.

\begin{definition}{---\bf Scalar Index Coding Scheme}\\ 
An  $(\mathcal{S}, n, \mathcal{R})$ index coding scheme is called a scalar coding scheme if and only if 
\begin{eqnarray}
\mathcal{R}&=&\left(\frac{1}{n}, \frac{1}{n}, \cdots, \frac{1}{n}\right)
\end{eqnarray}
\end{definition}
In other words, a scalar index coding scheme sends one symbol for each message over $n$ channel uses. 
\begin{definition}{--- \bf Linear Index Coding Scheme}\\
\noindent A linear $(\mathcal{S}, n, \mathcal{R})$   index coding scheme, achieving the rate vector $\mathcal{R}=\left(\frac{L_1}{n}, \frac{L_2}{n}, \cdots, \frac{L_M}{n}\right)$ over $n$ channel uses, corresponds to a choice of 
\begin{enumerate}
\item a finite field $\mathbb{F}=\mathcal{S}$ as the alphabet
\item ${\bf V}_m\in\mathbb{F}^{n\times L_m}, \forall m\in\mathcal{M}$ as precoding matrices
\item   ${\bf U}_{m,k}\in\mathbb{F}^{L_m\times n}, \forall m\in\mathcal{M}, \forall k$ such that $W_m\in\mathcal{W}_k$, as receiver combining matrices
\end{enumerate}
such that the following properties are satisfied
\begin{eqnarray}
\mbox{Property 1: }&& {\bf U}_{m,k}{\bf V}_{i}=0, ~~~~\forall m, i \in \mathcal{M},  k \in \mathcal{K} \mbox{ such that } m\neq i, W_{m}\in\mathcal{W}_k, W_{i}\notin\mathcal{A}_k \nonumber\\
\mbox{Property 2: }&& \mbox{det}\left({\bf U}_{m,k}{\bf V}_{m}\right)\neq 0, ~~~~\forall m \in \mathcal{M},  k \in \mathcal{K} \mbox{ such that } W_m\in\mathcal{W}_k\nonumber
\end{eqnarray}
where all operations are over $\mathbb{F}$.
\end{definition}
The  transmitted symbol sequence $S^n\in\mathbb{F}^{n\times 1}$ in a linear index coding scheme is 
\begin{eqnarray}
S^n = \sum_{m=1}^{M} \mathbf{V}_{m} \mathbf{X}_{m}
\end{eqnarray}
where $\mathbf{X}_i=(x_{m,1}, x_{m,2}, \cdots, x_{m,L_m})^T \in \mathbb{F}^{L_m\times 1}$ is an $L_m \times 1$ vector representing $W_m$. In other words, message  $W_m$ is split into $L_m$ independent scalar streams, each of which carries one symbol from $\mathbb{F}$, and is transmitted along the corresponding column vectors (the ``beamforming" vectors) of the precoding matrix for $\mathbf{V}_m$. The decoding operation for message $W_m\in\mathcal{W}_k$, desired at destination $D_k$, is 
\begin{eqnarray}
\hat{\mathbf{X}}_m&=&\left({\bf U}_{m,k}{\bf V}_{m}\right)^{-1}\mathbf{U}_{m,k}\left(S^n-\sum_{W_i\in\mathcal{A}_k}\mathbf{V}_{i} \mathbf{X}_{i}\right)
\end{eqnarray}
Thus, first the contribution from undesired messages available as antidotes, $W_i\in\mathcal{A}_k$,  is eliminated from $S^n$,  then the remaining undesired symbols are zero-forced by Property 1, and finally the desired symbols $\mathbf{X}_m$ are recovered by the invertibility of ${\bf U}_{m,k}{\bf V}_{m}$, which is guaranteed by Property 2.

 Note that a linear encoding scheme, as explained above, is a zero-error encoding scheme. The linear coding scheme described above is also called \emph{vector} linear coding schemes. This includes the special case where $L_i=1, \forall i\in\{1,2,\cdots, M\}$, which is called  a \emph{scalar} linear encoding scheme. 

\bigskip

\section{Index Coding as an Interference Alignment Problem}

In wireless communications, interference is a natural phenomenon because of the broadcast nature of the medium. In wireline network communications, interference occurs because of multiple data streams contending for a common link. Indeed this scenario is best exemplified by the index coding problem where a single bottleneck link is shared by all data streams. In wireless systems, interference alignment provides surprising gains by exploiting the inherent diversity in distributed linear systems - i.e., the notion that every receiver sees a different alignment of signal dimensions, and therefore signals can be designed to align at one receiver and stay separable at a different receiver. In the index coding problem too, we exploit the inherent diversity that exists among the receivers because the set of antidotes is different for different receivers. This diversity ensures that even if two signals align along the botteleneck link, they can be decoded at the desired receivers if they have the appropriate set of antidotes. This alignment frees up the available dimensions (on the botteleneck link) for other messages and hence makes the system more efficient.

To better understand the role of interference alignment in the index coding problem, let us examine a linear index coding achievable scheme. For simplicity of exposition, let us consider a symmetric rate setting, i.e., $L_m=L, \forall m\in\mathcal{M}$. For any subset of the set of messages $\mathcal{B}\subset\mathcal{W}$, we use the following notation.
$$\mathcal{V}_{\mathcal{B}} = \{\mathbf{v} \in \mathbb{F}^{n}: \mathbf{v} \mbox{ is a column vector of }\mathbf{V}_{i}, i \in \mathcal{B}\},$$
i.e., 
$\mathcal{V}_{\mathcal{B}}$ is the set of all column vectors which belong to at least one matrix $\mathbf{V}_{i}: i \in \mathcal{B}$.

 Receiver $k$ receives a $n$ dimensional vector $S^n$ which is a linear combination of the $ML$ column vectors of $\mathcal{V}_\mathcal{W}.$ If  all the $ML$ column vectors are linearly independent, then, clearly the desired signal is resolvable and a symmetric rate of $1/M$ per message is achievable. Note that this is the rate achieved by routing; the routing solution is in fact one such realization of $\mathcal{V}_\mathcal{W}.$ In general, however, the messages may be resolvable even if the $ML$ column vectors of $\mathcal{V}_\mathcal{W}$ are linearly \emph{dependent}, because of the presence of antidotes. In fact, a rate greater than $1/M$ is possible only by making the column vectors of $\mathcal{V}_\mathcal{W}$ linearly \emph{dependent}. Consider, for instance a receiver, say destination $D_k$ that wants to decode messages $\mathcal{W}_k$ and has antidotes for messages $\mathcal{A}_k$. It receives the linear combination of $ML$ vectors in an $n$ dimensional space, of which it can cancel the impact of $|\mathcal{A}_{k}| L $ vectors, $\mathcal{V}_{\mathcal{A}_k}$, using the antidotes. Therefore, from the perspective of  destination $D_k$, it observes (after cancellation), the linear combination of $|\mathcal{W}_k|L$ desired vectors along $\mathcal{V}_{\mathcal{W}_k}$, and $(M-|\mathcal{A}_{k}|-|\mathcal{W}_k|)L$ interfering vectors along the columns of $\mathcal{V}_{\mathcal{W}-(\mathcal{A}_k\cup\mathcal{W}_k)}$. A necessary condition for the resolvability of messages $\mathcal{W}_k$ at destination $D_k$ can be expressed as  
\begin{equation} \mbox{span}(\mathcal{V}_{\mathcal{W}_k}) \cap \mbox{span}(\mathcal{V}_{\mathcal{W}-(\mathcal{A}_k\cup\mathcal{W}_k)}) = \{0\}
\label{eq:resolvability}
\end{equation}
This means that a necessary condition for resolvability at destination $D_k$, is that the dimension of interference $\mbox{span}(\mathcal{V}_{\mathcal{W}-(\mathcal{A}_k\cup\mathcal{W}_k)})$ should be smaller than $n - |\mathcal{W}_k|L$ (because the vectors are all observed in an $n$ dimensional space). Clearly, if $n - |\mathcal{W}_k|L < |{\mathcal{W}-(\mathcal{A}_k\cup\mathcal{W}_k)}|L,$  then, the interfering vectors need to \emph{align}  in an $n-|\mathcal{W}_k|L $ dimensional space.

Next we illustrate the role of interference alignment in index coding, with a series of examples presented in increasing order of complexity.

\subsection{Example 1: Scalar linear index coding with One-to-One Alignment}
\begin{figure}[!h] \centering
\includegraphics[width=3.2in]{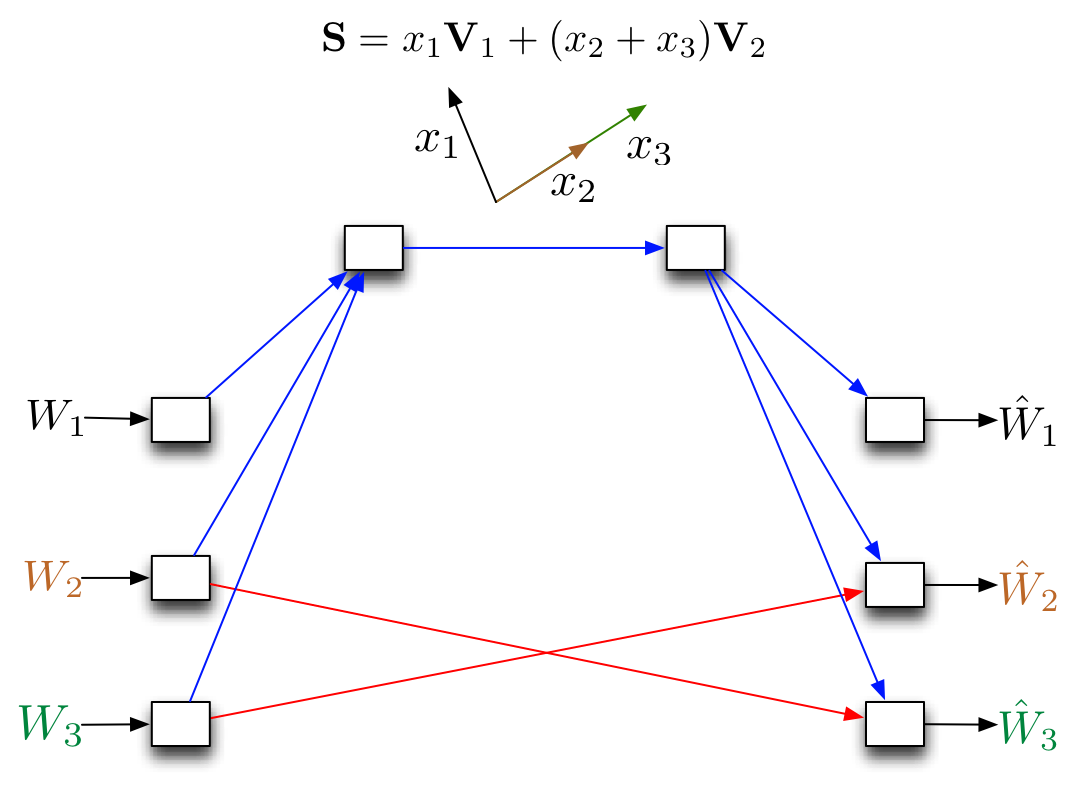}
\caption{\small A simple example where interference alignment is useful in the index coding problem. In the example, the alignment is realized with $\mathbf{V}_{2}=\mathbf{V}_{3}.$ The alignment enables transmission of $3$ scalars to the $3$ corresponding users in a $2$ dimensional vector space, ensuring that a rate of 1/2 is achievable. The field $\mathbb{F}$ can be chosen arbitrarily.}
\label{fig:AlignmentIndex}
\end{figure}
A simple interference alignment solution for an index coding problem  is demonstrated in Figure \ref{fig:AlignmentIndex}, where $K=3$ and each user sends $L=1$ vector. Because of alignment of $\mathbf{V}_{2}$ and $\mathbf{V}_{3},$ user $1$ is able to resolve $x_1.$ For instance one may choose $\mathbf{V}_1=[0,1]^T$, $\mathbf{V}_2=\mathbf{V}_3=[1,0]^T$, so that the two transmitted symbols on the bottleneck link are $S^2=(S_1, S_2)=(x_2+x_3, x_1)$, from which each destination is able to recover its desired message. Since only one symbol is sent per message, this is an example of a scalar linear index coding solution. Furthermore, since the alignment of vectors takes place in a one-to-one fashion, i.e., $\mathbf{V}_2$ aligns with $\mathbf{V}_3$, we refer to this as a one-to-one alignment solution, to be distinguished from the subspace alignment solutions to be presented soon.
 
 \subsection{Example 2: Vector linear index coding with One-to-One Alignment}\label{sec:example2}
 The next index coding example comes from the CBIA setting shown in Figure \ref{fig:5cellpic}(a), (b), (c). Depending upon whether each base station transmitter has a message only for one corresponding receiver, or an independent message for each of the receivers that are within-range, we have the interference channel or the X channel setting, respectively. Here we consider the interference channel setting.
\begin{figure}[h] \centering
\includegraphics[width=5.9in]{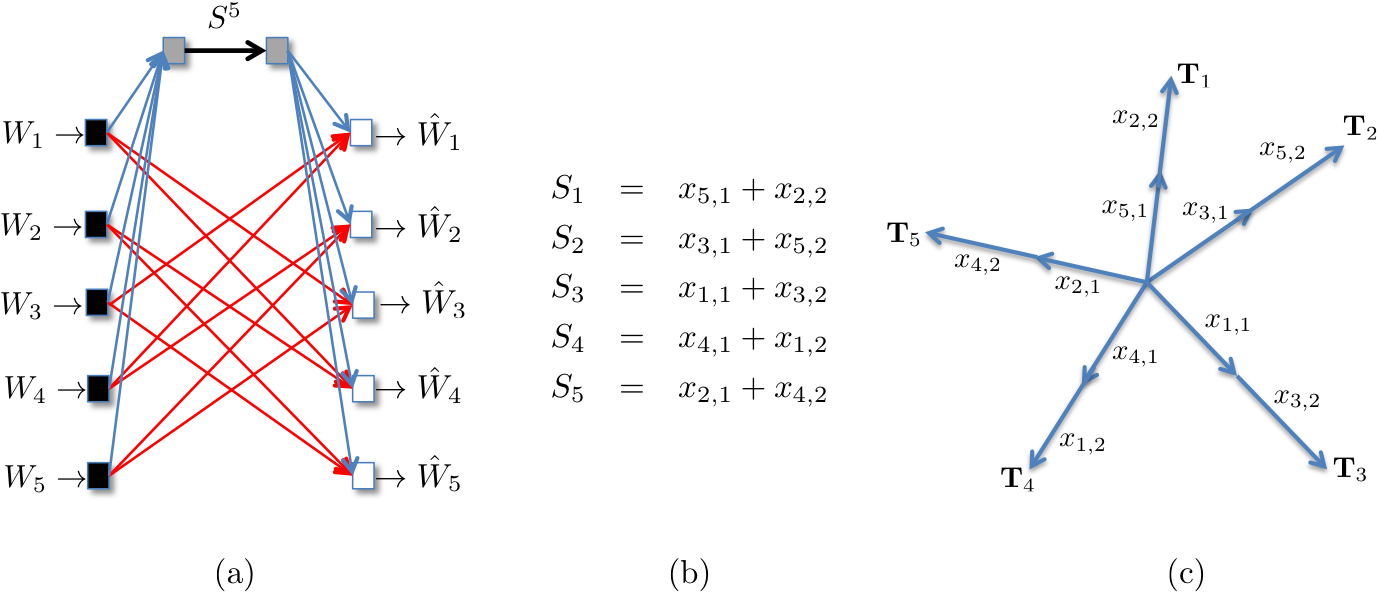}
\caption{\small Index coding problem corresponding to a CBIA interference channel setting, (a) Desired message and antidote sets, (b) Capacity optimal solution shows the 5 dimensional $S^5$ space, along which 2 symbols per message are sent with one-to-one pairwise alignments along the 5 orthogonal basis vectors. The  field $\mathbb{F}$ can be chosen arbitrarily.}
\label{fig:indexIAex2}
\end{figure}

The interference channel setting means that base station (source) $i,~1\leq i \leq 5$, has only one message, $W_i$, for its corresponding receiver (destination) $i$. So there are $M=5$ distinct messages and $K=5$ destinations. The goal is to achieve a symmetric rate of $\frac{2}{5}$ per message, which is also the capacity optimal solution. The problem is essentially identical to the setting considered  in \cite{Alon_Hasidim_Lubetzky_Stav_Weinstein} to show that vector linear index coding outperforms scalar linear index coding. The achievable scheme is a vector linear scheme operating over 5 channel uses.  The precoding vectors $\mathbf{V}_i, 1\leq i\leq 5$ are $5 \times 2$ matrices and $\mathbf{X}_i$ is a $2 \times 1$ vector $\mathbf{X}_i=[x_{i,1},x_{i,2}]^T$ representing $W_i$. Therefore, the transmitted symbol sequence $S^5$ is
\begin{eqnarray}
S^5=\sum_{i=1}^{5}\mathbf{V}_i\mathbf{X}_i
\end{eqnarray}
To see why interference alignment is necessary, note that each destination has access to the $5$-dimensional symbol $S^5$, and it also observes $2$ messages  as cognitive information. After removing the known streams from $S^5$ the receiver is left with $6$ remaining unknown symbols in a $5$-dimensional space. Since each destination is interested in two desired streams, these desired streams span $2$-dimensional space and the remaining $4$ streams that constitute interference, must align in way that they occupy at most a $3$-dimensional space. 

Suppose $\mathbf{T}_1,\mathbf{T}_2,\mathbf{T}_3,\mathbf{T}_4,\mathbf{T}_5$ are $5$ linearly independent vectors over the $5$-dimensional space. These $5$ vectors can be chosen to be the columns of $5 \times 5$ identity matrix. At destination 1, messages $W_2$ and $W_5$,  each composed of two independent scalar streams, should align such that they together occupy a 3 dimensional space. One way to do so is to perfectly align the precoding vector of one of two streams of $W_{2}$ with one of two streams of $W_5$, e.g., $\mathbf{V}_{2,2}=\mathbf{V}_{5,1}=\mathbf{T}_1$. Similarly, in order to satisfy the alignment constraint at all the destinations, the precoding vectors are chosen to be 

\begin{eqnarray}
\mathbf{V}_{2,2}=\mathbf{V}_{5,1}=\mathbf{T}_1,~\mathbf{V}_{5,2}=\mathbf{V}_{3,1}=\mathbf{T}_2,~
\mathbf{V}_{3,2}=\mathbf{V}_{1,1}=\mathbf{T}_3,~\mathbf{V}_{1,2}=\mathbf{V}_{4,1}=\mathbf{T}_4,~
\mathbf{V}_{4,2}=\mathbf{V}_{2,1}=\mathbf{T}_5
\end{eqnarray}

Thus,  the 4 undesired signal vectors are aligned at each destination such that they occupy only 3 dimensions. Now what we need to prove is the resolvability of desired messages at the corresponding destinations. They are resolvable because at each destination $i$ the desired messages are received in a space  $\mathcal{D}_i$ which is linearly independent from the space of interfering messages called $\mathcal{I}_i$.
\begin{eqnarray}
&&\mathcal{D}_1=\left[\begin{array}{ccc}
0&0\\
0&0\\
1&0\\
0&1\\
0&0
\end{array}\right],~~
\mathcal{I}_1=\left[\begin{array}{ccc}
1&0&0\\
0&1&0\\
0&0&0\\
0&0&0\\
0&0&1
\end{array}\right]~~
\mathcal{D}_2=\left[\begin{array}{ccc}
0&1\\
0&0\\
0&0\\
0&0\\
1&0
\end{array}\right],~~
\mathcal{I}_2=\left[\begin{array}{ccc}
0&0&0\\
1&0&0\\
0&1&0\\
0&0&1\\
0&0&0
\end{array}\right]\nonumber\\
&&\mathcal{D}_3=\left[\begin{array}{ccc}
0&0\\
1&0\\
0&1\\
0&0\\
0&0
\end{array}\right],~~
\mathcal{I}_3=\left[\begin{array}{ccc}
1&0&0\\
0&0&0\\
0&0&0\\
0&1&0\\
0&0&1
\end{array}\right]~~
\mathcal{D}_4=\left[\begin{array}{ccc}
0&0\\
0&0\\
0&0\\
1&0\\
0&1
\end{array}\right],
\mathcal{I}_4=\left[\begin{array}{ccc}
1&0&0\\
0&1&0\\
0&0&1\\
0&0&0\\
0&0&0
\end{array}\right] \nonumber\\
&&\mathcal{D}_5=\left[\begin{array}{ccc}
1&0\\
0&1\\
0&0\\
0&0\\
0&0
\end{array}\right],~~
\mathcal{I}_5=\left[\begin{array}{ccc}
0&0&0\\
0&0&0\\
1&0&0\\
0&1&0\\
0&0&1
\end{array}\right]
\end{eqnarray}
Finally, if we choose the receiver combining matrices as follows, we can easily verify that both property 1 and property 2 are satisfied
\begin{eqnarray}
&&\text{At destination 1:}~~~~\mathbf{U}_{1,1}=\left[\begin{array}{ccccc}
0&0&1&0&0\\
0&0&0&1&0
\end{array}\right],~~
\text{At destination 2:}~~~~\mathbf{U}_{2,2}=\left[\begin{array}{ccccc}
0&0&0&0&1\\
1&0&0&0&0
\end{array}\right] \nonumber\\
&&\text{At destination 3:}~~~~\mathbf{U}_{3,3}=\left[\begin{array}{ccccc}
0&1&0&0&0\\
0&0&1&0&0
\end{array}\right],~~
\text{At destination 4:}~~~~\mathbf{U}_{4,4}=\left[\begin{array}{ccccc}
0&0&0&1&0\\
0&0&0&0&1
\end{array}\right]\nonumber\\
&&\text{At destination 5:}~~~~\mathbf{U}_{5,5}=\left[\begin{array}{ccccc}
1&0&0&0&0\\
0&1&0&0&0\end{array}\right]
\end{eqnarray}
Note that the  field can be chosen arbitrarily, i.e., any choice of $\mathbb{F}$ works for the linear solution presented above. The solution also translates to the real or complex fields, thereby establishing the DoF for the corresponding CBIA problem as $2/5$ per message.

\subsection{Example 3: Scalar linear index coding with Subspace Alignment}\label{sec:Xexample}

Like the previous example, this example also corresponds to the CBIA problem in Fig. \ref{fig:5cellpic}(a),(b),(c). The difference is in the message sets. While in the previous example, each base station served only one receiver, here we assume that each base station has 3 independent messages, one for each of the 3 receivers that are within receiving range of the base station. In the parlance of wireless networks, while the previous setting is an interference network, the current setting is an $X$ network. In an $X$ network, there is an independent message to be communicated between each transmitter-receiver pair that are within range of each other, i.e.,  have a non-zero channel coefficient between them. 

The index coding problem for the $X$ network setting is also well defined. There is an independent message between each source and destination pair that are not connected via an antidote link. The specific index coding problem that we solve in this example is shown in Figure \ref{fig:indexIAex3}. Notice that an $X$ network setting is a mulitple unicast setting, since each message has a unique source and a unique destination. It is also sometimes referred to as the \emph{all unicast} setting, since here all possible non-trivial unicast flows are simultaneously active (the unicast flows between source destination pairs connected by infinite capacity antidote links are ignored because they trivially have infinite rate). Also note that as always it is possible, without  loss of generality, to represent this index coding problem with only one message per source and only one message per destination by increasing the number of sources and destinations, but we prefer the compact representation  shown in Figure \ref{fig:indexIAex3} which directly reflects the $X$ channel setting. $X$ networks often lead to interesting interference alignment problems. Indeed, that is the case with this example as well, where one-to-one alignment does not suffice and the optimal index coding scheme is a scalar linear subspace alignment scheme. Next we proceed to describe the alignment solution.

\begin{figure}[h] \centering
\includegraphics[width=5.9in]{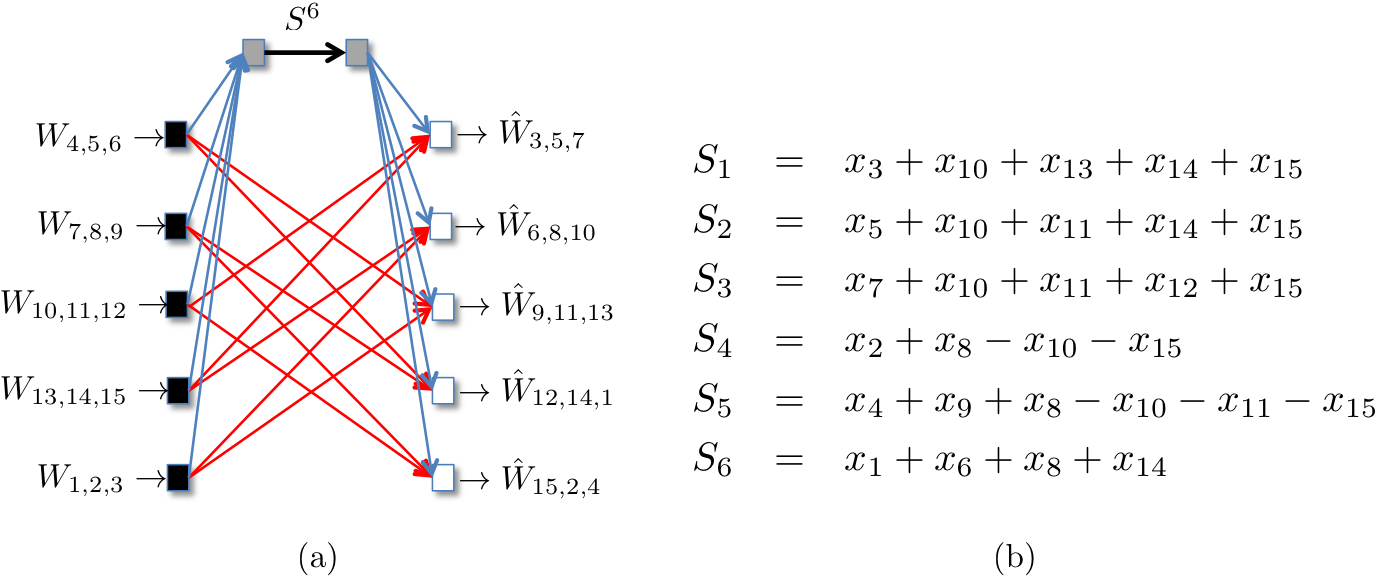}
\caption{\small Index coding problem corresponding to a CBIA X channel setting, (a) Desired message and antidote sets, (b) Capacity optimal solution shows the 6 dimensional $S^6$ space, along which 15 symbols  are sent. The  field $\mathbb{F}$  can be chosen arbitrarily.}
\label{fig:indexIAex3}
\end{figure}

There are $M=15$ distinct messages and $K=5$ destinations. Each black square on the left is a base station (source) sending $3$ distinct messages.  Base station (source) $i, 1 \leq i \leq 4$ sends messages $W_{3i+1:3i+3}$ and the $5$-th  base station  sends $W_{1:3}$. We have 
\begin{eqnarray}
&\mathcal{W}_1=\{W_{3,5,7}\},&\mathcal{A}_1=\{W_{1,2,4,6,8,9}\cup{\mathcal{W}}_1\}^c \nonumber\\
&\mathcal{W}_2=\{W_{6,8,10}\},&\mathcal{A}_2=\{W_{4,5,7,9,11,12}\cup{\mathcal{W}}_2\}^c \nonumber\\
&\mathcal{W}_3=\{W_{9,11,13}\},&\mathcal{A}_3=\{W_{7,8,10,12,14,15}\cup\mathcal{W}_3\}^c\nonumber\\
&\mathcal{W}_4=\{W_{12,14,1}\},&\mathcal{A}_4=\{W_{10,11,13,15,2,3}\cup\mathcal{W}_4\}^c\nonumber\\
&\mathcal{W}_5=\{W_{15,2,4}\},&\mathcal{A}_5=\{W_{13,14,1,3,5,6}\cup\mathcal{W}_5\}^c \nonumber
\end{eqnarray}
Our goal is to achieve the symmetric rate of $\frac{1}{6}$ per message, which is also the capacity of this network (the outer bound follows from a subsequent solution of a broader class of $X$ networks, presented in Section \ref{sec:X}). The achievable scheme is a scalar linear achievable scheme over $6$ channel uses and hence the precoding vectors $\mathbf{V}_i, 1\leq i\leq 15$, are $6 \times 1$ vectors. Therefore, the transmitted symbol sequence $S^6$ is
\begin{eqnarray}
S^6=\sum_{i=1}^{15}\mathbf{V}_ix_i
\end{eqnarray}
 where the scalar symbol $x_i$ represents message $W_i$. To see why interference alignment is necessary, note that each destination has access to the $6$-dimensional symbol $S^6$, and it also observes $6$ messages as cognitive information. After removing the known symbols from $S^6$ the receiver is left with $9$ remaining unknown symbols in a $6$-dimensional space. Since each destination is interested in $3$ desired symbols,  they must occupy a $3$-dimensional space, leaving only $3$ dimensions within which the remaining $6$ undesired symbols that constitute interference, must align.

Suppose $\mathbf{T}_1,\mathbf{T}_2,\mathbf{T}_3,\mathbf{T}_4,\mathbf{T}_5,\mathbf{T}_6$ are $6$ linearly independent vectors over $6$-dimensional space. These $6$ vectors can be chosen to be the columns of $6 \times 6$ identity matrix. Since $W_{3,5,7}$ are desired at destination 1, they should span a $3$-dimensional space and therefore are sent over  $3$ linearly independent vectors $\mathbf{T}_1,\mathbf{T}_2,\mathbf{T}_3$. On the other hand, $W_{2,4}$ are the interfering messages at destination $1$ and both are desired at destination $5$. So these two messages should be sent over vectors that are linearly independent of each other and linearly independent from $\mathbf{T}_1,\mathbf{T}_2,\mathbf{T}_3$. Therefore, we send $W_{2},W_{4}$ over $\mathbf{T}_4,\mathbf{T}_5$, respectively. Also $W_{1}$ is desired at destination $4$ and is considered as an interfering message at both destination $1$ and destination $5$. So $W_{1}$ should be sent over a vector that is linearly independent from $\mathbf{T}_1,\ldots,\mathbf{T}_5$ and is chosen to be  sent along $\mathbf{T}_6$.  

Interfering messages $W_{1},W_{2},W_{4}$ occupy a $3$-dimensional space at destination $1$. The remaining interfering messages at destination $1$, i.e., $W_{6,8,9}$, should be sent along  precoding vectors chosen such that they remain in the same span as $\text{span}(\mathbf{T}_4,\mathbf{T}_5,\mathbf{T}_6)$, i.e.,

\begin{eqnarray} \label{eqn:alignexample0}
\text{span} (\mathbf{V}_6,\mathbf{V}_8,\mathbf{V}_9) \in \text{span}(\mathbf{T}_4,\mathbf{T}_5,\mathbf{T}_6)
\end{eqnarray}
One way to satisfy (\ref{eqn:alignexample0}) is by choosing $\mathbf{V}_6=\mathbf{T}_6$, $\mathbf{V}_9=\mathbf{T}_5$. The remaining precoding vector is designed later to satisfy the following 
\begin{eqnarray} \label{eqn:alignexample}
\text{At destination}~1:\text{span} (\mathbf{V}_8) \in \text{span}(\mathbf{T}_4,\mathbf{T}_5,\mathbf{T}_6) \label{eqn:span1}
\end{eqnarray}
Similarly, to satisfy the requirement of aligning interfering messages in $3$-dimensional space at different destinations, we design the precoding vectors as follows:
\begin{eqnarray}
\text{At destination}~2:~&&\text{span} (\mathbf{V}_{11}) \in \text{span}(\mathbf{T}_2,\mathbf{T}_3,\mathbf{T}_5) \\
\text{At destination}~3:~&&\text{span} (\mathbf{V}_{14}) \in \text{span}(\mathbf{V}_8,\mathbf{T}_3,\mathbf{V}_{10})\\
\text{At destination}~4:~&&\text{span} (\mathbf{V}_2)=\text{span} (\mathbf{T}_4) \in \text{span}(\mathbf{V}_{11},\mathbf{V}_{10},\mathbf{T}_{1})\\
\text{At destination}~5:~&&\text{span} (\mathbf{V}_5)=\text{span} (\mathbf{T}_2) \in \text{span}(\mathbf{V}_{14},\mathbf{T}_1,\mathbf{T}_{6}) \label{eqn:span5}\\ 
&&\mathbf{V}_{12}=\mathbf{T}_{3},~\mathbf{V}_{15}=\mathbf{V}_{10},~\mathbf{V}_{13}=\mathbf{T}_{1}
\end{eqnarray} 

\noindent leading to the formulation
\begin{eqnarray}
&&\mathbf{V}_8 =a_1\mathbf{T}_4+a_2\mathbf{T}_5+a_3\mathbf{T}_6 \label{eqn:lincomb1}\\
&&\mathbf{V}_{11}=b_1\mathbf{T}_2+b_2\mathbf{T}_3+b_3\mathbf{T}_5 \label{eqn:lincomb2}\\
&&\mathbf{V}_{10} =c_1\mathbf{V}_{14}+c_2\mathbf{T}_3+c_3\mathbf{V}_{8}\label{eqn:lincomb3}\\
&&\mathbf{V}_{10}=d_1\mathbf{V}_{11}+d_2\mathbf{T}_{1}+d_3\mathbf{T}_{4}\label{eqn:lincomb4}\\
&&\mathbf{V}_{14}=e_1\mathbf{T}_1+e_2\mathbf{T}_{2}+e_3\mathbf{T}_6\label{eqn:lincomb5}
\end{eqnarray} 
where $a_1,b_1,\ldots,d_3,e_3$ are linear combination coefficients. In order to satisfy (\ref{eqn:lincomb1})-(\ref{eqn:lincomb5}), the only thing that restricts us from choosing the linear combination coefficients to be random is that $\mathbf{V}_{10}$ should satisfy both (\ref{eqn:lincomb3}) and (\ref{eqn:lincomb4}). If we substitute from (\ref{eqn:lincomb1}),(\ref{eqn:lincomb2}) and (\ref{eqn:lincomb5}) into (\ref{eqn:lincomb3}) and (\ref{eqn:lincomb4}), $\mathbf{V}_{10}$ should satisfy the following two equations 
\begin{eqnarray}
&&\mathbf{V}_{10} =c_1e_1\mathbf{T}_{1}+c_1e_2\mathbf{T}_2+c_2\mathbf{T}_{3}+c_3a_1\mathbf{T}_{4}+c_3a_2\mathbf{T}_{5}+(c_3a_3+c_1e_3)\mathbf{T}_{6}\label{eqn:lincomb6}\\
&&\mathbf{V}_{10}=d_2\mathbf{T}_{1}+d_1b_1\mathbf{T}_{2}+d_1b_2\mathbf{T}_{3}+d_3\mathbf{T}_{4}+d_1b_3\mathbf{T}_{5}\label{eqn:lincomb7}
\end{eqnarray} 
To satisfy (\ref{eqn:lincomb6}) and (\ref{eqn:lincomb7}), we have the following
\begin{eqnarray}
&&d_2=c_1e_1\\
&&d_1b_1=c_1e_2\\
&&d_1b_2=c_2\\
&&d_3=c_3a_1\\
&&d_1b_3=c_3a_2\\
&&c_3a_3+c_1e_3=0
\end{eqnarray}
Clearly, there are many solutions. One of the solutions for these system of nonlinear polynomial equations is $a_1=a_2=a_3=b_1=b_2=c_1=c_2=d_1=d_2=e_1=e_2=e_3=1$, $b_3=c_3=d_3=-1$, according to which,  the precoding vectors are chosen as follows 
\begin{eqnarray}
\mathbf{V}_8&=&\mathbf{T}_4+\mathbf{T}_5+\mathbf{T}_6 \nonumber\\
\mathbf{V}_{11}&=&\mathbf{T}_2+\mathbf{T}_3-\mathbf{T}_5 \nonumber\\
\mathbf{V}_{14}&=&\mathbf{T}_1+\mathbf{T}_2+\mathbf{T}_6 \nonumber\\
\mathbf{V}_{10}&=&\mathbf{T}_1+\mathbf{T}_2+\mathbf{T}_3-\mathbf{T}_4-\mathbf{T}_5 \nonumber
\end{eqnarray} 
Note that this is not a one-to-one alignment solution, e.g., $\mathbf{V}_8$ does not align with $\mathbf{T}_4, \mathbf{T}_5, \mathbf{T}_6$ individually. In fact it is pairwise linearly independent of all three. $\mathbf{V}_8$ aligns only within the subspace spanned by  $\mathbf{T}_4, \mathbf{T}_5, \mathbf{T}_6$. This is referred to as subspace alignment. 

After satisfying all the alignment constraints, we need to prove the resolvability of desired messages at the corresponding destinations. They are resolvable because at each destination $i$ the desired messages are received in a space called $\mathcal{D}_i$ which is linearly independent from the space of interfering messages, $\mathcal{I}_i$.
\allowdisplaybreaks
\begin{eqnarray*}
&&\mathcal{D}_1=[\mathbf{V}_3 ~\mathbf{V}_5~\mathbf{V}_7]=\left[\begin{array}{ccc}
1&0&0\\
0&1&0\\
0&0&1\\
0&0&0\\
0&0&0\\
0&0&0
\end{array}\right],~~
\mathcal{I}_1=\left[\begin{array}{ccc}
0&0&0\\
0&0&0\\
0&0&0\\
1&0&0\\
0&1&0\\
0&0&1
\end{array}\right]\\
&&\mathcal{D}_2=[\mathbf{V}_6 ~\mathbf{V}_8~\mathbf{V}_{10}]=\left[\begin{array}{ccc}
0&0&1\\
0&0&1\\
0&0&1\\
0&1&-1\\
0&1&-1\\
1&1&0
\end{array}\right],~~
\mathcal{I}_2=\left[\begin{array}{ccc}
0&0&0\\
1&0&0\\
0&1&0\\
0&0&0\\
0&0&1\\
0&0&0
\end{array}\right]\nonumber\\
&&\mathcal{D}_3=[\mathbf{V}_9 ~\mathbf{V}_{11}~\mathbf{V}_{13}]=\left[\begin{array}{ccc}
0&0&1\\
0&1&0\\
0&1&0\\
0&0&0\\
1&-1&0\\
0&0&0
\end{array}\right],~~
\mathcal{I}_3=\left[\begin{array}{ccc}
0&0&1\\
0&0&1\\
0&1&0\\
1&0&0\\
1&0&0\\
1&0&1
\end{array}\right]\\
&&\mathcal{D}_4=[\mathbf{V}_{12} ~\mathbf{V}_{14}~\mathbf{V}_{1}]=\left[\begin{array}{ccc}
0&1&0\\
0&1&0\\
1&0&0\\
0&0&0\\
0&0&0\\
0&1&1
\end{array}\right],~~
\mathcal{I}_4=\left[\begin{array}{ccc}
1&0&0\\
0&0&1\\
0&0&1\\
0&1&0\\
0&0&-1\\
0&0&0
\end{array}\right] \nonumber\\
&&\mathcal{D}_5=[\mathbf{V}_{15} ~\mathbf{V}_{2}~\mathbf{V}_{4}]=\left[\begin{array}{ccc}
1&0&0\\
1&0&0\\
1&0&0\\
-1&1&0\\
-1&0&1\\
0&0&0
\end{array}\right],~~
\mathcal{I}_5=\left[\begin{array}{ccc}
1&0&0\\
0&1&0\\
0&0&0\\
0&0&0\\
0&0&0\\
0&0&1
\end{array}\right]
\end{eqnarray*}
Finally, the receiver combining matrices are chosen as follows, so that both property 1 and property 2 are satisfied
\begin{eqnarray}
&&\text{At destination 1:}~~~~\mathbf{U}_{3,1}=[1~0~0~0~0~0],~~\mathbf{U}_{5,1}=[0~1~0~0~0~0],~~\mathbf{U}_{7,1}=[0~0~1~0~0~0] \nonumber \\
&&\text{At destination 2:}~~~~\mathbf{U}_{6,2}=[-1~0~0~-1~0~1] ,~~\mathbf{U}_{8,2}=[1~0~0~1~0~0],~~\mathbf{U}_{10,2}=[1~0~0~0~0~0] \nonumber \\
&&\text{At destination 3:}~~~~\mathbf{U}_{9,3}=[0~1~0~0~1~-1],~~\mathbf{U}_{11,3}=[0~1~0~1~0~-1],~~\mathbf{U}_{13,3}=[1~0~0~1~0~-1]\nonumber \\
&&\text{At destination 4:}~~~~\mathbf{U}_{12,4}=[0~0~1~0~1~0],~~\mathbf{U}_{14,4}=[0~1~0~0~1~0],~~\mathbf{U}_{1,4}=[0~-1~0~0~-1~1]\nonumber \\
&&\text{At destination 5:}~~~~\mathbf{U}_{15,5}=[0~0~1~0~0~0],~~\mathbf{U}_{2,5}=[0~0~1~1~0~0],~~\mathbf{U}_{4,5}=[0~0~1~0~1~0]\nonumber
\end{eqnarray}
Note that the  field can be chosen arbitrarily, i.e., any choice of $\mathbb{F}$ works for the linear solution presented above. The solution also translates to the real or complex fields, thereby establishing the DoF for the corresponding CBIA problem as $1/6$ per message.

\section{Background and  Results}
In this  section we present the results of this work along with the relevant background comprised of related prior work. Our first result is the insufficiency of linear codes for the multiple unicast index coding problem.
\subsection{Insufficiency of Linear Codes for Multiple Unicast Index Coding}
The index coding problem was introduced in the multiple \emph{unicast} setting by Birk and Kol in \cite{Birk_Kol_Trans} and subsequently studied by Bar-Yossef et al. in \cite{Yossef_Birk_Jayram_Kol}, where it was shown that the minimum number of channel uses required to send 1 bit per message with the optimal binary scalar linear index coding scheme is equal to the $\mathrm{minrank}$ function. Further, as noted in \cite{Yossef_Birk_Jayram_Kol_Trans}, this result can be generalized by allowing the coding to be performed over an arbitrary finite field. Algorithms for finding or approximating the minrank function are proposed in \cite{Chaudhry_Sprintson} and \cite{Chlamtac_Haviv}. Haviv and Langberg in \cite{Haviv_Langberg} investigated the min-ranks of random digraphs. Within the class of scalar index coding schemes, the insufficiency of linear codes was established by Lubetzky and Stav in \cite{Lubetzky_Stav}.  Alon et al.  further established the sub-optimality of scalar linear solutions in \cite{Alon_Hasidim_Lubetzky_Stav_Weinstein}  by highlighting  the benefits of vector linear coding. Since vector linear coding is a much more powerful form of linear coding than scalar linear coding, its optimality relative to non-linear coding schemes has been a topic of great interest, and has produced a series of results in  general contexts that have progressively closed in on the index coding problem. The series of breakthroughs started with the work of Dougherty et al. in \cite{Dougherty_Freiling_Zeger} who used the connections between the network coding problem and representability of matroids to establish a gap between vector linear coding and non-linear coding for the multiple unicast \emph{network coding} problem. However, the gap result of Dougherty et al. did not apply to the index coding problem which is a special case of the  network coding problem. Building upon the work of \cite{Dougherty_Freiling_Zeger},  Rouayheb et al. established connections between the representability of matroids and the multiple \emph{groupcast} index coding problem.  Most recently, Blasiak et al. further developed this relationship in \cite{Blasiak_Kleinberg_Lubetzky_2011}, and used the construction of lexicographic products of graphs to show that for the multiple groupcast index coding problem,  vector linear codes are strictly out-performed by non-linear codes.   However, the gap result for the multiple \emph{groupcast} index coding problem does not apply to the multiple \emph{unicast} index coding problem. Hence, in the original setting of multiple \emph{unicast} index coding, which is also the most commonly studied form of index coding, the optimality of  linear codes remains open. 

Our first result, presented in Theorem \ref{thm:insuff},  settles this issue.
\begin{theorem}\label{thm:insuff}
Linear coding is insufficient to achieve the capacity region of the multiple {\bf unicast} index coding
problem.
\end{theorem}


The proof of Theorem \ref{thm:insuff} is presented in Section \ref{sec:insuff}. 
The key  to this result is the construction of an equivalent multiple unicast problem for an arbitrary groupcast setting. Loosely speaking, this equivalent multiple unicast setting has the following two properties:
\begin{itemize}
\item Any rate that is achievable with linear coding in the groupcast setting is also achievable with linear coding in the multiple-unicast setting and \emph{vice-versa}. 
\item Any rate achievable in the groupcast setting (using any achievable scheme including possibly non-linear strategies) is achievable in the multiple unicast setting.
\end{itemize}
Since \cite{Blasiak_Kleinberg_Lubetzky_2011} has a construction of a groupcast index coding problem where linear coding schemes are insufficient, the above conditions imply that linear coding is insufficient in the equivalent multiple unicast setting as well. 

The construction of the equivalent multiple unicast index coding problem is described next. The basic idea of the construction is to replace each message that is desired by multiple destinations, with a new set of independent messages, one for each original desired destination. Thus, each destination desires only one of these independent messages (and has the rest as antidotes), giving us a unicast setting.  The equivalence to the original groupcast setting is enforced through a requirement that these independent messages are to align into the same space as the original message that they replace.  This is accomplished by introducing auxiliary messages and destinations (that contain the subscript $0$ in the description below), one for each original message. Each auxiliary destination sees  only one set of messages as interference (has all other messages as antidotes) and its desired auxiliary message leaves only as many signal dimensions for interference as each of the members of the set, thus forcing them into alignment. It is this alignment that allows the optimal solution in the multiple groupcast setting to be used in the multiple unicast setting and vice-versa.

\begin{construction}
Consider an arbitrary groupcast index coding problem where there are $M$ messages and $K$ destination nodes. Without loss of generality\footnote{There is no loss of generality in the first assumption because if there is a message desired by $L'<L$  destinations, then we can add $L-L'$ virtual destinations with antidote sets identical to any of the $L'$ (original) destinations, and which desire the appropriate subset of the messages desired by the original destinations. Similarly, there is no loss of generality in the second assumption because if a destination desires multiple messages, it can be equivalently replaced by multiple copies of itself, each interested in only one of the originally desired messages.}, we assume that \begin{itemize}
\item each message is desired by $L$ destinations so that $K=LM$, and
\item each destination desires exactly one message.
\end{itemize} In this groupcast index coding problem, the set of messages, the set of destinations, and the set of antidotes are respectively denoted by $\mathcal{W}, \mathcal{D}, \{\mathcal{A}_{k}:k=1,2,\ldots,K\},$ where $$D_{k} = (\{W_{m}\},\mathcal{A}_{k}), m=\lceil k/L\rceil$$

We now construct an equivalent multiple unicast index coding problem as follows. The multiple unicast setting has $LM+M$ messages and $K+M$ destinations. In this multiple unicast index coding problem, denote the messages as 

$$\overline{\mathcal{W}} = \{\overline{W}_{1,0}, \overline{W}_{1,1}, \overline{W}_{1,2},\ldots, \overline{W}_{1,L}, \overline{W}_{2,0}, \overline{W}_{2,1},\ldots,\overline{W}_{2,L},\ldots,\overline{W}_{M,0},\overline{W}_{M,1}, \overline{W}_{M,2},\ldots,\overline{W}_{M,L}\}. $$ 

\noindent The  destinations  in the multiple unicast system are denoted as $$\overline{D}_{i,j}:i=1,2,\ldots, M, j=0, 1,2,\ldots,L$$ and the antidotes are denoted as $$\overline{\mathcal{A}}_{i,j}, i=1,2,\ldots, M, j=0,1,2,\ldots,L$$ so that destination $\overline{D}_{i,j} = (\overline{W}_{i,j}, \overline{\mathcal{A}}_{i,j}).$ The set of antidotes is 
\begin{equation} \label{eq:equivalent_antidotes}\overline{\mathcal{A}}_{i,j} = \left\{ \begin{array}{ll} \mathcal{A}_{(i-1)L+j,.} \cup \overline{\mathcal{W}}_{.,0} \cup \{\overline{W}_{i,l}: l \neq j\}, & j \neq 0 \\ \overline{\mathcal{W}}-\overline{\mathcal{W}}_{i,.}, & j=0 \end{array}\right\} , \end{equation}
where
\begin{eqnarray}
\mathcal{A}_{k,.} &=& \{\overline{W}_{m,l}:  W_{m} \in \mathcal{A}_{k}, l=0,1,2,\ldots,L \}\label{eq:antidotes_first}\\
 \overline{\mathcal{W}}_{.,0} &=& \{\overline{W}_{1,0}, \overline{W}_{2,0},\cdots, \overline{W}_{M,0}\}\\
\overline{\mathcal{W}}_{k,.} &= &\{\overline{W}_{k,0},\overline{W}_{k,1},\cdots,\overline{W}_{k,L}\}
\label{eq:antidotes_last}
\end{eqnarray}
\label{construction1}
\end{construction}

As an example, a groupcast index coding problem and its equivalent multiple unicast index coding problem are shown in Figure {\ref{fig:equivalentsetting}}.
\begin{figure}[!t] \centering
\includegraphics[width=4in]{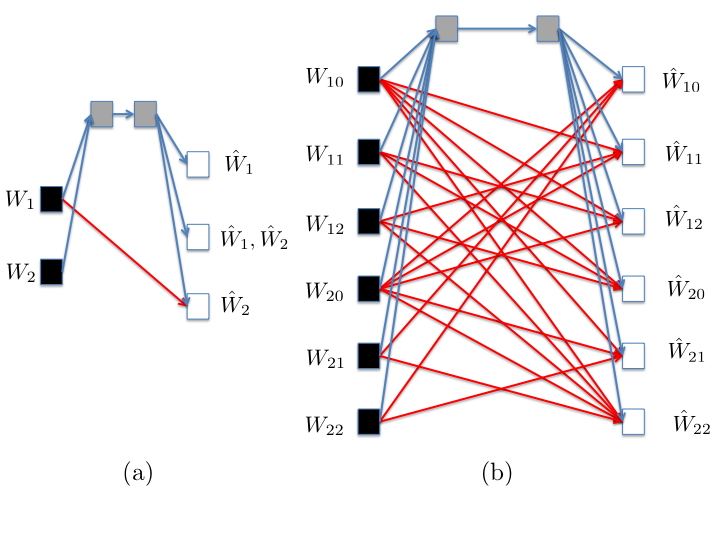}
\caption{\small (a) Groupcast index coding where $M=2$, $K=3$  and each message is desired by two destinations (b) Equivalent multiple unicast index coding problem}
\label{fig:equivalentsetting}
\end{figure}

\begin{theorem}
\label{thm:MUConstruction}
For every  groupcast index coding problem, the equivalent multiple unicast index coding problem specified by Construction \ref{construction1}, satisfies the following properties:
\begin{itemize}
\item[]{\bf Property 1:} Rate tuple $\mathcal{R}=(R_{1}, R_{2}, \ldots R_{M})$ is achievable in the groupcast setting only if the rate tuple 
\begin{equation}R_{i,j} = \left\{\begin{array}{cc} R_{i}, & j \neq 0 \\ 1-R_{i}, & j =0\end{array}\right\} \label{eq:rates}\end{equation}
is achievable in the equivalent multiple unicast problem. 
\item[]{\bf Property 2:} Rate tuple $\mathcal{R}=(R_{1}, R_{2}, \ldots R_{M})$ is achievable via {\bf linear} coding in the groupcast setting {\bf if and only if} the rate tuple 
\begin{equation}\label{eq:rates2}R_{i,j} = \left\{\begin{array}{cc} R_{i}, & j \neq 0 \\ 1-R_{i}, & j =0\end{array}\right\} \end{equation}
is achievable via  linear coding in the multiple unicast problem.
\end{itemize}
\end{theorem}
The proof of Theorem \ref{thm:MUConstruction} is presented in Section \ref{sec:insuff}, where we also argue that  it implies Theorem \ref{thm:insuff}. 
 Here, we summarize the intuition behind the proof. In the groupcast setting, a message $W_{m}$ is desired to be decoded by $L$ destinations, $D_{(m-1)L+1}, D_{(m-1)L+2},\ldots,D_{mL}$. The equivalent multiple unicast construction described formally above can be intuitively viewed as one obtained by expanding the message $W_{m}$ into $L$ independent messages $\overline{W}_{m,1}, \overline{W}_{m,2},\ldots, \overline{W}_{m,L}$ --- each of these $L$ messages desired uniquely by one destination as required for a multiple unicast index coding problem. Because of this construction, there is a correspondance between destination $D_{(m-1)L+j}$ in the groupcast setting, and destination $\overline{D}_{ m,j}$ in the equivalent multiple unicast setting. This correspondance is maintained in the antidote structure which is carried over from the groupcast setting to the multiple unicast setting. In particular, if a message $W_{l}$ is present as an antidote at destination $D_{(m-1)L+j}$ in the groupcast setting, then, destination $\overline{D}_{m,j}$ has $\overline{W}_{l,1}, \overline{W}_{l,2},\ldots,\overline{W}_{l,L}$ as antidotes in the multiple unicast setting (and vice-versa). In addition, destination $\overline{D}_{m,l}$ has as antidotes, all the messages $\overline{W}_{m,l^{'}}, l^{'} \neq l.$ This antidote structure allows all messages in the set $\{\overline{W}_{m,1}, \overline{W}_{m,2},\ldots,\overline{W}_{m,L}\}$ to occupy the same ``space'' in a linear coding scheme, since a destination that desires any one of these messages,  has all the other messages in this set as antidotes. In addition, because the antidote structure from the groupcast setting is carried over to the multiple unicast setting, for any achievable scheme in the groupcast setting, the messages $\overline{W}_{m,1}, \overline{W}_{m,2},\ldots, \overline{W}_{m,L}$ can occupy (i.e. align in) the same ``space'' occupied by $W_{m}$ in the groupcast setting. This means that any achievable scheme in the groupcast setting naturally translates to the multiple unicast setting. For the converse establishing equivalence (in Property 2), we also need to show that a scheme achieving rate $R_{m}$ for messages $\overline{W}_{m,1},\ldots,\overline{W}_{m,L}$ in the multiple unicast setting can be translated to a scheme achieving rate $R_{m}$ for message $W_{m}$ in the groupcast setting. To ensure this, we use the \emph{auxilliary} destination $\overline{D}_{m,0}$ in the multiple unicast construction. The message $\overline{W}_{m,0}$ desired by this auxiliary destination is provided as an antidote to all the other destinations and therefore does not affect achievability of rate $R_{m}$ for any other destination. Destination $\overline{D}_{m,0}$ has as antidotes, all the messages except $\overline{W}_{m,0}, \overline{W}_{m,1},\ldots, \overline{W}_{m,L}.$ This means that the space occupied by $\overline{W}_{m,0}$ has to be linearly independent of $\overline{W}_{m,1}, \overline{W}_{m,2},\ldots \overline{W}_{m,L}$. Now, if this message $\overline{W}_{m,0}$ has a rate of $1-R_{m},$ then the interfering messages faced by this destination -- $\overline{W}_{m,1}, \overline{W}_{m,2}, \ldots, \overline{W}_{m,L}$ -- have to \emph{together} occupy a space of dimension $nR_{m}$. This implies that if each of these interfering messages have to achieve a rate $R_{m},$ they have to align (nearly) perfectly. This aligned space can be used in the groupcast setting to encode $W_{m}$ at rate $R_{m}$ enabling translation of achievable scheme from the multiple unicast setting to the groupcast setting.  Equivalence for non-linear schemes is based on random coding arguments and superposition coding for auxiliary messages, according to the detailed proof presented in Section \ref{sec:insuff}.



\subsection{Feasibility of Symmetric Rate $\frac{1}{L+1}$ when $\forall k,|\mathcal{W}_k|\geq L$}


Given that  $\forall k,|\mathcal{W}_k|\geq L$, without loss of generality we can assume that  $|\mathcal{W}_k|=L$, $\forall~k\in\mathcal{K}$, i.e., each destination is interested in decoding exactly $L$ distinct messages, i.e.,  $ \mathcal{W}_k=\{W_{k_1},W_{k_2},\ldots,W_{k_L}\},~\forall k_i\in \mathcal{M}$. This is because destinations that wish to decode more than $L$ messages can be split into multiple destinations with the same set of antidotes, that each wish to decode a subset of size $L$ of the original messages, such that the union of these subsets is the original set of desired messages.


If $M=L$ or $M=L+1$, it is easy to achieve rate $\frac{1}{L+1}$ per message by sending each message separately at each time and each destination can achieve rate $\frac{1}{L}$ or $\frac{1}{L+1}$, respectively, which are both greater than or equal to $\frac{1}{L+1}$ and hence rate $\frac{1}{L+1}$ per message is feasible. 
If $M >L+1$, intuitively, since all desired messages must pass through the bottleneck link of capacity 1 and all messages must
 simultaneously achieve rate $\frac{1}{L+1}$ each, then some overlap of signal dimensions within the bottleneck
symbol $S^n$ is unavoidable. The interfering messages that are available through antidote links can be subtracted. The desired signals consume a fraction $\frac{L}{L+1}$ of the capacity of the bottleneck link, which must be free from interference. This leaves only the remaining $\frac{1}{L+1}$ of the signal space for interference within which all interfering messages, each of which carries rate $\frac{1}{L+1}$, should overlap nearly perfectly. The intuitive explanation is  formalized  for all possible coding schemes through a Shannon theoretic framework.

The following terminology is introduced specifically for the setting where each message wants to
achieve rate $\frac{1}{L+1}$.
\begin{itemize}
\item \textbf{Alignment Relation:} We define a relation $W_i\overset{k}\leftrightarrow W_j$ as follows. $W_i\overset{k}\leftrightarrow W_j$ iff $W_i \notin \mathcal{A}_k$, $W_i \notin \mathcal{W}_k$, $W_j \notin \mathcal{A}_k$ and $W_j \notin \mathcal{W}_k$ for $k\in\mathcal{K}$ and distinct indices $i$, $j\in \mathcal{M}$.  In the $\frac{1}{L+1}$ rate feasibility problem, the relation $W_i\overset{k}\leftrightarrow W_j$ represents the understanding that $W_i$ and $W_j$ must align (into $\frac{1}{L+1}$ of the signal space within the bottleneck link in order to 
leave the remaining $\frac{L}{L+1}$ of the signal space for $\mathcal{W}_k$). We may occasionally use the notation $W_i\leftrightarrow W_j$ when the identity of the destination is not important.
\item \textbf{Alignment Subsets:} The set of messages $\mathcal{W}$ is partitioned into alignment subsets, created as follows. If $W_i\leftrightarrow W_j$ , then both $W_i$, $W_j$ belong to the same alignment
subset. Further, if $W_i\leftrightarrow W_j$ and $W_j\leftrightarrow W_m$ then $W_i$, $W_j$, $W_m$ all belong to the same alignment subset. For the $\frac{1}{L+1}$ rate feasibility problem, we expect that the messages within an alignment subset will need to align almost perfectly within the
bottleneck link signal space.
\end{itemize}
As an example, consider an index coding problem shown in Fig. \ref{fig:feasible}, where $\mathcal{W}_1=\{W_1,W_2\}, \mathcal{W}_2=\{W_1,W_3\}, \mathcal{W}_3=\{W_2,W_4\}, \mathcal{A}_1=\O, \mathcal{A}_2=\{W_4\}, \mathcal{A}_3=\{W_3\}$. In this example, the alignment relations are $W_3\overset{1}\leftrightarrow W_4$ and the alignment subsets are $\{W_1\},\{W_2\},\{W_3,W_4\}$ and the rate $\frac{1}{3}$ per message is feasible.

\begin{figure}[!h] \centering
\includegraphics[width=3in]{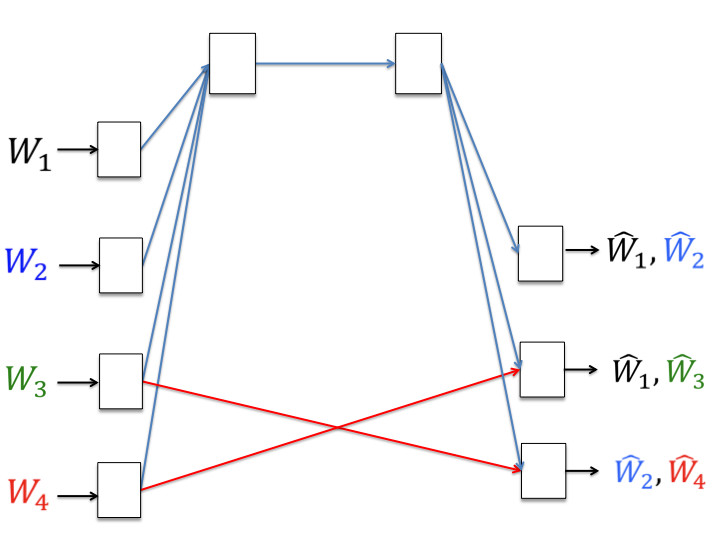}
\caption{Groupcast index coding problem with $M=4, K=3$ and $L=2$ where the rate $\frac{1}{3}$ per message is feasible}
\label{fig:feasible}
\end{figure}

Based on this definition for the alignment subset, we have the following theorem for the $\frac{1}{L+1}$ rate feasibility.
\begin{theorem}
\label{thm:HRF}
The rate tuple $\mathcal{R}$ with $R_1 = R_2 = ... = R_M= \frac{1}{L+1}$ is not achievable in the multiple groupcast index coding problem where  $|\mathcal{W}_k|=L$, $\forall~k \in \mathcal{K}$, if and only if there exist distinct indices $i, j \in \mathcal{M}$ such that $W_i$, $W_j$ belong to the
same alignment subset and $W_j\in \mathcal{W}_k$ and $W_i\notin \mathcal{A}_k$ for $k\in \mathcal{K}$
\end{theorem}

The statement of Theorem \ref{thm:HRF} is intuitively interpreted as follows. If messages $W_i$, $W_j$ belong to the same alignment set then they should overlap almost perfectly. Because of their overlap, it is not possible to recover one of them unless the other message is available as an antidote. If such an antidote is not available then the achievability of $\frac{1}{L+1}$ rate for every message becomes infeasible. Note that the feasibility condition refers to all possible coding schemes and not just linear coding schemes.
The proof of Theorem \ref{thm:HRF} is presented in Section \ref{sec:HRF}. 

As a special case of Theorem \ref{thm:HRF}, for $L=1$, we recover the feasibility condition for achievability of rate half per message, previously obtained by Blasiak et al. in \cite{Blasiak_Kleinberg_Lubetzky}. 
\begin{corollary}\label{cor:HRF}
The rate tuple $\mathcal{R}$ with $R_1 = R_2 = ... = R_M= \frac{1}{2}$ is not achievable in the multiple groupcast index coding problem where  each destination desires only one message, if and only if there exist distinct indices $i, j \in \mathcal{K}$ such that $W_i$, $W_j$ belong to the same alignment subset and $\{W_j\}= \mathcal{W}_k$ and $W_i\notin \mathcal{A}_k$ for $k\in \mathcal{K}$\end{corollary}

The rate-half feasibility condition of Corollary \ref{cor:HRF} has previously been presented by Blasiak et al. in \cite{Blasiak_Kleinberg_Lubetzky} using graph theoretic terminology which involves notions such as  almost-alternating cycles and graph-compatible functions. While the two results are essentially identical, viewing the problem through the lens of interference alignment allows  a much more intuitively transparent statement in terms of alignment subsets. 

While the feasibility conditions are not limited to linear schemes, remarkably, linear coding is sufficient for achievability for feasible settings. The following theorem further elaborates on the  linear coding scheme and the required field size.

\begin{theorem}\label{thm:fieldsize}
Whenever rate half is feasible according to Corollary \ref{cor:HRF}, it can be achieved through:
\begin{enumerate}
\item scalar linear coding over 2 channel uses if the finite field is large enough. Specifically, $|\mathbb{F}| \geq Z$.
\item vector linear coding over a given finite field $\mathbb{F}$ if the number of channel uses $n$ is large enough. Specifically, $n/2\geq\log_{|\mathbb{F}|}Z$ and $n$ is even.
\end{enumerate}
where $Z$ is the number of alignment subsets.
\end{theorem} 
\proof From the detailed proof of achievability of Theorem \ref{thm:HRF} presented in Section \ref{sec:achievable}, it is evident that what is needed for half-rate achievability is simply a one-to-one mapping, from each alignment subset, to an $n/2$-dimensional  sub-space of an $n$-dimensional vector space, such that the subspaces assigned to any two alignment subsets are non-intersecting. The number of pairwise non-intersecting $M_t$-dimensional subspaces of an $m$-dimensional vector space over a field $\mathbb{F}$ is shown by \cite{Oggier_Sloane_Diggavi} to be $\frac{|\mathbb{F}|^m-1}{|\mathbb{F}|^{M_t}-1}$ whenever $M_t$ divides $m$. For rate-half achievability, we have $m=n$, $M_t=n/2$, and the result of Theorem \ref{thm:fieldsize} follows.

Note that while it clearly takes only polynomial complexity to identify the alignment subsets (also pointed out by Blasiak et al. \cite{Blasiak_Kleinberg_Lubetzky}), finding the \emph{minimum} possible number of alignment subsets is much more challenging. This is because it is generally possible to further consolidate alignment subsets as long as alignment constraints are not violated. In other words, two alignment subsets may be combined if there does not exist a message in either subset that cannot be aligned with a message in the other subset. Recall that two messages cannot be aligned if one of them is desired at any destination that does not have the other message as an antidote. Finding the minimum number of alignment subsets for a feasible problem is, however, NP-hard. As a consequence, deciding the feasibility of e.g., rate half, when both the field  $\mathbb{F}$ and the number of channel uses $n$ is fixed, is NP-complete. This is noted specifically in the context of rate-half feasibility over a binary field and 2 channel uses, by Dau et al. in \cite{Dau_Skachek_Chee}. On the other hand, as noted above, if either the field size or the number of channel uses can be chosen to be large enough, then there is no need to consolidate the number of alignment sets, and determining the feasibility of rate half involves only polynomial complexity.

\subsection{Symmetric Instances of the Multiple Unicast Index Coding Problem}

As with most multiuser capacity problems, part of the difficulty of the index coding problem lies in the potentially unlimited number of parameters in the number of users and the desired and antidote message sets for each user, which makes a systematic study difficult. In order to limit the number of parameters while still covering broad classes of index coding problems, in this section we study symmetric instances of the multiple unicast index coding problem, e.g., where relative to its own position, each destination is associated with the same set of desired and antidote messages.

Solutions to the multiple unicast index coding problem have been found for a variety of symmetric settings. Since here $\mathcal{M}=\mathcal{K}$ and  destination $k$ desires only message $W_k$, the problem can be represented by a directed graph $G$ on the vertex set $\mathcal{M}$, in which a vertex $i$ is connected to a vertex $j$ if and only if the destination $D_i$ knows $W_j$. Bar-Yossef et al. in \cite{Yossef_Birk_Jayram_Kol_Trans} found the optimal symmetric rate for directed acyclic graphs, perfect graphs, odd holes (undirected odd-length cycles of length at least 5) and odd anti-holes (complements of odd holes).
Also, Blasiak et al. in \cite{Blasiak_Kleinberg_Lubetzky} found the capacity per message of the following symmetric index coding instances (the parameters will become clear later on in this section).
\begin{itemize}
\item Neighboring antidotes where $D=U$, showing that the sum capacity is $ \frac{U+1}{K}$ per message. 
\item  $\mathcal{A}_k=\{W_{k+1},W_{k+K/2}\}$ for even $K$, showing that the sum capacity is $\frac{2}{K}$ per message.
\item Neighboring interference for arbitrary $K$ where $U=D=1$ and showing that the sum capacity is $\frac{\lfloor K/2\rfloor}{K}$ per message.
\end{itemize}

Berliner and Langberg in \cite{Berliner_Langberg} characterized the solution of index coding problems with outerplanar side information graphs in terms of the clique cover size of the information graph where the encoding functions are (scalar) linear. Ong et.al in \cite{Ong_Ho} defined {\em uniprior} index coding problems as the case where  $\mathcal{A}_i \cap \mathcal{A}_j=\O~\text{for}~i\neq j$ and { \em single uniprior} as the case where $\mid \mathcal{A}_i \mid=1,~ \forall i \in \mathcal{K}$. They derived the optimal symmetric rate for all single uniprior index coding problems.

In the following subsections,  we present our capacity results for various symmetric classes of the  index coding problem, mostly inspired by the natural settings for the CBIA problem.

%

\subsubsection{\bf Neighboring antidotes} \label{sec:neighboringantidotes}
Consider a symmetric multiple unicast index coding problem where each destination has a total of $U+D=A<K$ antidotes, corresponding to the $U$ messages before (``up" from) and $D$ messages after (``down" from) its desired message. For this setting, we state the index coding capacity in the following theorem. 
\begin{theorem}\label{thm:sym1}
The capacity of the index coding problem with $M=K<\infty$,  (all subscripts modulo $K$)
\begin{eqnarray}
D_k&=& \left(\{W_k\}, \{W_{k-U}, W_{k-U+1}, \cdots, W_{k-1}\}\cup\{W_{k+1}, W_{k+2}, \cdots, W_{k+D}\}\right)
\end{eqnarray}
and 
\begin{eqnarray}
&&U, D\in\mathbb{Z}\\
&&0\leq U\leq D\\
&&U+D=A<K
\end{eqnarray}
 is
\begin{eqnarray}
C = \left\{\begin{array}{lr}
1,&A=K-1\\
\frac{U+1}{K-A+2U},&A\leq K-2
\end{array}
\right.
\end{eqnarray}
per message.
\end{theorem}

As an example, consider the $K=5$ user setting with $A=2$ antidotes, $U=D=1$, where the capacity is $2/5$ per message. Incidentally, example 2 in Section \ref{sec:example2} also has $K=5$ users, $A=2$ antidotes, and has the same capacity of $2/5$ per message. In general, however, for a fixed number of users, $K$, and a fixed number of total antidotes, $A$, the capacity per message depends on the relative position of antidotes. Notably, for  fixed $K, A$ values and neighboring antidotes as considered in Theorem \ref{thm:sym1}, the capacity improves as the number of antidotes on either side becomes more evenly distributed, i.e., as $D-U$ becomes smaller. The best case setting, i.e., the setting with the highest capacity, is when the antidotes are symmetrically distributed on both sides, e.g., $A=2U=2D$, (for even $A$) which leads to a capacity of $\frac{A+2}{2K}$ per message. The worst case setting is when the antidotes are all on the same side, i.e., $A=D, U=0$, which leads to a capacity of $\frac{1}{K-A}$ per message. No interference alignment is needed in the latter case.

The achievability scheme in general is a vector linear coding scheme with one-to-one alignments, where each message is sent through $U+1$ scalar symbols over a $U+1$ dimensional signal space. Adjacent messages overlap in $U$ dimensions. Because of this interference alignment, at any receiver, the total number of signal dimensions occupied by the $K-A-1$ interfering messages is equal to $U+(K-A-1)$. The $U+1$ dimensional desired signal space  is chosen to not have an intersection with the interference space, so that the dimension of the total space, i.e., the number of channel uses $n$ equals $U+(K-A-1)+U+1=K-A+2U$ and the capacity (normalized by the number of channel uses) is $(U+1)/(K-A+2U)$ per message. The details of the achievability proof and the converse are presented in Section \ref{sec:sym1proof}.

\subsubsection{\bf Neighboring interference}
Consider the following CBIA setting. We have a locally connected network with $M=K=\infty$ where each receiver $k$ has only one desired message $W_{k}$ from its corresponding (base station) transmitter $k$. There are totally $U + D+1$ transmitters with non-zero channel coefficients to receiver $k$. One of them is the desired transmitter $k$, $U$ of them are  the transmitters  
before (up from) and $D$ of them are the transmitters after (down from) the desired transmitter, as illustrated in Fig. \ref{fig:interferingneighbors}. Basically, what is sent over the $U+D$ neighboring links constitutes interference for receiver $k$. The  index coding problem  for this locally connected network has the antidote graph that is the complement of the connectivity graph shown in Fig. \ref{fig:interferingneighbors}.

\begin{figure}[!h] \centering
\includegraphics[width=4in]{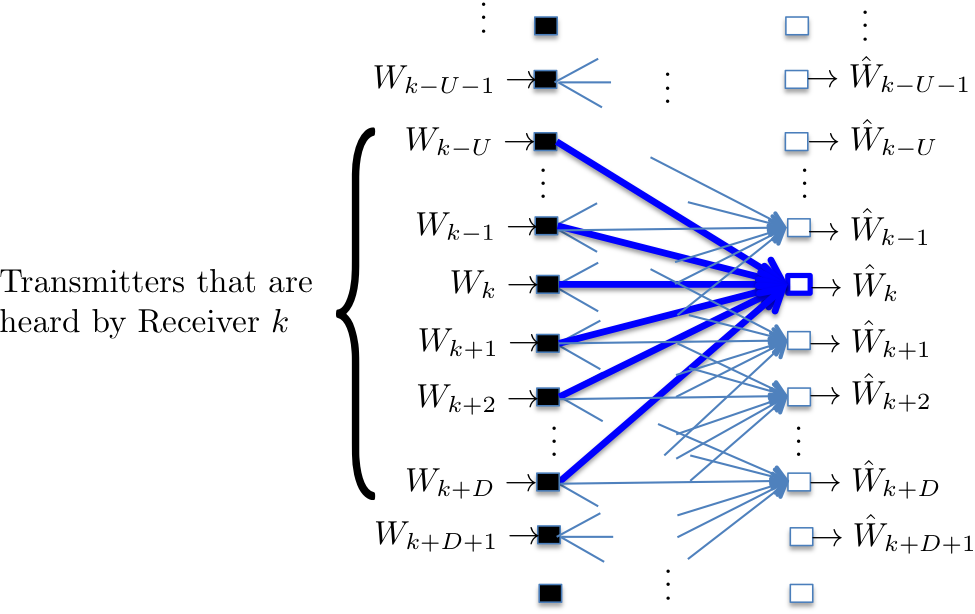}
\caption{\small CBIA setting: Locally connected network representation for multiple unicast with neighboring interference. The corresponding index coding problem has the complementary antidote graph, i.e., there is an antidote link for each transmitter receiver pair that are not connected above, and there is no antidote link for each transmitter receiver pair that are connected in the picture shown above.}
\label{fig:interferingneighbors}
\end{figure}

For such a network, we have the following:
\begin{theorem}\label{thm:sym2}
The capacity of the index coding problem associated with Fig. \ref{fig:interferingneighbors}, where $M=K=\infty$,  
\begin{eqnarray}
D_k&=& \left(\{W_k\}, \{W_{k-U,k-U+1,\cdots,k-1,k,k+1,k+2,\cdots,k+D}^c\}\right)
\end{eqnarray}
and 
\begin{eqnarray}
&&U, D\in\mathbb{Z}\\
&&0\leq U\leq D\\
\end{eqnarray}
 is
\begin{eqnarray}
C = \frac{1}{D+1}
\end{eqnarray}
per message.
\end{theorem}
Evidently, if each destination has totally $L=U+D$ missing antidotes corresponding to the $U$ messages before (ÒupÓ from) and $D$ messages after (ÒdownÓ from) its desired message, the best case setting, i.e., the setting with the highest capacity, is again when the missing antidotes are symmetrically distributed on both sides, i.e., $U=\lfloor\frac{L}{2}\rfloor$ and $D=\lceil\frac{L}{2}\rceil$, which leads to a capacity of $1/(\lceil\frac{L}{2}\rceil +1)$ per message. The worst case setting is when the missing antidotes are all on the same side, i.e., $D = L$, $U = 0$, which leads to a capacity of $\frac{1}{L+1}$ per message. No interference alignment is needed in the latter case.

The achievability scheme in general is a scalar linear coding scheme with one-to-one alignments, where
each message is sent through one scalar symbol over a $D + 1$ dimensional signal space. $U$ messages before (ÒupÓ from) each desired message are aligned with the last $U$ messages among $D$ messages after (ÒdownÓ from) that desired message, respectively. Because of this interference alignment, at any receiver, the total number of
signal dimensions occupied by the $D+U$ interfering messages is equal to $D$. The one dimensional desired signal is chosen to not have intersection with the interfering signal space and over $D+1$ channel uses, it is resolvable. The details of the achievability proof and the converse are presented in Section \ref{sec:sym2proof}.

\subsubsection{ X network setting with local connectivity}\label{sec:X}
Consider the following CBIA setting. We have a locally connected network where each destination is connected to $L$ consecutive base stations.  Suppose it is the X network setting where each base station has a distinct message for each connected destination. Without loss of generality, we can rename the destinations such that destination $i$ is connected to sources $i,i+1,\ldots,i+L-1$ as depicted in Figure \ref{fig:Xnetwork}. The index coding problem for this locally connected network 
has the antidote graph that is the complement of the connectivity graph shown in Fig. \ref{fig:Xnetwork}. For such a network, we have the following:
\begin{theorem}\label{thm:sym3}
The capacity of the symmetric index coding problem with $M=KL$ and $M,K\rightarrow\infty$, where
\begin{eqnarray}
\mathcal{W}_k&=& \{W_{kL,kL+L-1, (k+1)L+L-2,\cdots,(k+i)L+L-i-1,\cdots, (k+L-2)L+1}\}, \nonumber\\
\mathcal{A}_k&=&\{W_{(k-1)L+1:(k-1)L+L-1,\cdots,(k+i)L+1:(k+i)L+L-i-2,(k+i)L+L-i:(k+i)L+L,\cdots,(k+L-2)L+2:(k+L-2)L+L}\cup\mathcal{W}_k\}^c\nonumber
\end{eqnarray}
 is
\begin{eqnarray}
C = \frac{2}{L(L+1)}
\end{eqnarray}
per message.
\end{theorem}
\begin{figure}[!t] \centering
\includegraphics[width=4in]{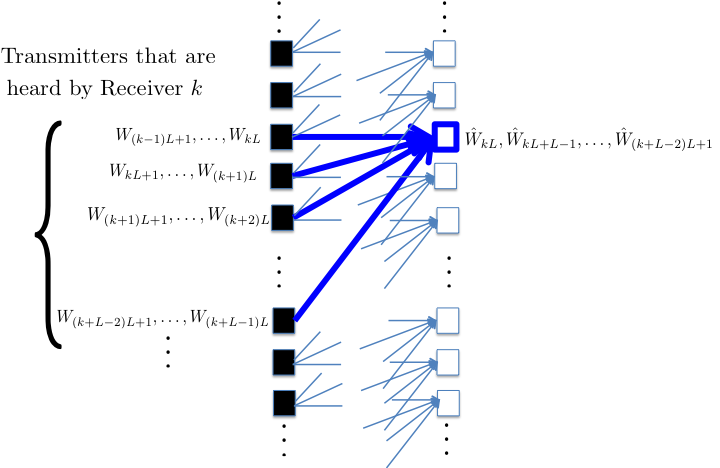}
\caption{\small CBIA setting: X network setting with local connectivity representation. The corresponding index coding problem has the complementary antidote graph, i.e., there is an antidote link for each transmitter receiver pair that are not connected above, and there is no antidote link for each transmitter receiver pair that are connected in the picture shown above.}
\label{fig:Xnetwork}
\end{figure}
The achievability scheme in general is a scalar linear coding scheme with one-to-one alignments, where
each message is sent through one scalar symbol over a $\frac{L(L+1)}{2}$ dimensional signal space. $L^2-L$ interfering messages are aligned into $\frac{L(L+1)}{2}-L$ dimensions using the pattern shown in (\ref{eqn:pattern}) where this pattern is repeated for each group of $L+1$ consecutive transmitters for the shown CBIA setting in Figure \ref{fig:Xnetwork}. The $L$ desired symbols are chosen to not have intersection with the interfering signal space and over $\frac{L(L+1)}{2}$ channel uses, they are resolvable. The details of the achievability proof and the converse are presented in Section \ref{sec:Xproof}.

 \section{Proof of Theorem \ref{thm:insuff}: Insufficiency of Linear  Codes for Multiple Unicast Index Coding}\label{sec:insuff}
 \noindent{\it Notation:}
We use the overline notation to describe quantities in the equivalent multiple unicast setting. For example, for a given achievable scheme over the multiple unicast setting $\overline{\mathcal{S}},\overline{n}, \overline{\mathcal{R}}$ represent the alphabet, the number of channel uses, and the rate vector respectively. The codeword and the encoding function for this achievable scheme are denoted as $\overline{{S}}^{n}$ and $\overline{f}(\overline{W}_{1,0},\overline{W}_{1,1},\ldots,\overline{W}_{M,L})$ respectively.

\subsection{Proof of Theorem \ref{thm:MUConstruction} - Property 1}
The rate $(R_1, R_2, \ldots, R_M)$ is achievable in the groupcast setting. This means that, for every given $\delta>0$ there exists $n$ and an encoding function $f(W_1, W_2, \ldots, W_M)$  such that over $n$ channel uses, destination $D_m$ can decode message $W_{\lceil m/L \rceil}$ of rate $R_{\lceil m/L \rceil} - \delta_{\lceil m/L \rceil}$ with a probability of error \emph{equal}\footnote{This is because, as noted in Section \ref{sec:Problem Formulation} the $\epsilon$-error capacity is the same as the zero-error capacity in the index coding problem.} to $0$ for some $\delta_{\lceil m/L \rceil} \leq \delta$. Without loss of generality, assume that the alphabet for the achievable scheme is $\mathcal{S} = \{0,1\}.$ Now, we turn to the equivalent  multiple unicast setting. In this setting, we will show that for any $\epsilon > 0$ there exists an achievable encoding scheme (over the same alphabet $\overline{\mathcal{S}} = \mathcal{S}= \{0,1\}$) with rates as in (\ref{eq:rates}) such that the probability of error for each message is smaller than $\epsilon$. The encoding scheme for achievability can be described in the following three steps.
\begin{itemize}
\item[Step 1:] Re-use $f$ to encode $\overline{W}_{i,j}, j \neq 0,i=1,2,\ldots,M$, and 
\item[Step 2:] form a random codebook for messages $\overline{W}_{i,0}, i=1,2,\ldots,M$, and then, 
\item[Step 3:] superpose the $M+1$ codewords formed in Step 1 and Step 2.  
\end{itemize}
We will show that this approach achieves the desired rates in the multiple unicast setting. 
For the achievable scheme, we set the block length to be $\overline{n}=n n_0$ where $n_0$ is chosen to be sufficiently large for purposes that will be described. For this blocklength, we have $|\overline{\mathcal{W}}_{i,j}| = 2^{n n_0 (R_{i}-\delta_i)}.$ For Step 1 above, we first represent $\overline{W}_{i,j}$ as  $n n_0 (R_{i}-\delta_i)$ bits, and then split it into $n_0$ blocks of length $n(R_i - \delta_i).$ Thus, we have 
\begin{equation} \overline{W}_{i,j} = (\overline{W}_{i,j}^{[1]},\overline{W}_{i,j}^{[2]},\ldots, \overline{W}_{i,j}^{[n_0]})\label{eq:message_splitting1}\end{equation}
where $\overline{W}_{i,j}^{[k]}$ is a bit vector of length $n(R_i - \delta_i)$. Now, we set 
$$\overline{W}_{i}^{[m]} = \sum_{j=1}^{L}\overline{W}_{i,j}^{[m]}$$
where the sum above is a bit-wise XOR of all the messages. Note that $\overline{W}_{i}^{[m]}$ is a bit-vector of length $n(R_i - \delta_i)$. Now, we are ready to form the codeword of Step 1, denoted as $\overline{S}_1^{n n_0} \in \{0,1\}^{n n_0}$.
\begin{equation}
\label{eq:step1}\overline{{S}}_1^{n n_0} = \left[\begin{array}{c}{f}(\overline{W}_1^{[1]}, \overline{W}_{2}^{[1]},\ldots,\overline{W}_{M}^{[1]})\\ 
{f}(\overline{W}_1^{[2]}, \overline{W}_{2}^{[2]},\ldots,\overline{W}_{M}^{[2]})\\
\vdots\\
{f}(\overline{W}_1^{[n_0]}, \overline{W}_{2}^{[n_0]},\ldots,\overline{W}_{M}^{[n_0]})
\end{array}\right] \define \left[\begin{array}{c} \overline{S}_{1}^{[1]}\\ \overline{S}_{1}^{[2]} \\ \vdots \\ \overline{S}_{1}^{[n_0]} \end{array} \right] \end{equation}
where $\overline{S}_1^{[m]}$ is an $n$ bit column vector. 

Now, we proceed to Step 2. Message $\overline{W}_{k,0},$ is encoded in $\{0,1\}^{n n_0}$ by treating the codeword as $n_0$ block-symbols, each of size $n,$ and using i.i.d. random coding for each block. In particular, we generate a random codebook of length $n_0$ over alphabet ${\{0,1\}}^{n},$ where each of the $n_0$ entries are chosen uniformly over all possible $n$-bit vectors, and independently of each other. We denote by $\overline{{S}}_{i,0}^{n n_0},$ the codeword corresponding to $\overline{W}_{i,0},$ and set 
\begin{equation} \overline{{S}}_{0}^{n n_0} = \sum_{i=1}^{M} \overline{{S}}_{i,0}^{n n_0},\label{eq:step2}\end{equation}
where  $\overline{{S}}_{0}, \overline{{S}}_{i,0} \in \{0,1\}$.

We are now ready for Step 3. The codeword sent on the multiple unicast channel is set as 
$$\overline{{S}}^{nn_0} = \overline{{S}}_{1}^{nn_0} + \overline{{S}}_{0}^{nn_0},$$
where $\overline{S}_{1}^{nn_0}$ and $\overline{S}_{0}^{nn_0}$ are chosen from (\ref{eq:step1}) and (\ref{eq:step2}), and the addition is bit-wise XOR of the corresponding entries of the vector.

Now we need to show that each message can be decoded with a probability of error smaller than $\epsilon$. 
First, consider message $\overline{W}_{i,j}, j \neq 0$ to be decoded at destination $\overline{D}_{i,j}$. Note that this destination has, as antidotes, $\mathcal{W}_{.,0}$ because of (\ref{eq:equivalent_antidotes}). Therefore, it can compute $\overline{S}_{0}^{nn_0}$ and subtract it from the received codeword to obtain $f(\overline{W}_1^{[m]}, \overline{W}_{2}^{[m]},\ldots,\overline{W}_{M}^{[m]})$ for $m=1,2,\ldots,n_0$. For all $l$ such that $W_{l}\in \mathcal{A}_{(i-1)L+j},$ this destination has, as antidote, $\{\overline{W}_{l,j}:j=0,1,2,\ldots, L\}$ and can therefore compute $\overline{W}_{l}$. Now, destination $\overline{D}_{i,j}$ resembles destination $D_{(i-1)L+j}$ in the groupcast setting. In particular, the decoding strategy used by $D_{(i-1)L+j}$ in the groupcast setting can be used by $\overline{D}_{i,j}$ in the equivalent multiple unicast setting to decode $\overline{W}_{i}^{[m]},$ free of error.  Also, note that the destination has $\{\overline{\mathcal{W}}_{i,l}: l \neq j\}$ as antidotes. Therefore, it can decode $\overline{W}_{i,j}^{[m]}$ from $\overline{W}_{i}^{[m]}.$

Now, consider destination $\overline{D}_{i,0}$. This destination has $\overline{\mathcal{W}} - \overline{\mathcal{W}}_{i,.}$ as antidotes. Therefore, it can compute $\overline{S}_{k,0}^{nn_0}, k \neq  i$ and cancel the effect of this to obtain $\overline{S}_{1}^{nn_0}+ \overline{S}_{i,0}^{nn_0}.$ Intuitively, the signal obtained after this cancellation can be interpreted as an additive noise channel over alphabet $\{0,1\}^{n},$ where the noise vector faced by the $m$th channel use is $\overline{S}_{1}^{[m]},m=1,2,\ldots,n_0.$ Using the random coding argument, for sufficiently large $n_0$, $\overline{W}_{i,0}$ can be decoded by this destination with probability of error smaller than $\epsilon,$ for any $\epsilon > 0$ for any rate $R_{i,0}$ up to the mutual information


\begin{eqnarray}
\label{eq:part1proof1}
 \frac{1}{nn_o}I\left(\overline{S}_{1}^{nn_0}+ \overline{S}_{i,0}^{nn_0}, \overline{\mathcal{A}}_{i,0}; \overline{W}_{i,0}\right)
&=& \frac{1}{nn_o}I\left(\overline{S}_{1}^{nn_0}+ \overline{S}_{i,0}^{nn_0} ; \overline{W}_{i,0} | \overline{\mathcal{A}}_{i,0}\right)\\
&=& \frac{1}{nn_o}H(\overline{S}_1^{nn_0}+ \overline{S}_{i,0}^{nn_0}| \overline{\mathcal{A}}_{i,0}) - \frac{1}{nn_o}H(\overline{S}_1^{nn_0}+ \overline{S}_{i,0}^{nn_0}|\overline{\mathcal{A}}_{i,0}, \overline{{W}}_{i,0})\\
&=&1 - \frac{1}{nn_o}H(\overline{S}_1^{nn_0}|\overline{\mathcal{A}}_{i,0},\overline{{W}}_{i,0})\label{eq:(a)}\\
&=& 1 - (R_i-\delta_i)
\label{eq:part1proof2}
\end{eqnarray}
where, (\ref{eq:(a)}) comes from the fact that $\overline{S}_{i,0}^{nn_0}$ is uniformly distributed over $\{0,1\}^{nn_0}$ through the random coding construction (and is independent of $\overline{\mathcal{A}}_{i,0}$). The final  bound comes from the fact that given $(\overline{\mathcal{A}}_{i,0},\overline{{W}}_{i,0})$, $S_1^{nn_o}$ is an invertible function of $\overline{W}_i$ and $H(\overline{W}_i)=nn_o(R_i-\delta_i)$.

%

\subsection{Proof of Theorem \ref{thm:MUConstruction} - Property 2}
Here, we first consider the case where rate $(R_1, R_2, \ldots, R_K)$ is achievable in the groupcast setting via a \emph{linear} coding based achievable scheme. This means that for any $\delta>0,$ there exists $n,\mathbf{U}_{l}, \mathbf{V}_{k}, l=1,2,\ldots,ML, k=1,2,\ldots, M$, such that
\begin{eqnarray}
&& {\bf U}_{l}{\bf V}_{k}=0, ~~~~\forall l=1,2,\ldots,{ML},   W_{k}\notin\mathcal{A}_{l}, k \neq \lceil l/L\rceil \label{eq:groupcast_linear1}\\
&& \mbox{det}\left({\bf U}_{l}{\bf V}_{\lceil l/L \rceil}\right)\neq 0, ~~~~\forall l \in \{1,2,\ldots, ML\}, \label{eq:groupcast_linear2}
\end{eqnarray}
where $\mathbf{V}_{k}$ is a $n \times n(R_k - \delta_k)$ dimensonal matrix, and $\mathbf{U}_{l}$ is a $n(R_{\lceil{l/L}\rceil} - \delta_{\lceil{l/L}\rceil}) \times n$ matrix, where $\delta_k \leq \delta$.
The encoding scheme for this setting is 
$$ {S}^{n} = \sum_{k=1}^{M}\mathbf{V}_k {X}_{k} ,$$
where $S\in \mathbb{F}$ and $X_k$ is a $n(R_k-\delta_k)\times 1$ vector. Now, we turn to the equivalent multiple unicast setting, where we provide linear encoding achievable scheme for rates
\begin{equation}\overline{R}_{i,j} = \left\{\begin{array}{cc} R_{i}-\delta_i, & j \neq 0 \\ 1-R_{i}, & j =0\end{array}\right\}. \end{equation}
The above automatically implies achievability of rate (\ref{eq:rates2}).
 The encoding scheme for the multiple unicast setting is formed by setting $\overline{\mathbf{V}}_{i,j} = \mathbf{V}_{i}$ for $j\neq 0$ so that 
$$ \overline{S}^{n} = \sum_{i,j} \overline{\mathbf{V}}_{i,j}\overline{X}_{i,j} =  \sum_{i=1}^{M} \mathbf{V}_{i}\left(\overline{X}_{i,1}+ \overline{X}_{i,2}+\ldots+\overline{X}_{i,L}\right) + \sum_{i=1}^{M} \overline{\mathbf{V}}_{i,0} \overline{X}_{i,0}$$
where $\overline{S} \in \mathbb{F},$ $\overline{X}_{i,j}$ is a $n(R_{i}-\delta_i)\times 1$ vector for $j \neq 0$ and $n(1-R_i) \times 1$ vector for $j=0$.  The matrix $\overline{\mathbf{V}}_{i,0}$ is chosen to be a $n \times n(1-R_{i})$ full rank matrix so that 
$$\mbox{colspan}(\mathbf{V}_{i}) \cap \mbox{colspan}(\overline{\mathbf{V}}_{i,0}) = \{0\}.$$ Since $\mbox{rank}(\mathbf{V}_{i}) \leq n(R_{i} - \delta_i)$ there exists a matrix $\overline{\mathbf{V}}_{i,0}$ satisfying the above condition. This condition automatically implies that there exists $\overline{\mathbf{U}}_{i,0}$ such that 
$$\overline{\mathbf{U}}_{i,0} \mathbf{V}_{i} = 0,$$
$$\mbox{det}(\overline{\mathbf{U}}_{i,0} \overline{\mathbf{V}}_{i,0}) \neq 0.$$
Since $\{\overline{W}_{i,j}: j=1,2,\ldots,L\}$ are the only messages not present at $\overline{D}_{i,0}$ as antidotes, and all these messages are encoded using $\mathbf{V}_{i},$ the above equations imply that the destination can decode message $\overline{W}_{i,0}.$ Now, consider destination $\overline{D}_{i,j}, j \neq 0.$ We can set $\overline{\mathbf{U}}_{i,j} = \mathbf{U}_{(i-1)L+j}.$ Then, we have for $j\neq 0,$
\begin{eqnarray*} \mathbf{U}_{(i-1)L+j} \mathbf{V}_{k} &=& 0, \forall k \notin \mathcal{A}_{(i-1)L+j}, k \neq i\\
\Rightarrow  \overline{\mathbf{U}}_{i,j} \overline{\mathbf{V}}_{k,l} &=&0, \forall \overline{W}_{k,l} \notin \overline{\mathcal{A}}_{i,j}, k \neq i \end{eqnarray*}
and, because of (\ref{eq:groupcast_linear2}), we have
$$\mbox{det}\left(\overline{\mathbf{U}}_{i,j} \overline{\mathbf{V}}_{i,j}\right) = \mbox{det}\left(\mathbf{U}_{(i-1)L+j}\mathbf{V}_{i}\right) \neq 0$$
as required. This ensures achievability of rates (\ref{eq:rates2}) via linear coding in the equivalent multiple unicast setting.


Now, we consider the case where rates in (\ref{eq:rates2}) are achievable in the multiple unicast problem via linear coding. This means that for any $\delta>0$ there exist matrices $\overline{\mathbf{V}}_{i,j}, \overline{\mathbf{U}}_{i,j}, i=1,2,\ldots,M$ $j=0,1,2,\ldots,L$, such that 
\begin{equation} \overline{\mathbf{U}}_{i,j} \overline{\mathbf{V}}_{k,l} = 0, \forall k,l\mbox{ s.t. }\overline{W}_{k,l} \notin \overline{\mathcal{A}}_{i,j}, (k,l)\neq(i,j) \label{eq:MULinearAchievability1} \end{equation}
$$ \mbox{det}(\overline{\mathbf{U}}_{i,j} \overline{\mathbf{V}}_{i,j}) \neq 0$$
where $\overline{\mathbf{V}}_{i,j}$ is a $n \times n(R_{i}-\delta_{i,j})$ matrix if $j\neq 0$ and $n \times n(1-R_{i} - \delta_{i,0})$ matrix if $j = 0,$ for some $\delta_{i,j} \leq \delta,j=0,1,2,\ldots,L,i=1,2,\ldots,M$. Similarly, $\overline{\mathbf{U}}_{i,j}$ is a $n(R_i - \delta_{i,j}) \times n$ matrix if $j \neq 0$ and it is a $n(1-R_i - \delta_{i,0}) \times n$ matrix if $j=0.$ Now, note that the above implies that 
$$\overline{\mathbf{U}}_{i,0}\overline{\mathbf{V}}_{i,j} = 0, j=1,2,\ldots,L.$$
We now intend to show that $(R_1, R_2,\ldots, R_M)$ is achievable in the groupcast setting via linear coding. For the achievable scheme on the groupcast setting, we choose \begin{equation} \mathbf{V}_{i} = \langle \bigcap_{j=1}^{L} \mbox{colspan}(\overline{\mathbf{V}}_{i,j})\rangle.\label{eq:Vconstruction}\end{equation}
where, for a set of column vectors $\mathcal{V},$ the matrix denoted by $\langle \mathcal{V} \rangle$ is a full rank matrix whose columns span $\mathcal{V}.$ We now intend to show that rate $R_{i}^{'}$ is achievable for  $W_{i}$ where $nR_{i}^{'}=\mbox{rank}(\mathbf{V}_{i}).$ To show this, we need to prove the following. 
\begin{itemize}
\item[(C1)] We need to show that there exist $nR_i^{'} \times n$ decoding matrices $\mathbf{U}_{i},i=1,2,\ldots,ML$ that satisfy the desired critiera, i.e., that 
$$\mathbf{U}_{l} \mathbf{V}_{m} = 0,W_m \notin \mathcal{A}_{l}, m \neq \lceil l/L \rceil,$$ 
$$ \mbox{det}\left(\mathbf{U}_{l} \mathbf{V}_{\lceil l/L \rceil}\right) \neq 0.$$ 
\item[(C2)] We need to show that the above has an appropriate rate, i.e., we need to show that $R_{i}^{'}$ can be made arbitrarily close to $R_{i}.$ 
\end{itemize}
Instead of (C1) above, we show equivalently that there exist matrices $\mathbf{U}_{i}^{'}$ such that 
\begin{eqnarray*}
\mathbf{U}_{l}^{'} \mathbf{V}_{m} &=& 0, \forall~W_m \notin \mathcal{A}_{l}\cup \{W_{\lceil l/L\rceil}\}\\
\mbox{rank}(\mathbf{U}_{l}^{'} \mathbf{V}_{\lceil l/L \rceil}) &=& \mbox{rank}(\mathbf{V}_{\lceil l/L \rceil}),\forall~i\in\{1,2,\ldots,ML\}.
\end{eqnarray*}
For this we set \begin{equation}\label{eq:Uconstruction}\mathbf{U}^{'}_{(i-1)L+j} = \overline{\mathbf{U}}_{i,j}.\end{equation} Let $W_{m} \notin \mathcal{A}_{l}, m \neq \lceil l/L \rceil.$ Because of (\ref{eq:equivalent_antidotes}), we infer that 
\begin{eqnarray}
\overline{W}_{m,k} &\notin& \overline{\mathcal{A}}_{\lceil l/L \rceil, l- L(\lceil l/L \rceil-1)}\cup \{\overline{W}_{\lceil l/L \rceil, l- L(\lceil l/L \rceil-1)}\}, \forall k=1,2,\ldots,L\\
\Rightarrow \overline{\mathbf{U}}_{\lceil l/L \rceil, l- L(\lceil l/L \rceil-1)} \overline{\mathbf{V}}_{m, k} &=& 0, \forall k=1,2,\ldots,L\\
\Rightarrow \mathbf{U}^{'}_{l} \mathbf{V}_{m} &=& 0
\end{eqnarray}
where the final equation comes from (\ref{eq:Vconstruction}) and (\ref{eq:Uconstruction}). In particular, the final equation comes from noting that $\mbox{colspan}(\mathbf{V}_{m})\subseteq \mbox{colspan}(\mathbf{V}_{m,l}),l\neq 0.$ Similarly, we have 
\begin{eqnarray}
\mbox{det}\left(\overline{\mathbf{U}}_{\lceil l/L \rceil, l- L(\lceil l/L \rceil-1)} \overline{\mathbf{V}}_{\lceil l/L \rceil, l- L(\lceil l/L \rceil-1)} \right)&\neq& 0\\
\Rightarrow \mbox{det}\left(\mathbf{U}_{l}^{'}\overline{\mathbf{V}}_{\lceil l/L \rceil, l- L(\lceil l/L \rceil-1)}\right) &\neq& 0\\
\Rightarrow \mbox{rank}\left(\mathbf{U}_{l}^{'}\mathbf{V}_{\lceil l/L \rceil}\right) &=& \mbox{rank}\left(\mathbf{V}_{\lceil l/L \rceil}\right),
\end{eqnarray}
 where, the final equation follows because $\mbox{colspan}(\mathbf{V}_{\lceil l/L\rceil}) \subseteq \mbox{colspan}(\overline{\mathbf{V}}_{\lceil l/L\rceil,l-L(\lceil l/L\rceil -1)}).$  This completes the proof of (C1). We now need to show (C2).
To show this, we use the following Lemma 
\begin{lemma}
\label{lem:UsefulForInsuff}
Let $\mathbf{A},\mathbf{B},\mathbf{C}$ be three matrices respectively of sizes $n \times D_{A}, n \times D_{B}$ and $n \times D_{C}.$ If $\mbox{colspan}(\mathbf{A}) \cap \mbox{colspan}(\mathbf{C})=\{0\}$ and $\mbox{colspan}(\mathbf{B}) \cap \mbox{colspan}(\mathbf{C})=\{0\}$, then, 
$$\mbox{dim}\left(\mbox{colspan}(\mathbf{A})\cap \mbox{colspan}(\mathbf{B})\right) \geq \mbox{rank}(\mathbf{A})+\mbox{rank}(\mathbf{B})+\mbox{rank}(\mathbf{C})-n.$$
\end{lemma}
\begin{proof}
Let the rank of $\mathbf{A}, \mathbf{B}$ and $\mathbf{C}$ respectively be $R_{A}, R_{B}$ and $R_{C}$. Also, let
$$R_{A \cup B}\define \mbox{dim}\left(\mbox{colspan}(\langle \mbox{colspan}(\mathbf{A}) \cup \mbox{colspan}(\mathbf{B})\rangle)\right)$$  $$R_{A\cap B} \define  \mbox{dim}\left(\mbox{colspan}(\mathbf{A}) \cap \mbox{colspan}(\mathbf{B})\right).$$
where $\langle \mathcal{V} \rangle$ is a matrix whose column-span is equal to the column-span of the set of column vectors $\mathcal{V}$. We have
\begin{eqnarray*}
\mbox{colspan}\left(\langle\mbox{colspan}(\mathbf{A})\cup\mbox{colspan}(\mathbf{B})\rangle\right) \cap \mbox{colspan}(\mathbf{C}) &=& \{0\}\\
\Rightarrow R_{A \cup B} + \mbox{dim}(\mbox{colspan}(\mathbf{C})) &=& \mbox{dim}\left(\mbox{colspan}\left(\langle\mbox{colspan}(\mathbf{A})\cup\mbox{colspan}(\mathbf{B})\cup \mbox{colspan}(\mathbf{C}) \rangle\right)\right)\\
\Rightarrow R_{A \cup B} + R_{C} &\leq& n\\
\Rightarrow R_{A} + R_{B} - R_{A\cap B}+ R_{C} &\leq& n\\
\Rightarrow\mbox{dim}\left(\mbox{colspan}(\mathbf{A})\cap \mbox{colspan}(\mathbf{B})\right) = R_{A\cap B} &\geq& R_A+R_B+R_C-n
\end{eqnarray*}
\end{proof}

Using the fact that $\mbox{colspan}(\overline{\mathbf{V}}_{i,j}) \cap \mbox{colspan}(\overline{\mathbf{V}}_{i,0}) = \{0\}$  and the above lemma, we lower bound the rank of $\mathbf{V}_{i} = \langle \bigcap_{l=1}^{L} \mbox{colspan}(\overline{\mathbf{V}}_{i,l})\rangle$ as follows. 
\begin{align}
&nR_{i}^{'} = \mbox{rank}(\mathbf{V}_{i}) = \mbox{dim}\left(\bigcap_{l=1}^{L} \mbox{colspan}(\overline{\mathbf{V}}_{i,l})\right)\\
&\geq \mbox{dim}\left(\bigcap_{l=1}^{L-1} \mbox{colspan}(\overline{\mathbf{V}}_{i,l})\right)+\mbox{rank}(\overline{\mathbf{V}}_{i,L})-(n-\mbox{rank}(\overline{\mathbf{V}}_{i,0})) \\
&\geq \mbox{dim}\left(\bigcap_{l=1}^{L-2} \mbox{colspan}(\overline{\mathbf{V}}_{i,l})\right)+\sum_{l=L-1}^{L} \mbox{rank}(\overline{\mathbf{V}}_{i,l})-2(n-\mbox{rank}(\overline{\mathbf{V}}_{i,0})) \\
&\vdots\\
&\geq \mbox{dim}\left(\mbox{colspan}(\overline{\mathbf{V}}_{i,1}) \cap \mbox{colspan}(\overline{\mathbf{V}}_{i,2})\right)+\sum_{l=3}^{L} \mbox{rank}(\overline{\mathbf{V}}_{i,l})-(L-2) (n-\mbox{rank}(\overline{\mathbf{V}}_{i,0})) \\
&\geq \sum_{l=1}^{L} \mbox{rank}(\overline{\mathbf{V}}_{i,l})-(L-1) (n-\mbox{rank}(\overline{\mathbf{V}}_{i,0})) \\
&= \sum_{l=1}^{L} n(R_{i} - \delta_{i,l}) - (L-1)(n-n(1-R_{i} - \delta_{i,0}))\\
&\geq n\left(R_{i} - \sum_{i=1}^{L} \delta_{i,l} - (L-1) \delta_{i,0}\right)
\end{align}
where, the final equation above comes from the achievability of rates (\ref{eq:rates}) on the multiple unicast setting. By choosing $\delta_{i,j}$ to be arbitrarily small, the rate $R_{i}'$ can be made arbitrarily close to $R_{i}.$ This completes the proof of (C2), and hence the proof of the Theorem \ref{thm:MUConstruction}.

\subsection{Proof of Theorem \ref{thm:insuff}}
We show Theorem \ref{thm:insuff} using the results of Theorem \ref{thm:MUConstruction} and \cite{Blasiak_Kleinberg_Lubetzky_2011}. \cite{Blasiak_Kleinberg_Lubetzky_2011} implies the following.  There exists a groupcast setting where, there exists a rate vector $(R_1, R_2, \ldots, R_K)$ in its capacity region that not achievable via linear coding. Now, we consider the equivalent multiple unicast setting based on Construction \ref{construction1}. Because of Property 2 of Theorem \ref{thm:MUConstruction}, this implies that rate-tuple
$$R_{i,j} = \left\{\begin{array}{cc}R_i & j \neq 0\\ 1-R_{i} & j = 0  \end{array} \right\} $$ is not achievable via \emph{linear coding} in this multiple unicast setting. Further, because of Property 1 of Theorem \ref{thm:MUConstruction}, the achievability of $(R_1, R_2,\ldots, R_K)$ in the groupcast setting (using a non-linear scheme) implies that the above rate-tuple is achievable in the multiple unicast setting. Therefore, linear coding does not suffice for achievability in the multiple unicast setting.
\section{Proof of Theorem \ref{thm:HRF}: Feasibility of rate $\frac{1}{L+1}$ per message}\label{sec:HRF}
\subsection{Achievability}\label{sec:achievable}
We partition $\mathcal{W}=\{W_1,W_2,\ldots,W_M\}$ into alignment subsets $\mathcal{P}_{1},\mathcal{P}_{2},\ldots,\mathcal{P}_{T}$, and define the mapping ${P}(m):\mathcal{M}\rightarrow\{1,2,\cdots,T\}$ so that $W_m\in\mathcal{P}_{{P}(m)}$, $\forall m\in\mathcal{M}$. 
 
 We use a scalar linear achievable scheme. In particular, we choose $ n=L+1$ and $\mathcal{S}=\mathbb{F}_q$, where $q$ is chosen to be sufficiently large so that there exist $T$ vectors $\mathbf{V}_1,\mathbf{V}_2,\ldots,\mathbf{V}_{T}$,  in the $L+1$ dimensional vector space over $\mathbb{F}_q$, such that every $L+1$ of them are linearly independent. These are the linear coding vectors along which the aligned messages from each partition will be sent.


Destination $r$, which desires messages $\mathcal{W}_r=\{W_{r_1}, W_{r_2},\cdots, W_{r_L}\}$,  receives
\begin{eqnarray*} S^n &=& \sum_{i=1}^{M} x_i \mathbf{V}_{{P}(i)} \\
&=& \sum_{i=1}^{L} x_{r_i} \mathbf{V}_{{P}(r_i)}  + \underbrace{\sum_{i: W_i \in \mathcal{A}_{r}} x_i \mathbf{V}_{i}}_{\begin{array}{c}\mbox{Side information}\\ \mbox{at destination $r$}\end{array}} + \left(\sum_{i:W_i \notin \mathcal{A}_{r}\cup\mathcal{W}_r} x_i\right) \mathbf{V}_{\mathcal{P}_{t}}
\end{eqnarray*}
where the last equation follows because all elements of $\{i: W_i \notin \mathcal{A}_{r}\cup\mathcal{W}_r\}$ belong to the same alignment subset by definition; this subset is denoted by $\mathcal{P}_{t}.$ After cancelling the second term above, destination $r$  obtains a linear combination of $L+1$ vectors, $\mathbf{V}_{P(r_1)},\mathbf{V}_{P(r_2)},\ldots,\mathbf{V}_{P(r_L)}$ and $\mathbf{V}_{\mathcal{P}_t}$. We need to show that these $L+1$ vectors are linearly independent if the feasibility condition is satisfied. The feasibility condition implies that for a feasible system, if two messages belong to the same alignment subset, then whenever one of them is desired, the other must be available as an antidote. Since this is not true for any two of the $L+1$ vectors $\mathbf{V}_{P(r_1)},\mathbf{V}_{P(r_2)},\ldots,\mathbf{V}_{P(r_L)}$ and $\mathbf{V}_{\mathcal{P}_t}$, they must each belong to distinct partitions. Since every $L+1$ coding vectors are linearly independent, we conclude that $x_{r_1}, x_{r_2}, \ldots, x_{r_L}$ are resolvable at destination $r$ as required.


\subsection{Outer Bounds}
We start with a simple information theoretic outer bound on the index coding problem.

\begin{theorem}
\label{thm:simplebound}
 \emph{The achievable rates $(R_1,R_2,...,R_M)$ for the index coding problem where  $|\mathcal{W}_k|=L$, $\forall~k\in\{1,2,\ldots,K\}$, satisfy the
following inequalities:} 
\begin{eqnarray}
\sum_{\{i:W_i\in\mathcal{W}_k\}}R_i+\sum_{\{i:W_i\in\mathcal{W}_j\cap\left(\mathcal{W}_k\cup\mathcal{A}_k\right)^c\}}R_i\leq 1, ~~~\forall k\in\mathcal{K}
\end{eqnarray}
\end{theorem}
In other words, the sum of the rates of the messages desired at a destination $D_k$ and the rates of the interfering messages (i.e., undesired messages that are not available as antidotes to $D_k$) that are intended for any other destination $D_j$ cannot exceed 1.\\
\begin{proof}
 Consider a reliable achievable coding scheme for the network. Assume all the messages \emph{except} $\mathcal{W}_k$ and $\mathcal{W}_j\cap\left(\mathcal{W}_k\cup\mathcal{A}_k\right)^c$ are given by a genie to destination $k$ and destination $j$. Further, assume that the genie provides $\mathcal{W}_k$ to destination $j$ as well.
 Now, using the reliable  achievable scheme, destination $k$ can decode $\mathcal{W}_k$. Having decoded $\mathcal{W}_k$, in this genie aided index coding problem, destination $k$ has all the information that destination $j$ has. Therefore, destination $k$ can decode $\mathcal{W}_j\cap\left(\mathcal{W}_k\cup\mathcal{A}_k\right)^c$ as well. Since destination $i$ can decode messages $\mathcal{W}_k$ and $\mathcal{W}_j\cap\left(\mathcal{W}_k\cup\mathcal{A}_k\right)^c$,  the sum of $H(\mathcal{W}_k)$ and $H(\mathcal{W}_j\cap\left(\mathcal{W}_k\cup\mathcal{A}_k\right)^c)$ must be bounded above by the entropy of the bottleneck link $H(S^n)$. 
 \end{proof}

\emph{Remark:} A reader familiar with Carlieal's outer bound for the interference channel \cite{Carleial_int} may notice similarities with the approach in the above proof. Also note that the bound, as with all bounds in this section, is information theoretic, so it applies to all linear and non-linear coding schemes.

As an example, suppose we have an index coding problem with $M=4, K=2$ and $L=2$,  where $\mathcal{W}_1=\{W_1,W_2\},\mathcal{W}_2=\{W_3,W_4\}, \mathcal{A}_1=\O, \mathcal{A}_2=\{W_1,W_2\}$ as shown in Fig. \ref{fig:422case}. Using the above outer bound, we have $R_1+R_2+R_3+R_4\leq1$, which is tight because it is also achievable by sending each message separately in 4 consecutive time slots. 
\begin{figure}[!h] \centering
\includegraphics[width=3in]{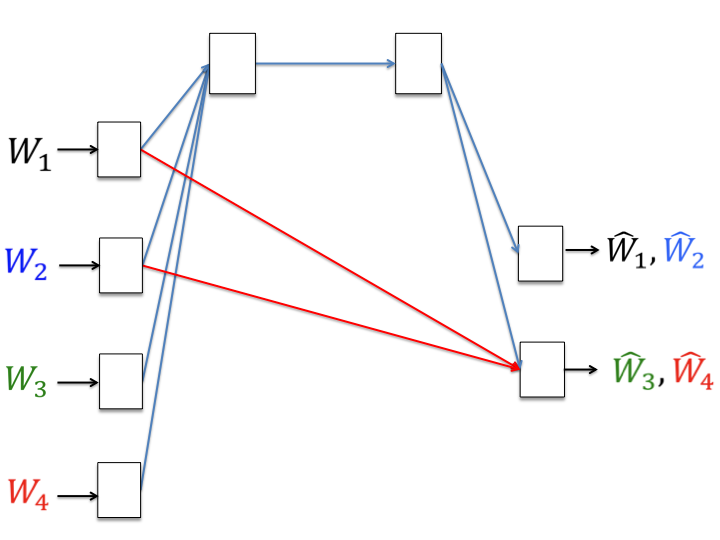}
\caption{Index coding problem with $M=4, K=2$ and $L=2$ where $R_1+R_2+R_3+R_4\leq1$}
\label{fig:422case}
\end{figure}

Next we derive  outer bounds which are valid for general index coding problems and are tight for the $\frac{1}{L+1}$ rate feasibility problem. The main result of this section is the following.
\begin{theorem}
\label{thm:outerbound1}
\emph{For any $N+1$ distinct indices $i_0, i_1,...,i_N \in \mathcal{M}$ and $N$ indices $j^1,...,j^N \in \mathcal{K},$ if}
\begin{eqnarray}
W_{i_0}\overset{j^1}\leftrightarrow W_{i_1}\overset{j^2}\leftrightarrow W_{i_2}\overset{j^3}\leftrightarrow...\overset{j^{N-1}}\leftrightarrow W_{i_{N-1}}\overset{j^N}\leftrightarrow W_{i_N}~~\text{and}~~W_{i_0} \notin \mathcal{A}_k~~\text{where}~~W_{i_N} \in \mathcal{W}_k~~\text{for}~~k \in \mathcal{K},
\end{eqnarray}
\emph{then}
\begin{eqnarray}
R_{i_0}+R_{j^1_{1:L}}+R_{i_1}+...+R_{j^N_{1:L}}+R_{i_N}\leq N
\end{eqnarray}
\end{theorem}
The converse for Theorem \ref{thm:HRF} follows from the above theorem. This converse is expressed next as a corollary to the above theorem.
\begin{corollary}
The rate tuple $\mathcal{R}$ with $R_1 = R_2 = ... = R_M= \frac{1}{L+1}$ is not achievable in a single
bottleneck network where  $|\mathcal{W}_k|=L$, $\forall~k\in\{1,2,\ldots,K\}$ if and only if there exist distinct indices $i, j \in \mathcal{M}$ such that $W_i$, $W_j$ belong to the
same alignment subset and $W_j\in \mathcal{W}_k$ and $W_i\notin \mathcal{A}_k$ for $k\in \mathcal{K}$
\end{corollary}
\begin{proof}
If $W_i, W_j$ belong to the same alignment subset, there must exist a chain of alignment relations connecting $W_i$, $W_j$ as
\begin{eqnarray}
W_i\overset{j^1}\leftrightarrow W_{i_1}\overset{j^2}\leftrightarrow W_{i_2}\overset{j^3}\leftrightarrow...\overset{j^{N-1}}\leftrightarrow W_{i_{N-1}}\overset{j^N}\leftrightarrow W_{j}
\end{eqnarray}
where $N$ is the length of the chain. If $W_j\in \mathcal{W}_k$ and $W_i\notin \mathcal{A}_k$ for $k \in \mathcal{K}$, from the result of Theorem 5 shown later in this paper we have an explicit rate bound
\begin{eqnarray}
R_i+R_{j^1_{1:L}}+R_{i_1}+R_{j^2_{1:L}}+...+R_{i_{N-1}}+R_{j^{N}_{1:L}}+R_j\leq N
\end{eqnarray}
Clearly, $R_i = R_{i_1}=R_{i_2}=\ldots = R_{i_{N-1}}=R_j = R_{j^1_1} =R_{j^1_2}=\ldots=  R_{j^N_L} = \frac{1}{L+1}$ does not satisfy the above bound. This completes the proof of the corollary.
\end{proof}

\subsubsection{Proof of Theorem \ref{thm:outerbound1}}
Theorem \ref{thm:outerbound1} affords a simple proof for $N=1$. We therefore begin with this case.

If $N=1$, we have $W_i\overset{k}\leftrightarrow W_j \label{eqn:intuitive1}, W_j \in \mathcal{W}_r, W_i \notin \mathcal{A}_r~~ \text{for}~~r\in \mathcal{K} \label{eqn:intuitive2}$, then we intend to show that
$R_i+R_j+R_{k_{1:L}}\leq1.$

Consider any reliable index coding scheme. Now, we form a genie-enhanced index coding problem (where the scheme continues to remain reliable) as follows. Let $W_i \in \mathcal{W}_{m}$. In this enhanced problem, assume that destinations $k,r,m$ are given all the messages except messages $W_{i,j,k_1,k_2,\ldots,k_L}$ by a genie. We also assume that $W_{k_1,k_2,\ldots,k_L}$ are provided via the genie to destinations $r,m$. Since $W_j\in\mathcal{W}_r$, for a reliable scheme, destination $r$ can decode $W_{j}.$ In this enhanced problem, either $m=r$ or destination $m$ is a degraded version of destination $r$. In both cases, destination $r$ can decode $W_{i}$  as well. 
Now consider destination $k$. Achievability implies that destination $k$ can decode $W_{k_1,k_2,\ldots,k_L}$. Now, destination $k$ has all the information available at destination $r$ and can therefore deocode $W_{i},W_{j}$ as well. Since all messages $W_{i,j,k_1,k_2,\ldots,k_L}$ are decoded at a single destination, we can follow steps similar to the proof of Theorem \ref{thm:simplebound}, to the conclusion that $R_i+R_j+R_{k_{1:L}}\leq1$. 

As an example, suppose we have an index coding problem with $M=4, K=3$ and $L=2$ where $\mathcal{W}_1=\{W_1,W_3\},\mathcal{W}_2=\{W_2,W_3\}, \mathcal{W}_3=\{W_3,W_4\}, \mathcal{A}_1=\{W_2,W_4\}, \mathcal{A}_2=\O, \mathcal{A}_3={W_2}$ as shown in Fig. \ref{fig:432case}. We have the following alignment chain $W_1\overset{2}\leftrightarrow W_{4}$, $W_4\in \mathcal{W}_3, W_1 \notin \mathcal{A}_3$. Using the above outer bound, we have $R_1+R_2+R_3+R_4\leq1$,  which is tight because it is also achievable by sending each message separately in 4 consecutive time slots.

\begin{figure}[!h] \centering
\includegraphics[width=3in]{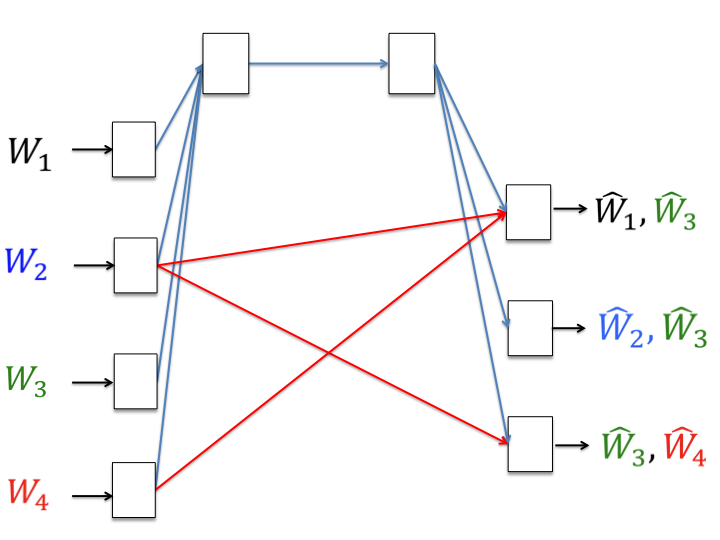}
\caption{Index coding problem with $M=4, K=3$ and $L=2$ where $R_1+R_2+R_3+R_4\leq1$ }
\label{fig:432case}
\end{figure}

Now, we consider the case where $N > 1$. Without loss of generality, we assume that $\mathcal{S}=\{0,1\}$ in the proof. We first present the proof explicitly for $N=2$ and $N=3$. These cases do not afford a simple explanation as above, and capture all the ideas required for proving the theorem for an arbitrary value of $N$. We begin with $N=2.$ In what follows we show that for any 3 distinct indices $i, j, k \in \mathcal{M}$ and indices $l, m \in \mathcal{K}$ if $W_i\overset{l}\leftrightarrow W_j \overset{m} \leftrightarrow W_k$, $W_{k} \in \mathcal{W}_p, W_{i} \notin \mathcal{A}_p$ for $p \in \mathcal{K}$ then $R_i+R_j+R_k+R_{l_{1:L}}+R_{m_{1:L}}\leq2$. 

\begin{eqnarray}
H(W_{l_1,l_2,\ldots,l_L}) = nR_{l_{1:L}}&=&I(W_{l_1,l_2,\ldots,l_L};S^n,\mathcal{W}_{\mathcal{A}_l})+o(n) \label{eqn:54}\\ 
&\leq& I(W_{l_1,l_2,\ldots,l_L};S^n,W_{i,j,l_1,l_2,\ldots,l_L}^c)+o(n)\\
&=&I(W_{l_1,l_2,\ldots,l_L};S^n \mid W_{{i,j,l_1,l_2,\ldots,l_L}}^c)+o(n)\\
&=&H(S^n\mid W_{{i,j,l_1,l_2,\ldots,l_L}}^c)\nonumber\\
&&-H(S^n\mid W_{i,j}^c)+o(n)\\
&\leq& n-H(S^n\mid W_{i,j}^c)+o(n) \label{eqn:58}
\end{eqnarray}
Similarly,
\begin{eqnarray}
H(W_{m_1,m_2,\ldots,m_L}) = nR_{m_{1:L}} &\leq& n-H(S^n\mid W_{j,k}^c)+o(n) \label{eqn:63}
\end{eqnarray}

\noindent For destination $p$ that is interested in $W_k$, we have
\begin{eqnarray}
nR_k&=&I(W_k;S^n,\mathcal{A}_p)+o(n)\label{eqn:64}\\
&\leq& I(W_k;S^n,W_{i,k}^c)+o(n)\\
&=&I(W_k;S^n \mid W_{i,k}^c)+o(n)\\
&\leq& H(S^n\mid W_{i,k}^c)-H(S^n\mid W_{i}^c)+o(n)\\
&\leq& H(S^n\mid W_{i,k}^c)-nR_i+o(n) \label{eqn:68}\\
&\leq& H(S^n\mid W_{i,j}^c)+H(S^n\mid {W}_{j,k}^c)\nonumber\\
&&-H(S^n \mid W_{j}^c)-nR_i+o(n)\label{eqn:followslemma2}\\
&=&H(S^n\mid W_{i,j}^c)+H(S^n\mid W_{j,k}^c)\nonumber\\
&&-nR_j-nR_i+o(n)\label{eqn:70}\\
&\leq& n(1-R_{l_{1:L}})+n(1-R_{m_{1:L}})-nR_i-nR_j+o(n) \label{eqn:71}\\
\Rightarrow R_k&\leq& 2-R_i-R_j-R_{l_{1:L}}-R_{m_{1:L}}, \label{eqn:72}
\end{eqnarray}
where (\ref{eqn:followslemma2}) follows from Lemma \ref{lem:useful_lemma} (proved later in this section). Inequality (\ref{eqn:71}) follows from substituting from (\ref{eqn:58}) and (\ref{eqn:63}) into (\ref{eqn:70}), and (\ref{eqn:72}) is obtained by dividing by $n$ and taking the limit as $n\rightarrow \infty$. 

As an example, suppose we have an index coding problem with $M=5, K=5$ and $L=2$ where $\mathcal{W}_1=\{W_1,W_5\},\mathcal{W}_2=\{W_1,W_2\}, \mathcal{W}_3=\{W_2,W_5\},  \mathcal{W}_4=\{W_2,W_4\}, \mathcal{W}_5=\{W_2,W_3\}, \mathcal{A}_1=\{W_2\}, \mathcal{A}_2=\{W_3\}, \mathcal{A}_3=\{W_1,W_4\}, \mathcal{A}_4=\{W_1,W_3,W_5\}, \mathcal{A}_5=\{W_1,W_4,W_5\}$ as shown in Fig. \ref{fig:552case}. We have the following alignment chain $W_3\overset{1}\leftrightarrow W_{4}\overset{2}\leftrightarrow W_{5}$, $W_5\in \mathcal{W}_3, W_3 \notin \mathcal{A}_3$. Using the above outerbound, we have $R_3+R_1+R_5+R_4+R_1+R_2+R_5\leq2$.

\begin{figure}[!h] \centering
\includegraphics[width=3in]{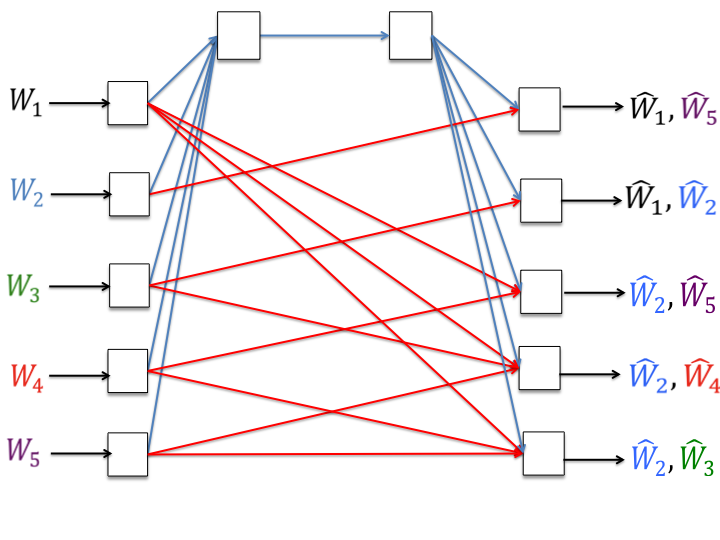}
\caption{\small Index coding problem with $M=5, K=5$ and $L=2$ where $R_3+R_1+R_5+R_4+R_1+R_2+R_5\leq2$}
\label{fig:552case}
\end{figure}

We next consider the case where the length of the alignment chain is $N=3.$ Suppose we have an alignment chain of length 3, i.e.,  $W_{i}\overset{p}\leftrightarrow W_{j}\overset{q}\leftrightarrow W_{k}\overset{r}\leftrightarrow W_{l}$ and the ends of the chain interfere, i.e., $W_{l} \in \mathcal{W}_m, W_{i}\notin \mathcal{A}_m$, then our goal is to show that $R_i+R_{p_{1:L}}+R_j+R_{q_{1:L}}+R_k+R_{r_{1:L}}+R_l\leq3$.\\
We start with the inequalities at the destinations $p, q, r$ , each of which is shown exactly through the steps followed in (\ref{eqn:54})-(\ref{eqn:58}). 
\begin{eqnarray}
H(W_{p_1,p_2,\ldots,p_L}) = nR_{p_{1:L}}\leq n-H(S^n\mid {W}_{i,j}^c)+o(n)\\
H(W_{q_1,q_2,\ldots,q_L}) = nR_{q_{1:L}}\leq n-H(S^n\mid {W}_{j,k}^c)+o(n)\\
H(W_{r_1,r_2,\ldots,r_L}) = nR_{r_{1:L}}\leq n-H(S^n\mid {W}_{k,l}^c)+o(n)
\end{eqnarray}
For destination $m$ that is interested in $W_{l}$, we have
\begin{eqnarray}
nR_l&=&I(W_l;S^n,\mathcal{A}_m)+o(n)\\
&\leq& I(W_l;S^n,{W}_{i,l}^c)+o(n)\\
&=&I(W_l;S^n \mid {W}_{i,l}^c)+o(n)\\
&=&H(S^n\mid {W}_{i,l}^c)-H(S^n\mid {W}_{i}^c)+o(n)\\
&\leq& H(S^n\mid {W}_{i,l}^c)-nR_i+o(n) \label{eqn:90}\\
&\leq& H(S^n\mid {W}_{i,j}^c)+H(S^n\mid {W}_{j,l}^c)\nonumber\\
&&-H(S^n \mid {W}_{j}^c)-nR_i+o(n) \label{eqn:91}\\
&\leq& n(1-R_{p_{1:L}})+H(S^n\mid {W}_{j,l}^c)-nR_j-nR_i+o(n) \label{eqn:92}\\
&\leq& n(1-R_{p_{1:L}})+H(S^n\mid {W}_{j,k}^c)\nonumber\\
&&+H(S^n\mid {W}_{k,l}^c)-H(S^n \mid {W}_{k}^c)\nonumber\\
&&-nR_j-nR_i+o(n) \label{eqn:93}\\
&\leq& n(1-R_{p_{1:L}})+n(1-R_{q_{1:L}})+n(1-R_{r_{1:L}})-nR_k-nR_j\nonumber\\
&&-nR_i+o(n)\label{eqn:94}\\
\Rightarrow R_l&\leq &3-R_{p_{1:L}}-R_{q_{1:L}}-R_{r_{1:L}}-R_k-R_j-R_i
\end{eqnarray}

Note that Lemma \ref{lem:useful_lemma} is applied twice, first in arriving at (\ref{eqn:91}) and then to obtain (\ref{eqn:93}). 

Now we prove Theorem \ref{thm:outerbound1} for an arbitrary $N>2$. Suppose we have an alignment chain of length $N$, $W_{i_0}\overset{j^1}\leftrightarrow W_{i_1}\overset{j^2}\leftrightarrow W_{i_2}\overset{j^3}\leftrightarrow...\overset{j^{N-1}}\leftrightarrow W_{i_{N-1}}\overset{j^N}\leftrightarrow W_{i_N}~~\text{and}~~W_{i_0} \notin \mathcal{A}_k~~\text{where}~~W_{i_N} \in \mathcal{W}_k~~\text{for}~~k \in \mathcal{K}$, then our goal is to show that $R_{i_0}+R_{j^1_{1:L}}+R_{i_1}+...+R_{j^N_{1:L}}+R_{i_N}\leq N$. We start with the inequalities at the destinations $j^1, \ldots, j^N$ , each of which is shown exactly through the steps followed in (\ref{eqn:54})-(\ref{eqn:58}). 
\begin{eqnarray}
H(W_{j^1_1,j^1_2,\ldots,j^1_L}) &=& nR_{j^1_{1:L}}\leq n-H(S^n\mid {W}_{i_0,i_1}^c)+o(n)\\
H(W_{j^2_1,j^2_2,\ldots,j^2_L}) &=& nR_{j^2_{1:L}}\leq n-H(S^n\mid {W}_{i_1,i_2}^c)+o(n)\\
&\vdots&\\
H(W_{j^N_1,j^N_2,\ldots,j^N_L}) &=& nR_{j^N_{1:L}}\leq n-H(S^n\mid {W}_{i_{N-1},i_N}^c)+o(n)
\end{eqnarray}
For destination $k$ that is interested in $W_{i_N}$, we have
\begin{eqnarray}
nR_{i_N}&=&I(W_{i_N};S^n,\mathcal{A}_k)+o(n)\\
&\leq& I(W_{i_N};S^n,{W}_{i_0,i_N}^c)+o(n)\\
&=&I(W_{i_N};S^n \mid {W}_{i_0,i_N}^c)+o(n)\\
&=&H(S^n\mid {W}_{i_0,i_N}^c)-H(S^n\mid {W}_{i_0}^c)+o(n)\\
&\leq& H(S^n\mid {W}_{i_0,i_N}^c)-nR_{i_0}+o(n) \\
&\leq& H(S^n\mid {W}_{i_0,i_1}^c)+H(S^n\mid {W}_{i_1,i_N}^c)\nonumber\\
&&-H(S^n \mid {W}_{i_1}^c)-nR_{i_0}+o(n) \label{eqn:uselemma1}\\
&\leq& n(1-R_{j^1_{1:L}})+H(S^n\mid {W}_{i_1,i_N}^c)-nR_{i_1}-nR_{i_0}+o(n) \\
&\leq& n(1-R_{j^1_{1:L}})+H(S^n\mid {W}_{i_1,i_2}^c)+H(S^n\mid {W}_{i_2,i_N}^c)-H(S^n \mid {W}_{i_2}^c)\nonumber\\
&&-nR_{i_1}-nR_{i_0}+o(n) \label{eqn:uselemma2}\\
&\leq& n(1-R_{j^1_{1:L}})+n(1-R_{j^2_{1:L}})+H(S^n\mid {W}_{i_2,i_N}^c)-nR_{i_2}-nR_{i_1}\nonumber\\
&&-nR_{i_0}+o(n)\\
&&\vdots \nonumber\\
&\leq& n(1-R_{j^1_{1:L}})+n(1-R_{j^2_{1:L}})+\ldots+n(1-R_{j^{N-2}_{1:L}})+H(S^n\mid {W}_{i_{N-2},i_{N-1}}^c)+H(S^n\mid {W}_{i_{N-1},i_{N}}^c)\nonumber\\
&&-H(S^n\mid {W}_{i_{N-1}}^c)-nR_{i_{N-2}}+\ldots-nR_{i_{1}}-nR_{i_0} \label{eqn:uselemma3}\\
&\leq& n(1-R_{j^1_{1:L}})+n(1-R_{j^2_{1:L}})+\ldots+n(1-R_{j^{N-2}_{1:L}})+n(1-R_{j^{N-1}_{1:L}})+n(1-R_{j^{N}_{1:L}})\nonumber\\
&&-nR_{i_{N-1}}-nR_{i_{N-2}}+\ldots-nR_{i_{1}}-nR_{i_0}\\
&\Rightarrow& R_{i_N}\leq N-R_{j^1_{1:L}}-R_{j^2_{1:L}}-\ldots-R_{j^N_{1:L}}-R_{i_{N-1}}-R_{i_{N-2}}-\ldots-R_{i_0}
\end{eqnarray}
Note that Lemma \ref{lem:useful_lemma} is applied in (\ref{eqn:uselemma1}),  (\ref{eqn:uselemma2}) and (\ref{eqn:uselemma3}). This proves Theorem \ref{thm:outerbound1}.

Finally, we prove Lemma \ref{lem:useful_lemma} which is used repeatedly in the proof of outer bounds.
\begin{lemma}
\label{lem:useful_lemma}
\begin{eqnarray}
H(S^n\mid {W}_{j,k}^c)\leq H(S^n\mid {W}_{i,j}^c)
+H(S^n\mid {W}_{i,k}^c)-H(S^n\mid {W}_{i}^c)
\end{eqnarray}
\end{lemma}
\begin{proof}
Since entropy function is a submodular function, for any two subsets of random variables $\mathcal{C}, \mathcal{D}$, we have the following
\begin{eqnarray}
H(\mathcal{C})+H(\mathcal{D})\geq H(\mathcal{C}\cup \mathcal{D})+H(\mathcal{C}\cap \mathcal{D} )
\end{eqnarray}
If we choose subsets $\mathcal{C}, \mathcal{D}$ as $\mathcal{C}=\{S^n, {W}_{i,j}^c\}$, $\mathcal{D}=\{S^n, {W}_{i,k}^c\}$ and use the submodular property of entropy function, we have,

\begin{eqnarray}
H(S^n, {W}_{i,j}^c)+H(S^n, {W}_{i,k}^c)\geq H(S^n, {W}_{i}^c)+H(S^n, {W}_{i,j,k}^c)
\end{eqnarray}
which is equivalent to
\begin{eqnarray}
H(S^n\mid {W}_{i,j}^c)+H(S^n\mid {W}_{i,k}^c)&\geq& H(S^n\mid {W}_{i}^c)+H(S^n\mid {W}_{i,j,k}^c)\nonumber\\
&\geq& H(S^n\mid {W}_{i}^c)+H(S^n\mid {W}_{j,k}^c) \label{eqn:conditioning}
\end{eqnarray}
where (\ref{eqn:conditioning}) is true because conditioning reduces entropy. This completes the proof

\end{proof}

\section{Proof of Theorem \ref{thm:sym1}}\label{sec:sym1proof}
\subsection{Achievability}
First, note that if $A= U+D = K-1$, then it is obvious that each source can send 1 symbol per time slot and achieve a rate of $1$. If $A=K-2$, then, it is easy to verify that a rate of $\frac{1}{2}$ is achievable using Theorem \ref{thm:HRF}. Here we show achievability for $A \leq K-3.$ Using linear coding, we show that each user can send $L=(U+1)$ symbols in a $n=K-A+2U$ dimensional space $\mathbb{F}_{q}^{n},$ when $q$ is sufficiently large. Note that because $U\leq D$, we have $A-2U=D-U\geq 0$, and therefore, $n\leq K$. To begin the construction of our achievable scheme, we pick $K$ vectors over $\mathbb{F}_{q}^{n}$ such that any $n$ of them are linearly independent. We denote the vectors by $\mathbf{z}_1, \mathbf{z}_{2},\ldots, \mathbf{z}_{K}.$ Note that over a sufficiently large field, the vectors $\mathbf{z}_{i}, i=1,2,\ldots,K$ can be chosen to satisfy this property\footnote{One approach to such a construction is to use Vandermonde matrices.}. Our construction for the $n\times (U+1)$ matrix $\mathbf{V}_{i}$ --- whose columns are the beamforming vectors for the $U+1$ symbols comprising message $W_i$ --- is as follows.
$$ \mathbf{V}_{i} = \left[\mathbf{z}_{i}~~ \mathbf{z}_{i + 1}~~\mathbf{z}_{i + 2} ~~ \ldots \mathbf{z}_{i + U}\right]$$
where all subscripts are interpreted modulo $K$. Note that any two adjacent messages overlap in $U$ dimensions. For example, the signal spaces spanned by ${\bf V}_i$ and ${\bf V}_{i+1}$ overlap in dimensions $\mathbf{z}_{i+1}, \mathbf{z}_{i+2}, \cdots, \mathbf{z}_{i+U}$. This is the basis for interference alignment. With this construction, we intend to show that (\ref{eq:resolvability}) is satisfied. 

First, note that the $U+1$ columns of $\mathbf{V}_i$ are linearly independent since $U+1 < n$. In particular, denoting
$\mathbf{X}_{i} = \left[x_{i,0}~~ x_{i,1}~~ \ldots, x_{i,U} \right]$, we have
$$ S^n = \sum_{i=1}^{K} \sum_{j=0}^{U} x_{i,j}  \mathbf{z}_{i+ j}.$$
Now, because of the symmetric nature of the problem, we only need to show that $W_1$ is linearly resolvable at $D_1.$ This ensures resolvability at all other destination nodes. Note that destination $1$ has side information of $x_{i,j}$ for all $i \in \{K-U+1, K-U+2, \ldots, K, 2, 3, \ldots, D+1\}.$ The precoding vectors are known to everyone apriori. Cancelling the effect of the known symbols, destination 1 obtains
$$ S^n_{1} = \underbrace{\sum_{j=0}^{U}x_{1,j} \mathbf{z}_{j+1}}_{\mbox{Desired Signal}}+ \underbrace{\sum_{i=D+2}^{K-U} {\sum_{j=0}^{U} {x}_{i,j}} \mathbf{z}_{i+j}}_{\begin{array}{c}\mbox{Aligned}\\ \mbox{Interference}\end{array}} $$

Thus, destination 1 sees $(U+1)+(K-(D+1))=n$ signal vectors $\mathbf{z}_{1},\mathbf{z}_{2},\ldots, \mathbf{z}_{U+1}$,$\mathbf{z}_{D+2}$, $\mathbf{z}_{D+3}$, $\ldots$, $\mathbf{z}_{K}.$ Since these $n$ vectors are linearly independent by design,  the desired scalars $x_{1,j}, j=0,1,\ldots,U$ are resolvable at destination $1$. By symmetry, the same conclusion is applicable at all the $K$ destinations. 

To understand the role of alignment, note that there are $(K-A-1)$ undesired messages that interfere at destination 1, each occupying a $(U+1)$ dimensional signal space. However, because any two adjacent interferers overlap in $U$ dimensions, the $(K-A-1)$ interferers collectively occupy only $U+1+(K-A-1)-1=K-D-1$ dimensions. Since the desired signal  occupies $U+1$ dimensions, and $(U+1) + (K-D-1) = n$, the observed signal space is big enough to resolve desired signals from the interference. See Figure \ref{fig:outerboundexample} for an example.
\begin{figure}[!h] \centering
\includegraphics[width=4in]{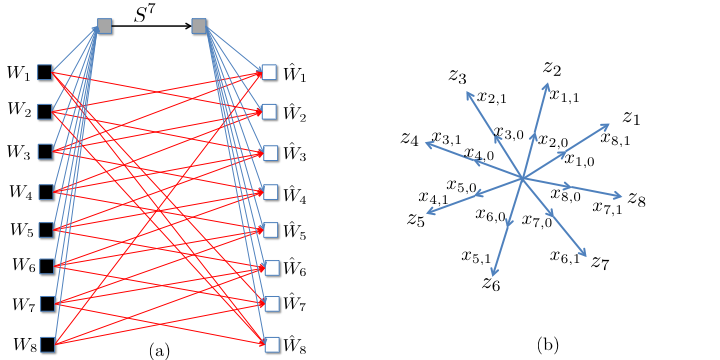}
\caption{\small Multiple unicast with neighboring antidotes where $(D,U,K)=(2,1,8)$. (a) Desired message and antidote
sets, (b) Capacity optimal solution -- 2 symbols per message are sent
with one-to-one pairwise alignments and a rate of $\frac{2}{7}$ per message is achieved using the scheme pictorially depicted. $z_1,\ldots,z_8$ may be chosen to be the $7$ columns of the $7\times 7$ identity matrix and the all 1's vector.}
\label{fig:outerboundexample}
\end{figure}

\subsection{Outer bound}
The trivial bound, that a message can not be transmitted at more than rate 1 (the capacity of the bottleneck link)  is tight when $A=K-1$, i.e., all undesired messages are available as antidotes. If $A=K-2$, i.e., only one message is missing from the antidote set, then the rate $\frac{1}{2}$ per user is the outerbound, as shown in the proof of Theorem \ref{thm:simplebound}. Therefore we focus on the setting $A\leq K-3$.

We first present an outer bound on the rates achievable via linear achievable schemes to gain intuition. Later, we will use this intuition to get information theoretic outer bounds for \emph{any} achievable scheme.
\subsubsection{Dimension Counting Outerbound}
Our goal is to show that the sum-capacity per message is bounded above by $\frac{U+1}{K-A+2U}$. To show this, we first prove that if we want to achieve a rate $d$ per message, then the total dimension of interference at each destination satisfies
\begin{eqnarray}
\dim(\text{interference})\geq \left[\frac{K-A+2U}{U+1}-1\right]d=\frac{K-A-1+U}{U+1}d
\end{eqnarray}
$K-A-1$ is the total number of interferers at each destination. Suppose $\mathcal{V}_i$ denotes the subspace assigned to user $i$. Hence, our goal is bounding $\sum_{i=1}^K\dim(\mathcal{V}_i \cup \mathcal{V}_{i+1} \cup \ldots \cup \mathcal{V}_{i+K-A-2})$, i.e., the total number of interfering dimensions at all the destinations. We define $\alpha_j$ as follows
\begin{eqnarray}
\alpha_j\overset{\triangle}=\sum_{i=1}^K {\dim(\mathcal{V}_i\cup \mathcal{V}_{i+1}\cup \ldots \cup \mathcal{V}_{i+j-1})}
\end{eqnarray}
So our goal is bounding the total number of interference at all destinations, i.e., $\alpha_{K-A-1}$. 
\begin{lemma}\label{lemma1}
For the following $m$ and $j$ 
\begin{eqnarray}
\left\{ 
  \begin{array}{l l}
    m=\lfloor\frac{K-A-1}{U+1}\rfloor,~ j=(K-A-1) \mod (U+1)& \quad \text{if}~K-A-1\mod (U+1)\neq 0\\
     m=\frac{K-A-1}{U+1}-1,~j=U+1 & \quad \text{if}~K-A-1\mod (U+1)=0\\
  \end{array} \right.
\end{eqnarray}
 we have 
\begin{eqnarray}
\alpha_{K-A-1} \geq m\alpha_{1}+\alpha_{j}
\end{eqnarray}
\end{lemma}
\begin{proof}
\begin{eqnarray}
\alpha_{K-A-1}&=&\alpha_{m(U+1)+j} \nonumber\\
&=&\sum_{i=1}^K\dim(\mathcal{V}_i\cup \mathcal{V}_{i+1}\cup \ldots \cup \mathcal{V}_{m(U+1)+j-1})\\
&\geq&\sum_{i=1}^K\{\dim(\mathcal{V}_i\cup \mathcal{V}_{i+1}\cup \ldots \cup  \mathcal{V}_{i+j-1} \cup \mathcal{V}_{i+j-1+(U+1)} \cup \mathcal{V}_{i+j-1+2(U+1)} \cup \ldots\cup \mathcal{V}_{i+j-1+m(U+1)})\}\\
&=&\sum_{i=1}^K\{\dim(\mathcal{V}_i\cup \mathcal{V}_{i+1}\cup \ldots \cup \mathcal{V}_{i+j-1})+\dim(\mathcal{V}_{i+j-1+(U+1)})+\dim(\mathcal{V}_{i+j-1+2(U+1)})+\ldots\nonumber\\
&&+\dim(\mathcal{V}_{i+j-1+m(U+1)})\} \label{eqn:disjointsubspaces}\\
&=&m\alpha_{1}+\alpha_{j}
\end{eqnarray}
where (\ref{eqn:disjointsubspaces}) is true becasue $W_{i,\ldots,i+j-1,i+j-1+(U+1),\ldots,i+j-1+l(U+1)} \notin \mathcal{A}_{i+j-1+l(U+1)}$ for $l=0,1,\ldots,m$. This proves Lemma \ref{lemma1}.
\end{proof}

\begin{lemma} \label{lemma2}
For $j=2,3,\ldots,U+1$, we have 
\begin{eqnarray}
\alpha_{j} &\geq& \alpha_{j-1}+\frac{\alpha_1}{U+1}\label{eqn:recursion1}\\
\Rightarrow ~~~\alpha_{j} &\geq& \frac{U+j}{U+1}\alpha_{1}\label{eqn:recursion2}
\end{eqnarray}
\end{lemma}
\begin{proof}
\begin{equation}
\sum_{i=1}^K\dim(\mathcal{V}_i \cup \mathcal{V}_{i+1} \cup\ldots \cup \mathcal{V}_{i+j-2} \cup \mathcal{V}_{i+j+U-1}) ~~~~~~~~~~~~~~~~~~~~~~~~~~~~~~~~~~~~~~~~~~~~~~~~~~~~~~~~~~~~~~~~~~~~~~~~~~~~~~~~~~~
\end{equation}
\begin{eqnarray}
&\leq&\sum_{i=1}^K\dim(\mathcal{V}_i \cup \mathcal{V}_{i+1} \cup\ldots \cup \mathcal{V}_{i+j-2} \cup \mathcal{V}_{i+j-1})+\sum_{i=1}^K\dim( \mathcal{V}_{i+1} \cup \mathcal{V}_{i+2} \cup\ldots \cup \mathcal{V}_{i+j-1} \cup  \mathcal{V}_{i+j+U-1})\nonumber\\
&&-\sum_{i=1}^K\dim(\mathcal{V}_{i+1} \cup \mathcal{V}_{i+2} \cup\ldots \cup \mathcal{V}_{i+j-1})\\
&\leq&\sum_{i=1}^K\dim(\mathcal{V}_i \cup \mathcal{V}_{i+1} \cup\ldots \cup \mathcal{V}_{i+j-2} \cup \mathcal{V}_{i+j-1})+\sum_{i=1}^K\dim( \mathcal{V}_{i+1} \cup \mathcal{V}_{i+2} \cup\ldots \cup \mathcal{V}_{i+j-1} \cup  \mathcal{V}_{i+j})\nonumber\\
&&+\sum_{i=1}^K\dim( \mathcal{V}_{i+2} \cup \mathcal{V}_{i+3} \cup\ldots \cup \mathcal{V}_{i+j} \cup  \mathcal{V}_{i+j+U-1})-\sum_{i=1}^K\dim(\mathcal{V}_{i+2} \cup \mathcal{V}_{i+3} \cup\ldots \cup \mathcal{V}_{i+j})\nonumber\\
&&-\sum_{i=1}^K\dim(\mathcal{V}_{i+1} \cup \mathcal{V}_{i+2} \cup\ldots \cup \mathcal{V}_{i+j-1})\\
&\vdots& \nonumber\\
&\leq&\sum_{i=1}^K\dim(\mathcal{V}_i \cup \mathcal{V}_{i+1} \cup\ldots \cup \mathcal{V}_{i+j-2} \cup \mathcal{V}_{i+j-1})+\sum_{i=1}^K\dim( \mathcal{V}_{i+1} \cup \mathcal{V}_{i+2} \cup\ldots \cup \mathcal{V}_{i+j-1} \cup  \mathcal{V}_{i+j})\nonumber\\
&&+\sum_{i=1}^K\dim( \mathcal{V}_{i+2} \cup \mathcal{V}_{i+3} \cup\ldots \cup \mathcal{V}_{i+j} \cup  \mathcal{V}_{i+j+1})+\ldots+\sum_{i=1}^K\dim( \mathcal{V}_{i+U} \cup \mathcal{V}_{i+U+1} \cup\ldots \cup \mathcal{V}_{i+j+U-2} \cup  \mathcal{V}_{i+j+U-1})\nonumber\\
&&-\sum_{i=1}^K\dim(\mathcal{V}_{i+U} \cup \mathcal{V}_{i+U+1} \cup\ldots \cup \mathcal{V}_{i+j+U-2})-\ldots-\sum_{i=1}^K\dim(\mathcal{V}_{i+2} \cup \mathcal{V}_{i+3} \cup\ldots \cup \mathcal{V}_{i+j})\nonumber\\
&&-\sum_{i=1}^K\dim(\mathcal{V}_{i+1} \cup \mathcal{V}_{i+2} \cup\ldots \cup \mathcal{V}_{i+j-1})\label{eqn:submodularproperty1}
\end{eqnarray}
where at each step, we use the submodular property of $\dim$ function. On the other hand, since $\{W_{i},W_{i+1},\ldots,W_{i+j-2}\} \notin \mathcal{A}_{i+j+U-1}$, we have
\begin{eqnarray}
\dim(\mathcal{V}_i \cup \mathcal{V}_{i+1} \cup\ldots \cup \mathcal{V}_{i+j-2} \cup \mathcal{V}_{i+j+U-1})=\dim(\mathcal{V}_i \cup \mathcal{V}_{i+1} \cup\ldots \cup \mathcal{V}_{i+j-2} )+\dim(\mathcal{V}_{i+j+U-1}) \label{eqn:disjointsub1}
\end{eqnarray}
By combining (\ref{eqn:submodularproperty1}) and (\ref{eqn:disjointsub1}), we have
\begin{eqnarray}
\alpha_{j-1}+\alpha_{1}\leq\alpha_j+\alpha_j+\ldots+\alpha_j-\alpha_{j-1}-\ldots-\alpha_{j-1}\\
\alpha_j\geq \frac{1}{U+1}\{(U+1)\alpha_{j-1}+\alpha_1\}
\end{eqnarray}
\end{proof}
This proves (\ref{eqn:recursion1}). Solving this recursive equation, we get (\ref{eqn:recursion2}). This proves Lemma \ref{lemma2}.

Combining Lemma \ref{lemma1} and Lemma \ref{lemma2}, we have the following
\begin{eqnarray}
\alpha_{K-A-1} \geq m\alpha_{1}+\frac{U+j}{U+1}\alpha_{1}=\frac{m(U+1)+j+U}{U+1}\alpha_{1}=\frac{K-A-1+U}{U+1}\alpha_{1}
\end{eqnarray}
Then the dimension of desired message plus the dimension of interference at all destinations should be less than or equal to the total available dimensions, i.e., $K$
\begin{eqnarray}
\sum_{i=1}^K\{\dim(\mathcal{V}_{i-D-1})+\dim(\mathcal{V}_i\cup \mathcal{V}_{i+1}\cup \ldots \cup \mathcal{V}_{i+K-A-2})\} \leq \alpha_1+\frac{K-A-1+U}{U+1}\alpha_{1} \leq K
\end{eqnarray}
Hence the optimal symmetric rate per message is $d=\frac{1}{K}\alpha_1\leq\frac{U+1}{K-A+2U}$
\subsubsection{Information Theoretic Outerbound}
We define $\alpha_j$ as follows
\begin{eqnarray}
\alpha_j\overset{\triangle}=\sum_{i=1}^K{H(S^n|W_{i,i+1,\ldots,i+j-1}^c)}
\end{eqnarray} 
Our first goal is to bound $\alpha_{K-A-1}$. We proceed as follows.
\begin{lemma}\label{lemma3}
For the following $m$ and $j$
\begin{eqnarray}
\left\{ 
  \begin{array}{l l}
    m=\lfloor\frac{K-A-1}{U+1}\rfloor,~ j=(K-A-1) \mod (U+1)& \quad \text{if}~K-A-1\mod (U+1)\neq 0\\
     m=\frac{K-A-1}{U+1}-1,~j=U+1 & \quad \text{if}~K-A-1\mod (U+1)=0\\
  \end{array} \right.
\end{eqnarray}
 we have 
\begin{eqnarray}
\alpha_{K-A-1} \geq m\alpha_{1}+\alpha_{j}+o(n) 
\end{eqnarray}
\end{lemma}
\begin{proof}
Note that destination $D_{i+j-1+(U+1)}$ can decode $W_{i+j-1+(U+1)}$ from $(S^n,W_{i,i+1,\ldots,i+j-1,i+j-1+(U+1)}^c )$ with $P_e \rightarrow 0$ as $n\rightarrow \infty$. Therefore, we can write
\begin{eqnarray}
nR_{i+j-1+(U+1)}&=&H(W_{i+j-1+(U+1)})\\
&=&I(W_{i+j-1+(U+1)}; S^n| W_{i,i+1,\ldots,i+j-1,i+j-1+(U+1)}^c ) \nonumber\\
&&+H(W_{i+j-1+(U+1)}|S^n,W_{i,i+1,\ldots,i+j-1,i+j-1+(U+1)}^c )\\
&=&H(S^n|W_{i,i+1,\ldots,i+j-1,i+j-1+(U+1)}^c ) \nonumber\\
&&-H(S^n| W_{i,i+1,\ldots,i+j-1}^c )+o(n)
\end{eqnarray}
which gives us 
\begin{eqnarray}
H(S^n|W_{i,i+1,\ldots,i+j-1,i+j-1+(U+1)}^c)&=&nR_{i+j-1+(U+1)}+H(S^n| W_{i,i+1,\ldots,i+j-1}^c)+o(n)\label{eq:W_ic}
\end{eqnarray}
\noindent Next, note that destination $D_{i+j-1+2(U+1)}$ does not have $W_{i,\ldots,i+j-1,i+j-1+(U+1),i+j-1+2(U+1)}$ as antidote. So, given $(S^n, W_{i,\ldots,i+j-1,i+j-1+(U+1),i+j-1+2(U+1)}^c)$ it must be able to reliably decode $W_{i+j-1+2(U+1)}$. 
\begin{eqnarray}
nR_{i+j-1+2(U+1)}&=&H(W_{i+j-1+2(U+1)})\\
&=&I(W_{i+j-1+2(U+1)}; S^n, W_{i,\ldots,i+j-1,i+j-1+(U+1),i+j-1+2(U+1)}^c)\nonumber\\
&&+H(W_{i+j-1+2(U+1)}|S^n,W_{i,\ldots,i+j-1,i+j-1+(U+1),i+j-1+2(U+1)}^c)\\
&=&I(W_{i+j-1+2(U+1)}; S^n| W_{i,\ldots,i+j-1,i+j-1+(U+1),i+j-1+2(U+1)}^c)+o(n)\nonumber\\
&=&H(S^n|W_{i,\ldots,i+j-1,i+j-1+(U+1),i+j-1+2(U+1)}^c)\nonumber\\
&&-H(S^n|W_{i,\ldots,i+j-1,i+j-1+(U+1)}^c)+o(n)
\end{eqnarray}
which along with (\ref{eq:W_ic}) gives us
\begin{eqnarray}
H(S^n|W_{i,\ldots,i+j-1,i+j-1+(U+1),i+j-1+2(U+1)}^c)&=&nR_{i+j-1+2(U+1)}+nR_{i+j-1+(U+1)} \nonumber\\
&&+H(S^n| W_{i,i+1,\ldots,i+j-1}^c)+o(n)\label{eq:W_i2c}
\end{eqnarray}
Similarly, we note that destination $D_{i+j-1+l(U+1)}$ for $3\leq l \leq m$ does not have $W_{i,\ldots,i+j-1,i+j-1+(U+1),\ldots,i+j-1+l(U+1)}$ as antidotes, so it must be able to decode $W_{i+j-1+l(U+1)}$ from $(S^n, W_{i,\ldots,i+j-1,i+j-1+(U+1),\ldots,i+j-1+l(U+1)}^c)$. 
\begin{eqnarray}
nR_{i+j-1+m(U+1)}&=&H(W_{i+j-1+m(U+1)})\\
&=&I(W_{i+j-1+m(U+1)}; S^n, W_{i,\ldots,i+j-1,i+j-1+(U+1),\ldots,i+j-1+m(U+1)}^c)+\nonumber\\
&&H(W_{i+j-1+m(U+1)}|S^n,W_{i,\ldots,i+j-1,i+j-1+(U+1),\ldots,i+j-1+m(U+1)}^c)\\
&=&I(W_{i+j-1+m(U+1)}; S^n| W_{i,\ldots,i+j-1,i+j-1+(U+1),\ldots,i+j-1+m(U+1)}^c)+o(n)\\
&=&H(S^n|W_{i,\ldots,i+j-1,i+j-1+(U+1),\ldots,i+j-1+m(U+1)}^c)-\nonumber\\
&&H(S^n|W_{i,\ldots,i+j-1,i+j-1+(U+1),\ldots,i+j-1+(m-1)(U+1)}^c)+o(n)
\end{eqnarray}
which gives us
\begin{eqnarray}
H(S^n|W_{i,\ldots,i+j-1,i+j-1+(U+1),\ldots,i+j-1+m(U+1)}^c)&=&nR_{i+j-1+m(U+1)}+\ldots+nR_{i+j-1+2(U+1)}\nonumber\\
&&+nR_{i+j-1+(U+1)}+H(S^n| W_{i,i+1,\ldots,i+j-1}^c)+o(n)\nonumber\\
\label{eq:W_i3c}
\end{eqnarray}
Our goal is to bound $\alpha_{K-A-1}$. We have
\begin{eqnarray}
\alpha_{K-A-1}&=&\alpha_{m(U+1)+j}\\
&=&\sum_{i=1}^KH(S^n|W_{i,i+1,\ldots,i+K-A-2}^c) \\
&\geq&\sum_{i=1}^K H(S^n|W_{i,\ldots,i+j-1,i+j-1+(U+1),i+j-1+2(U+1),\ldots,i+j-1+m(U+1)}^c)\label{eqn:condreduce}\\
&=&\sum_{i=1}^K \{nR_{i+j-1+(U+1)}+\ldots+nR_{i+j-1+m(U+1)}+H(S^n| W_{i,i+1,\ldots,i+j-1}^c)+o(n)\} \label{eqn:replaceW_i3c}\\
&=&m\alpha_1+\alpha_j+o(n)
\end{eqnarray}
where (\ref{eqn:condreduce}) is true because conditioning reduces the entropy. (\ref{eqn:replaceW_i3c}) is derived by replacing (\ref{eq:W_i3c}) into (\ref{eqn:condreduce}). This proves Lemma \ref{lemma3}.
\end{proof}
\begin{lemma} \label{lemma4}
For $j=2,3,\ldots,U+1$, we have 
\begin{eqnarray}
\alpha_{j} &\geq& \alpha_{j-1}+\frac{\alpha_1}{U+1}+o(n)\label{eqn:recursion3}\\ 
\Rightarrow~~~\alpha_{j} &\geq& \frac{U+j}{U+1}\alpha_{1}+o(n)\label{eqn:recursion4}
\end{eqnarray}
\end{lemma}
\begin{proof} 
We use Lemma \ref{lemma5} (proved later in this section) at each step. So we have
\begin{eqnarray}
\sum_{i=1}^K H(S^n,W_{i,i+1,\ldots,i+j-2,i+j+U-1}^c)&\leq& \sum_{i=1}^K \{H(S^n,W_{i,i+1,\ldots,i+j-1}^c)\nonumber\\
&&+H(S^n,W_{i+1,i+2,\ldots,i+j-1,i+j+U-1}^c)\nonumber\\
&&-H(S^n,W_{i+1,i+2,\ldots,i+j-1}^c)\}\\
&\leq&\sum_{i=1}^K \{H(S^n,W_{i,i+1,\ldots,i+j-1}^c)\nonumber\\
&&+H(S^n,W_{i+1,i+2,\ldots,i+j-1,i+j}^c) \nonumber\\
&&+H(S^n,W_{i+2,i+3,\ldots,i+j-1,i+j,i+j+U-1}^c)\nonumber\\
&&-H(S^n,W_{i+2,i+3,\ldots,i+j}^c)  \nonumber\\
&&-H(S^n,W_{i+1,i+2,\ldots,i+j-1}^c)\}\\
&\vdots&\nonumber\\
&\leq&\sum_{i=1}^K \{H(S^n,W_{i,i+1,\ldots,i+j-1}^c)\nonumber\\
&&+H(S^n,W_{i+1,i+2,\ldots,i+j-1,i+j}^c) \nonumber\\
&&+\ldots+H(S^n,W_{i+U,i+U+1,\ldots,i+j+U-1}^c)\nonumber\\
&&-H(S^n,W_{i+U,i+U+1,\ldots,i+j+U-2}^c) \nonumber\\
&&-\ldots-H(S^n,W_{i+2,i+3,\ldots,i+j}^c)-\nonumber\\
&&H(S^n,W_{i+1,i+2,\ldots,i+j-1}^c)\} \label{eqn:replaceinto}
\end{eqnarray}
We note that destination $D_{i+j+U-1}$ does not have $W_{i,i+1,\ldots,i+j-2,i+j+U-1}$ as antidotes, so it must be able to decode $W_{i+j+U-1}$ from $(S^n, W_{i,i+1,\ldots,i+j-2,i+j+U-1}^c)$. 
\begin{eqnarray}
nR_{i+j+U-1}&=&H(W_{i+j+U-1})\\
&=&I(W_{i+j+U-1}; S^n, W_{i,i+1,\ldots,i+j-2,i+j+U-1}^c)+\nonumber\\
&&H(W_{i+j+U-1}|S^n,W_{i,i+1,\ldots,i+j-2,i+j+U-1}^c)\\
&=&I(W_{i}; S^n| W_{i,i+1,\ldots,i+j-2,i+j+U-1}^c)+o(n)\\
&=&H(S^n|W_{i,i+1,\ldots,i+j-2,i+j+U-1}^c)-H(S^n|W_{i,i+1,\ldots,i+j-2}^c)+o(n)
\end{eqnarray}
which gives us
\begin{eqnarray}
H(S^n|W_{i,i+1,\ldots,i+j-2,i+j+U-1}^c)=nR_{i+j+U-1}+H(S^n|W_{i,i+1,\ldots,i+j-2}^c)+o(n) \label{eqn:replace}
\end{eqnarray}
Replacing (\ref{eqn:replace}) into (\ref{eqn:replaceinto}), we have
\begin{eqnarray}
\alpha_{j-1}+\alpha_{1}+o(n)\leq\alpha_j+\alpha_j+\ldots+\alpha_j-\alpha_{j-1}-\ldots-\alpha_{j-1}\\
\alpha_j\geq \frac{1}{U+1}\{(U+1)\alpha_{j-1}+\alpha_1\}+o(n)
\end{eqnarray}
Solving this recursive equation, we get
\begin{eqnarray}
\alpha_{j} \geq \frac{U+j}{U+1}\alpha_{1}+o(n)
\end{eqnarray}
This proves Lemma \ref{lemma4}.

Combining Lemma \ref{lemma3} and Lemma \ref{lemma4}, we have the following
\begin{eqnarray}
\alpha_{K-A-1} \geq m\alpha_{1}+\frac{U+j}{U+1}\alpha_{1}+o(n)=\frac{m(U+1)+j+U}{U+1}\alpha_{1}+o(n)=\frac{K-A-1+U}{U+1}\alpha_{1}+o(n)
\end{eqnarray}
Finally, we note that destination $D_{i-D-1}$ does not have any of the messages $W_{i,i+1,\ldots,i+K-A-2}$ as antidotes. So it must be able to decode $W_{i-D-1}$ from $(S^n,W_{i-D-1,i,i+1,\ldots,i+K-A-2})$.

\begin{eqnarray}
\sum_{i=1}^{K} nR_{i-D-1}&\leq& \sum_{i=1}^{K}I\left(W_{i-D-1};S^n|W_{i-D-1,i,i+1,\ldots,i+K-A-2}^c\right)+o(n)\\
&=& \sum_{i=1}^{K}\left\{H\left(S^n|W_{i-D-1,i,i+1,\ldots,i+K-A-2}^c\right)-
  H(S^n|W_{i,i+1,\ldots,i+K-A-2}^c)\right\}+o(n)\nonumber\\
  &\leq&
 \sum_{i=1}^{K}\{ n-\frac{K-A-1+U}{U+1}R_i\}+o(n)
\end{eqnarray}
Rearranging terms and applying the limit $n\rightarrow\infty$ we have
\begin{eqnarray}
\sum_{i=1}^{K}R_i&\leq&\frac{U+1}{K-A+2U}K
\end{eqnarray}
so that we have the information theoretic capacity outer bound of $\frac{U+1}{K-A+2U}$ per message.

\end{proof}
\begin{lemma}\label{lemma5}For $l=0,1,\ldots,U-1$, we have
\begin{eqnarray}
H(S^n|W_{i+l,i+l+1,\ldots,i+l+j-2,i+j+U-1}^c) &\leq& H(S^n|W_{i+l,i+l+1,\ldots,i+l+j-1}^c)\nonumber\\
&&+H(S^n|W_{i+l+1,i+l+2,\ldots,i+l+j-1,i+j+U-1}^c)\nonumber\\
&&-H(S^n|W_{i+l+1,i+l+2,\ldots,i+l+j-1}^c)
\end{eqnarray}
\end{lemma}
\begin{proof}
\noindent
Setting  $$\mathcal{C}=\{S^n,W_{i+l,i+l+1,\ldots,i+l+j-1}^c\},$$
$$\mathcal{D}=\{S^n,W_{i+l+1,i+l+2,\ldots,i+l+j-1,i+j+U-1}^c\}$$ and using the submodular property of entropy function
\begin{eqnarray}
H(\mathcal{C})+H(\mathcal{D})\geq H(\mathcal{C}\cup \mathcal{D})+H(\mathcal{C}\cap \mathcal{D} )
\end{eqnarray}
 we have
\begin{eqnarray*}
H(S^n,W_{i+l,\ldots,i+l+j-1}^c)+H(S^n,W_{i+l+1,\ldots,i+l+j-1,i+j+U-1}^c)&\geq& H(S^n,W_{i+l+1,i+l+2,\ldots,i+l+j-1}^c)\nonumber\\
&&+H(S^n,W_{i+l,i+l+2,\ldots,i+l+j-1,i+j+U-1}^c)
\end{eqnarray*}
Equivalently,
\begin{eqnarray}
H(S^n|W_{i+l,\ldots,i+l+j-1}^c)+H(S^n|W_{i+l+1,\ldots,i+l+j-1,i+j+U-1}^c)&\geq& H(S^n|W_{i+l+1,i+l+2,\ldots,i+l+j-1}^c)\nonumber\\
&&+H(S^n|W_{i+l,i+l+2,\ldots,i+l+j-1,i+j+U-1}^c)\nonumber\\
&\geq& H(S^n|W_{i+l+1,i+l+2,\ldots,i+l+j-1}^c)\nonumber\\
&&+H(S^n|W_{i+l,i+l+2,\ldots,i+l+j-2,i+j+U-1}^c) \nonumber\\ \label{eqn:condreduce2}
\end{eqnarray}
where (\ref{eqn:condreduce2}) is true because conditioning reduces the entropy. Equivalently,
\begin{eqnarray}
H(S^n|W_{i+l,i+l+1,\ldots,i+l+j-2,i+j+U-1}^c) &\leq& H(S^n|W_{i+l,i+l+1,\ldots,i+l+j-1}^c)\nonumber\\
&&+H(S^n|W_{i+l+1,i+l+2,\ldots,i+l+j-1,i+j+U-1}^c)\nonumber\\
&&-H(S^n|W_{i+l+1,i+l+2,\ldots,i+l+j-1}^c)
\end{eqnarray}
This proves Lemma \ref{lemma5}. 
\end{proof}

\section{Proof of Theorem \ref{thm:sym2}: Neighboring interference}\label{sec:sym2proof}
\subsection{Achievability}
We use a scalar linear achievable scheme. In particular, we choose $\mathcal{S}=\mathbb{F}_q^{D+1}$. Suppose we define $\mathbf{V}_1,\mathbf{V}_2,\ldots,\mathbf{V}_{D+1}$ as $D+1$ linearly independent vectors in a $D+1$ dimensional space. Then, starting from any arbitrarily chosen message, we assign these vectors to $D+1$ consecutive messages respectively and repeat this assignment periodically. So if we assign $\mathbf{V}_m$ to message $W_r$, then message $W_i$ will be encoded using $\mathbf{V}_{(i-r+m) \mod (D+1)}$. We now need to show the resolvability of message $W_r$ at destination $r$. Suppose $\mathbf{V}_m$ is the vector assigned to message $W_r$. Notice that  there are $U+D$ interfering messages at this destination: $W_{r-U, r-U+1,\ldots, r-1, r+1, r+2,\ldots, r+D}.$ The remaining messages are available as antidotes and can be cancelled. The $U$ interfering messages $W_{r-U},\ldots,W_{r-1}$ are encoded over $\mathbf{V}_{(m-U)\mod (D+1)}$,$\ldots,\mathbf{V}_{(m-1)\mod (D+1)}$, respectively. Since $\{i+D+1 \mod D+1\} = \{i \mod D+1\}$, these interfering vectors are the same as $\mathbf{V}_{(m+D+1-U)\mod (D+1)}$,$\ldots,\mathbf{V}_{(m+D)\mod (D+1)}$.  Also the remaining $D$ interfering messages of destination $r,$ i.e., $W_{r+1},\ldots,W_{r+D}$ are sent over $\mathbf{V}_{(m+1)\mod (D+1)},\ldots,\mathbf{V}_{(m+D)\mod (D+1)}$, respectively. As we can see all the interfering messages align in the $D$-dimensional space spanned by vectors $\mathbf{V}_{(m+1)\mod (D+1)},\ldots,\mathbf{V}_{(m+D)\mod (D+1)}.$ This interference space is linearly independent of $\mathbf{V}_m$ because of our construction. Since we are operating in a $D+1$ dimensional space, $W_{r}$ is linearly resolvable at destination $r$ as required. This proves achievability.

\emph{Remark:} $\mathbf{V}_1,\mathbf{V}_2,\ldots,\mathbf{V}_{D+1}$ can be chosen to be columns of $D+1$ dimensional identity matrix and the field can be chosen to be $\mathbb{F}_2$ . In other words, a $D+1$-symbol long achievable scheme is obtained by simply sending the bit $\ldots \oplus x_{r-D-1}\oplus x_{r} \oplus x_{r+D+1} \oplus \ldots, $ in the $r$th channel use, where $\oplus$ denotes the XOR. 

\subsection{Outerbound}
For the converse, note that we can set $U=0$. This is because, if $U > 0$,  a genie can provide messages $W_{r-U, r-U+1,\ldots,r-1}$ to destination $r$ as antidotes to make an enhanced index coding problem where each user is only missing $D+1$ antidotes after the desired message, where we will show that the capacity is $\frac{1}{D+1}$ per message.  In this setting where $U=0,$ we intend to show that 
$$R_{i}+R_{i+1}+\ldots+R_{i+D} \leq 1.$$ 
To do so, we give destinations $i, i+1, \ldots, i+D$ messages $W_{i,i+1,\ldots,i+D}^{c}$ through a genie. Now, in this genie-aided system, note that destination $i,$ which is missing antidotes $W_{i+1,i+2,\ldots,i+D}$ can decode $W_{i}$ (because of achievability in the original index coding problem). Having decoded $W_{i},$ this destination has all the messages present at destination $i+1,$ and can therefore decode $W_{i+1}.$ Having decoded $W_{i,i+1}$, destination $i$ is now equipped with all messages present as antidotes at destination ${i+2}$ and can therefore decode $W_{i+2}.$ Continuing this argument, it can be shown that destination $i$ can decode $W_{i,i+1,\ldots,i+D},$ and therefore we have, as required,
$$R_{i}+R_{i+1}+\ldots+R_{i+D} \leq 1.$$ 

\section{Proof of Theorem \ref{thm:sym3}: X network setting with local connectivity}\label{sec:Xproof}
\subsection{Achievability }
We use a scalar linear achievable scheme. In particular, we choose $\mathcal{S}=\mathbb{F}_{q}$ and $n={\frac{L(L+1)}{2}}$. We choose $\mathbf{V}_1,\mathbf{V}_2,\ldots,\mathbf{V}_{\frac{L(L+1)}{2}}$ to be $\frac{L(L+1)}{2}$ linearly independent vectors in a $n=\frac{L(L+1)}{2}$ dimensional space. Each message is sent over one of these vectors.
 
Consider an arbitrary destination, say destination $k$. After removing antidotes, this destination receives a linear combination of $L^2$ messages shown in (\ref{eqn:xrec}) above the line, and it is interested in $L$ messages 
$$\mathcal{W}_k=\{W_{kL,kL+L-1, (k+1)L+L-2,\cdots,(k+i)L+L-i-1,\cdots, (k+L-2)L+1}\}$$
 shown in blue in (\ref{eqn:xrec}). The destination faces $(L^2-L)$ interfering messages shown in (\ref{eqn:xrec}) with colors other than blue (and above the line). The ``blue'' messages are desired by destination $k$ and are  encoded using $L$ linearly independent vectors, $\mathbf{V}_1,\ldots,\mathbf{V}_L$. Because the total number of dimensions is $\frac{L(L+1)}{2}$, the $L^2-L$ interfering messages should align such that they span at most $\frac{L(L+1)}{2}-L$ dimensions. Among these interfering messages, the $L-1$ ``red'' messages $W_{(k+1)L, (k+1)L+L-1),\ldots,(k+L-2)+2}$ are desired by destination $k+1$ and should be linearly independent. These messages in red are encoded using $\mathbf{V}_{L+1},\ldots,\mathbf{V}_{2L-1}$.  Also all the $L-2$ ``green'' messages that are below the red messages and above the line, i.e., $W_{(k+2)L, (k+2)L+L-2,\ldots, (k+L-2)+3}$ are not available as antidotes at either destination $k$ or destination $k+1$ and therefore seen as interference at both these destinations. This implies that the red messages and blue messages can not align with these green messages. So we assign $\mathbf{V}_{2L},\ldots,\mathbf{V}_{3L-2}$ as encoding vectors respectively to green messages. Proceeding thus, we assign $V_{L(L+1)/2}$ to $W_{(k+L-1)L}$. So far, we have assigned $L(L+1)/2-L$ linearly independent vectors to $L(L+1)/2-L$ interferers at destination $k$.  

\begin{eqnarray}
\begin{array}{ccccc}\label{eqn:xrec}
W_{(k-1)L+1}&W_{(k-1)L+2}&\ldots&W_{(k-1)L+L-1}&\color{blue}W_{kL}\\
W_{kL+1}&W_{kL+2}&\ldots&\color{blue}W_{kL+L-1}& \color{red}W_{(k+1)L}\\
W_{(k+1)L+1}&W_{(k+1)L+2}&\ldots&\color{red}W_{(k+1)L+L-1}& \color{green}W_{(k+2)L}\\
\vdots&\vdots&\iddots&\vdots&\vdots\\
W_{(k+L-3)L+1}&\color{blue}W_{(k+L-3)L+2}&\ldots&\color{Orange}W_{(k+L-3)L+L-1}&\color{violet}W_{(k+L-2)L}\\
\color{blue}W_{(k+L-2)L+1}&\color{red}W_{(k+L-2)L+2}&\ldots&\color{violet}W_{(k+L-2)L+L-1}&\color{yellow}W_{(k+L-1)L}\\
\hline
W_{(k+L-1)L+1}&W_{(k+L-1)L+2}&\ldots&W_{(k+L-1)L+L-1}&W_{(k+L)L}\\
\end{array}
\end{eqnarray}

So the remaining interfering messages that are shown above the blue messages at (\ref{eqn:xrec}) should be sent over the vectors such that they stay in the same span of interfering messages below the blue messages. The way that we satisfy this constraint is by sending the messages shown in (\ref{eqn:xrec}) respectively over the following vectors

\begin{eqnarray}
\begin{array}{ccccc}\label{eqn:xvec}
\mathbf{V}_{L+1}&\mathbf{V}_{L+2}&\ldots&\mathbf{V}_{2L-1}&\mathbf{V}_1\\
 \mathbf{V}_{2L}&\mathbf{V}_{2L+1}&\ldots&\mathbf{V}_2& \mathbf{V}_{L+1}\\
 \mathbf{V}_{3L-1}&\mathbf{V}_{3L}&\ldots&\mathbf{V}_{L+2}& \mathbf{V}_{2L}\\
\vdots&\vdots&\iddots&\vdots&\vdots\\
\mathbf{V}_{\frac{L(L+1)}{2}}&\mathbf{V}_{L-1}&\ldots&\mathbf{V}_{\frac{L(L+1)}{2}-4}&\mathbf{V}_{\frac{L(L+1)}{2}-2}\\
\mathbf{V}_L&\mathbf{V}_{2L-1}&\ldots&\mathbf{V}_{\frac{L(L+1)}{2}-1}&\mathbf{V}_{\frac{L(L+1)}{2}}\\
\hline
\mathbf{V}_{1}&\mathbf{V}_{2}&\ldots&\mathbf{V}_{L-1}&\mathbf{V}_{L}
\end{array} \label{eqn:pattern}
\end{eqnarray}

Evidently, at destination $k$, all desired messages are seen over $\mathbf{V}_{1},\ldots,\mathbf{V}_{L}$ and all the interfering messages are seen over $\mathbf{V}_{L+1},\ldots,\mathbf{V}_{\frac{L(L+1)}{2}}$. Therefore, desired messages are resolvable at destination $k$. However our goal is showing that all the destinations are able to decode their desired messages. If we assign $\mathbf{V}_{1},\mathbf{V}_{2},\ldots,\mathbf{V}_{L-1},\mathbf{V}_{L}$ to $W_{(k+L-1)L+1},W_{(k+L-1)L+2},\ldots,W_{(k+L-1)L+L-1},W_{(k+L)L}$ respectively and repeat the pattern shown in (\ref{eqn:xvec}) as assigning vectors periodically for the remaining messages, we can show resolvability at every destination. This follows because, if we choose any $L$ consecutive rows (circularly) of (\ref{eqn:xvec}), the vectors assigned to anti-diagonal messages which are desired messages are linearly independent from each other and from the vectors assigned to interfering messages. This proves achievability.

\emph{Remark:} $\mathbf{V}_1,\mathbf{V}_2,\ldots,\mathbf{V}_{\frac{L(L+1)}{2}}$ can be chosen to be columns of $\frac{L(L+1)}{2}$ dimensional identity matrix, e.g., over $\mathbb{F}_2$ and therefore can be an orthogonal scheme.  For the corresponding CBIA problem this means that channel coherence is not required \cite{CBIA}.
\subsection{Outerbound}
To prove the outerbound, consider a set of $\frac{L(L+1)}{2}$ messages 
\begin{eqnarray}
 \mathcal{W_O}=\{W_{kL,kL+L-1:kL+L,\cdots,(k+i)L-i-1:(k+i)L+L,\cdots,(k+L-2)L+1:(k+L-1)L}\}, \nonumber
 \end{eqnarray}
i.e., the set of ``colored'' messages shown (above the line) in (\ref{eqn:xrec}).
  Among the messages in $\mathcal{W_O}$, the number of messages intended for destination $k+i$ is $L-i,$ where $i \in \{0,1,2,\ldots,L-1.\}$ Our goal is to argue that destination $k$ can decode all these $\frac{L(L+1)}{2}$ messages and hence symmetric rate per message is bounded as $C \leq \frac{2}{L(L+1)}$. 

Consider any reliable achievable index coding scheme.  Assume a genie provides all the messages except  $\mathcal{W_O}$ for all the destinations $k,k+1,\cdots,k+L-1$. With the considered index coding scheme, destination $k,$ which has $\mathcal{W_O}^{c}$ as antidotes, can decode the ``blue messages'' in (\ref{eqn:xrec}), i.e., $$\{W_{kL,kL+L-1, (k+1)L+L-2,\cdots,(k+i)L+L-i-1,\cdots, (k+L-2)L+1}\}.$$ destination $k+1$ can decode its desired $L-1$ ``red'' messages - $\{W_{kL+L, (k+1)L+L-1,\cdots,(k+i)L+L-i,\cdots, (k+L-2)L+2}\}$ - using $\mathcal{W_O}^{c}\cup \{W_{kL}\}$ as antidote. This automatically implies that destination $k,$ having decoded all the blue messages including $W_{kL}$ can decode the red messages as well. Now, having decoded all the blue messages and red messages, the set of messages known to destination $k$ includes $\mathcal{W_O}^c \cup W_{kL, kL+L-1, (k+1)L}$ - the antidote at destination $k+2$. Therefore all the green messages can be decoded at destination $k$. Proceeding further similarly we can argue that destination $k$ can decode all the messages in $\mathcal{W_O}$. This completes the proof.



\section{Conclusion}
As evident from this work, interference alignment is integral to the index coding problem.  The interference alignment perspective allows us to not only solve fairly complex index coding problems, but  also it makes the intuition behind the capacity optimal solutions quite transparent. As with wireless networks, while much of the initial intuition from interference alignment schemes is based on dimension counting based on linear codes, with few exceptions the dimension counting bounds are readily translated into tight information theoretic bounds. It is also remarkable that the interference alignment perspective allows us to prove the insufficiency of linear codes for multiple unicast index coding. As a side remark, we note that the capacity results of the index coding settings explored in this work, since they rely only on vector linear achievable schemes that readily translate into the field of complex numbers,  directly establish corresponding DoF results for the cellular blind interference alignment settings as well. We end this paper with a couple of intriguing questions. 

First, it is not clear that auxiliary messages and destinations, while convenient for our purpose, are \emph{necessary} in the equivalent multiple unicast problem. Since the purpose of auxiliary messages and destinations is only to force the expanded messages into alignment, the natural question is if auxiliary messages and destinations are not included, would it be possible to achieve a higher $\min$ rate for the expanded messages in the multiple unicast setting, presumably through a non-aligned solution? While this possibility seems unlikely, we do not yet have a proof that it is impossible. Such a proof would be desirable because it would make the equivalence between groupcast and unicast settings more direct. 

The second question pertains to the capacity of the $X$ channel setting with finite number of users. We showed in Section \ref{sec:X} that the capacity solution for the case where the number of users is infinity, is $\frac{2}{L(L+1)}$ per message, where $L$ is the number of messages per transmitter or per receiver. However, note that the finite user setting for the X channel, studied in Example 3 presented in Section \ref{sec:Xexample} is also consistent with this result. In that example, we have 3 messages per source/destination, i.e., $L=3$ and we achieve rate $\frac{2}{12}$ per message, even though the number of users is finite. A similar observation can be made for the $L=2$ setting considered in \cite{CBIA}. This suggests an interesting possibility -- is it always possible to achieve rate $\frac{2}{L(L+1)}$ per message with finite number of users $K$? Note that the outer bound applies to finite $K$ settings as well. On the other hand, even if the optimal rate per message is the same for finite $K$ as well as infinite $K$, evidently the alignment solution can be much more complex for finite $K$ settings. Note that a sophisticated subspace alignment solution is needed in Section \ref{sec:Xexample}, but much simpler orthogonal solutions suffice when $K$ is infinity, as shown in Section \ref{sec:Xproof}. 

Finally, we conclude with the observation that the index coding problem remains still an open problem of great interest, and we expect that the insights from interference alignment will continue to be useful not only in solving smaller networks or symmetric versions of extended networks as shown here, but perhaps also in other directions not explored in this work --- e.g., designing interference alignment inspired algorithms for arbitrary index coding settings, and studying order optimality of interference alignment techniques in index coding settings modeled as random graphs.

\bibliography{Thesis}
\end{document}